\documentclass[a4paper,11pt]{article}
\usepackage[top=1in, bottom=1in, left=1in, right=1in]{geometry}
\usepackage{amsthm}
\newtheorem{theorem}{Theorem}[section]
\newtheorem{lemma}[theorem]{Lemma}
\newtheorem{definition}[theorem]{Definition}
\newtheorem{corollary}[theorem]{Corollary}

\newtheorem{property}[theorem]{Property}

\usepackage{enumitem}
\usepackage{graphicx}
\usepackage{float}

\newcommand{\litem}[1]{\item[#1\hfill]}

\newenvironment{mylist}[1]{
    \setbox1=\hbox{#1}
    \begin{list}{}{
            \setlength{\labelwidth}{\wd1}
            \setlength{\leftmargin}{\wd1}
            \addtolength{\leftmargin}{0em}
            \addtolength{\leftmargin}{\labelsep}
            \setlength{\rightmargin}{1em}}}{\end{list}}

\usepackage[lined, algoruled, linesnumbered]{algorithm2e}
\usepackage{hyperref,xcolor}
\definecolor{winered}{rgb}{0.5,0,0}
\hypersetup
{
    pdfborder={0 0 0},
    colorlinks=true,
    linkcolor={winered},
    urlcolor={winered},
    filecolor={winered},
    citecolor={winered},
    linktoc=all,
}

\title{On maximal $k$-edge-connected subgraphs of undirected graphs}
\date{}
\author{Loukas Georgiadis\thanks{University of Ioannina, Greece.  E-mail: \texttt{loukas@uoi.gr}. Supported by the Hellenic Foundation for Research and Innovation (HFRI) under the ``First Call for HFRI Research Projects to support Faculty members and Researchers and the procurement of high-cost research grant'', Project FANTA (eFficient Algorithms for NeTwork Analysis), number HFRI-FM17-431.} \and Giuseppe F. Italiano\thanks{Luiss University, Rome, Italy. Email: \texttt{gitaliano@luiss.it.} Partially supported by MUR, the Italian Ministry for University and Research, under PRIN Project AHeAD (efficient Algorithms for HArnessing networked Data).} \and Evangelos Kosinas\thanks{University of Ioannina, Greece. Email: \texttt{ekosinas@cs.uoi.gr}. Supported by the Hellenic Foundation for Research and Innovation (HFRI) under the 3rd Call for HFRI PhD Fellowships (Fellowship Number: 6547).} \and Debasish Pattanayak\thanks{Luiss University, Rome, Italy. Email: \texttt{dpattanayak@luiss.it.} Supported by MUR, the Italian Ministry for University and Research, under PRIN Project AHeAD (efficient Algorithms for HArnessing networked Data).}} 
\begin{document}
\maketitle


\begin{abstract}
We show how to find and efficiently maintain maximal $k$-edge-connected subgraphs in undirected graphs. In particular, we provide the following results.
$(1)$
A general framework for maintaining the maximal $k$-edge-connected subgraphs upon insertions of edges or vertices, by successively partitioning the graph into its $k$-edge-connected components. This defines a decomposition tree, which can be maintained by using algorithms for the incremental maintenance of the $k$-edge-connected components as black boxes at every level of the tree.
$(2)$ As a concrete application of this framework, we provide two algorithms for the incremental maintenance of the maximal $3$-edge-connected subgraphs. These algorithms allow for vertex and edge insertions, interspersed with queries asking whether two vertices belong to the same maximal $3$-edge-connected subgraph, and there is a trade-off between their time- and space-complexity. Specifically, the first algorithm has $O(m\alpha(m,n) + n^2\log^2 n)$ total running time and uses $O(n)$ space, where $m$ is the number of edge insertions and queries, and $n$ is the total number of vertices inserted starting from an empty graph. The second algorithm performs the same operations in faster $O(m\alpha(m,n) + n^2\alpha(n,n))$ time in total, using $O(n^2)$ space.
$(3)$ 
We provide efficient constructions of (almost) sparse spanning subgraphs that have the same maximal $k$-edge-connected subgraphs as the original graph. These are useful in speeding up computations involving the maximal $k$-edge-connected subgraphs in dense undirected graphs.
$(4)$  We give two deterministic algorithms for computing the maximal $k$-edge-connected subgraphs in undirected graphs, with running times $O(m+k^{O(1)}n\sqrt{n}\,\mathrm{polylog}(n))$ and $O(m + k^{O(k)}n\sqrt{n} \log{n})$, respectively.
Hence, our first algorithm improves the dependence on $k$ for deterministic algorithms, while our second algorithm is faster than the randomized algorithm of Forster et al. (SODA 2020) for constant $k$.
$(5)$ A fully dynamic algorithm for maintaining information about the maximal $k$-edge-connected subgraphs for fixed $k$. Our update bounds are ${O}(n\sqrt{n}\,\log n)$ worst-case time for $k>4$ and ${O}(n\sqrt{n}\,)$ worst-case time for $k\in\{3,4\}$. In both cases, we achieve constant time for maximal $k$-edge-connected subgraph queries.
\end{abstract}

\newpage

\section{Introduction}
A dynamic graph algorithm aims to maintain the solution of a given problem after each update faster than recomputing it from scratch. An algorithm is \textit{fully dynamic} if it supports both insertions and deletions of edges, while it is  \textit{incremental} (resp.~\textit{decremental}) if it only supports insertions (resp.~deletions) of edges. 

Let $G=(V,E)$ be a connected undirected multigraph with $m$ edges and $n$ vertices.
We let $V(G)$ and $E(G)$ denote the vertex set and the edge set of $G$, respectively, i.e.,  $V(G)=V$ and $E(G)=E$.
Let $S \subseteq V$ be any subset of vertices of $G$. The \emph{subgraph induced by $S$}, denoted by $G[S]$, is defined as the graph having vertex set $S$ and edge set $\{e\in E\mid \mbox{both endpoints of } e \mbox{ lie in } S\}$.
Let $X,Y$ be two subsets of $V$. $[X,Y]$ denotes the set of edges that have one endpoint in $X$ and the other in $Y$. In other words, $[X,Y]=\{(x,y)\in E\mid x\in X \mbox{ and } y\in Y\}$.
Let $S \subset V$ be a non-empty (proper) subset of vertices of $G$, and let $\bar{S} = V\setminus S\neq\emptyset$. If both $G[S]$ and $G[\bar{S}]$ are connected, then $[S,\bar{S}]$ is called an \emph{edge cut} of $G$, i.e., a minimal set of edges of $G$ such that its removal disconnects the graph. An edge cut of cardinality $k$ is called a \emph{$k$-edge cut} of $G$. As a special case, $1$-edge cuts are also called \emph{bridges}. $G$ is said to be \emph{$k$-edge-connected} if it contains no $k'$-edge cuts, for $k'< k$.

Two vertices of $G$ are said to be $k$-edge-connected if there are $k$ edge-disjoint paths between them. 
%
The \emph{edge-connectivity} of a pair of vertices $x,y\in V$, denoted as $\lambda_{G}(x,y)$ (or simply $\lambda(x,y)$, when $G$ is clear from context), is the maximum $k$ such that $x$ and $y$ are $k$-edge-connected.
The relation of ``$k$-edge-connectivity'' between pairs of vertices of $G$ is an equivalence relation on $V(G)$ (see, e.g., \cite{AlgAspects}). The equivalence classes induced by this relation are called the \emph{$k$-edge-connected components} of $G$. Thus, a $k$-edge-connected component of $G$ is a maximal subset of vertices $C\subseteq V$ such that any pair of vertices in $C$ is $k$-edge-connected. Notice that every $k$-edge-connected component lies entirely within a $k'$-edge-connected component, for every $k'<k$.

Now let $S \subseteq V$ be a subset of vertices in $G$. We say that the induced subgraph $G[S]$ is a \emph{maximal $k$-edge-connected subgraph} of $G$ if (1) $G[S]$ is $k$-edge-connected and (2) no proper superset of $S$ has
this property. Unlike 2-edge connectivity, for $k\geq 3$ the $k$-edge-connected components of $G$ do not necessarily correspond to maximal $k$-edge-connected subgraphs. Indeed, for $k \geq 3$, two vertices in the subgraph induced by a $k$-edge-connected component may not be $k$-edge-connected in this subgraph, as some of the $k$ edge-disjoint paths may go outside of the component; see Figure~\ref{figure:max3ecsexample}. Also, Figure~\ref{figure:max3ecsexample-3} shows a graph where almost all its vertices are $3$-edge-connected but has only trivial maximal $3$-edge-connected subgraphs. 
%
Notice that, if $S$ is a subset of $V$, then $\lambda_{G[S]}(x,y)\leq\lambda_{G}(x,y)$ for any pair of vertices $x,y\in S$, since every path in $G[S]$ is also a path in $G$. Thus, every maximal $k$-edge-connected subgraph lies within the subgraph induced by a $k$-edge-connected component.

\begin{figure*}[t!]
\begin{center}
\centerline{\includegraphics[trim={0 0 0 2cm}, clip=true, width=\textwidth]{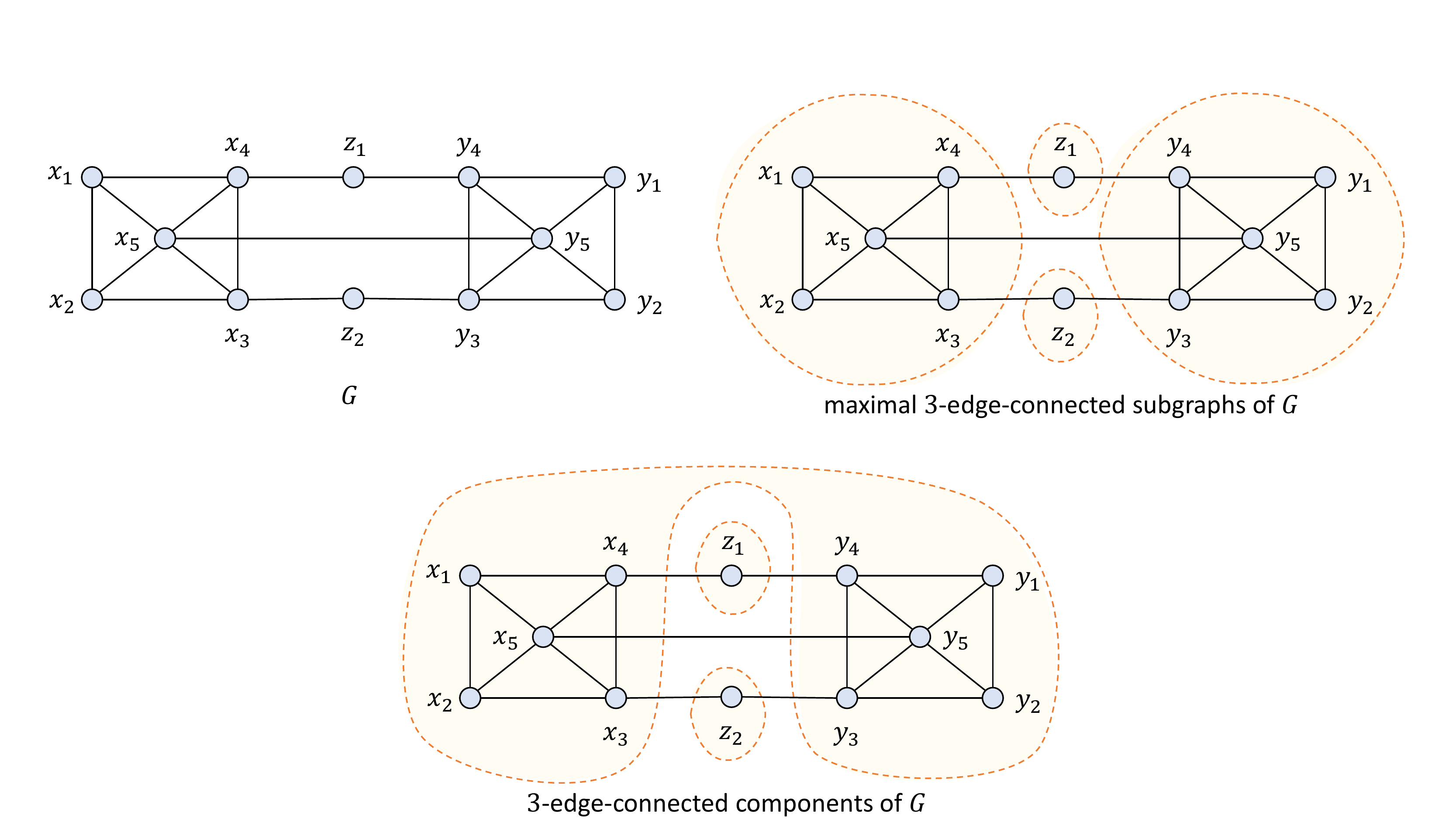}}
\caption{A $2$-edge-connected graph $G$, its maximal $3$-edge-connected subgraphs, and its $3$-edge-connected components. Note that while any two vertices $x_i$ and $y_j$ are $3$-edge-connected, they do not belong to the same maximal $3$-edge-connected subgraph. 
\label{figure:max3ecsexample}}
\end{center}
\end{figure*}

\begin{figure*}[t!]
\begin{center}
\centerline{\includegraphics[trim={0 0 0 11cm}, clip=true, width=\textwidth]{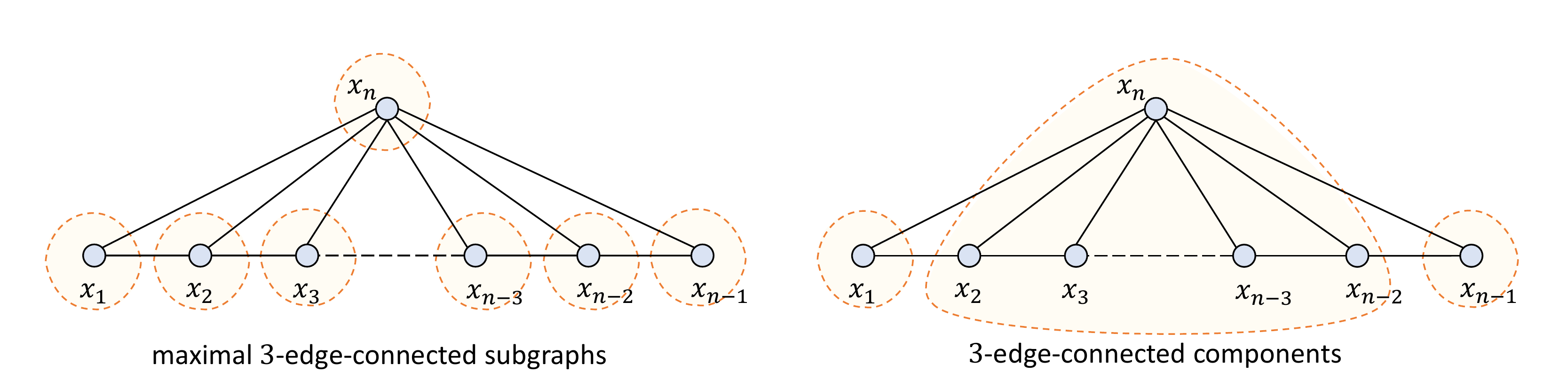}}
\caption{A $2$-edge-connected graph $G$ with only trivial maximal $3$-edge-connected subgraphs, despite that almost all vertices are $3$-edge-connected. 
\label{figure:max3ecsexample-3}}
\vspace{-1cm}
\end{center}
\end{figure*}

Determining or testing various notions of edge connectivity of undirected graphs, as well as computing edge-connected components or subgraphs, is a fundamental graph problem, and indeed $k$-edge connectivity has received much attention in the literature. The problem of computing maximal $k$-edge-connected subgraphs, however, appears to be harder than computing $k$-edge-connected components, as not much progress has been made for the former problem, despite its importance in the applications. As a matter of fact, 
finding the maximal $k$-{edge-connected subgraphs} of a graph is of significant interest in several  areas, such as in the field of databases, social networks, graph visualization etc.~\cite{akibaLineartimeEnumerationMaximal2013,borosGenerating3vertexConnected2008,changEfficientlyComputingKedge2013,liKvertexConnectedComponent2020,sunEfficientKedgeConnected2016,wenEnumeratingKVertexConnected2019,yuanEfficientECCGraph2016}
 (some of those papers refer to the maximal $k$-{edge-connected subgraphs} as $k$-edge-connected components).

In more detail, it is known how to compute the $k$-edge-connected components in linear time for $k \leq 4$~\cite{pathdfs:g00,GI:ECtoVC,DBLP:conf/esa/GeorgiadisIK21,3-connectivity:ht,DBLP:conf/esa/NadaraRSS21,NagamochiIbaraki:3CC,dfs:t,DBLP:journals/jda/Tsin09}.
On the other hand, linear-time bounds for computing the maximal $k$-edge-connected subgraphs are known only in the trivial case $k \le 2$, simply because in these cases the $k$-edge-connected components coincide with the maximal $k$-edge-connected subgraphs. 
As mentioned in~\cite{ChechikHILP17}, the maximal $3$-edge-connected subgraphs can be computed with a simple-minded algorithm in $O(mn)$ time in the worst case. 
Henzinger et al.~\cite{henzingerFinding2Edge2Vertex2015a} provided a deterministic algorithm  for computing the maximal $k$-edge-connected subgraphs, for constant $k$, 
that runs in $O(n^2\log n)$ time.
Chechik et al.~\cite{ChechikHILP17} presented  an algorithm with running time $O(k^{O(k)}(m+n\log n)\sqrt{n}\,)$ ($O(m\sqrt{n}\,)$ for $k=3$),
which improves the bound of~\cite{henzingerFinding2Edge2Vertex2015a} for sparse graphs. Forster et al.~\cite{forsterComputingTestingSmall2019} gave a Las Vegas randomized algorithm for computing the maximal $k$-edge-connected subgraphs in 
$O(km\log^2{n} + k^3 n \sqrt{n} \log{n})$
expected running time.
Very recently, Nalam and Saranurak~\cite{NalamSaran} provided a randomized algorithm for computing the maximal $k$-edge-connected subgraphs in \emph{weighted} undirected graphs in $\widetilde{O}(m\cdot\min\{m^{3/4},n^{4/5}\})$ time.
%
%
In the special case of planar graphs, Holm et al.~\cite{HolmIKLRS17} showed how to compute the maximal $3$-edge-connected subgraphs in $O(n)$ time.

\subsection{Overview of our results}
In this paper we present several new results on maximal $k$-edge-connected subgraphs of undirected graphs. 
In particular, we provide the following  results.

\begin{enumerate}[label={(\arabic*)}]
\item{A general framework for maintaining the maximal $k$-edge-connected subgraphs upon insertions of edges or vertices, by successively partitioning the graph into its $k$-edge-connected components. This defines a decomposition tree, which can be maintained by using algorithms for the incremental maintenance of the $k$-edge-connected components as black boxes at every level of the tree.}

\item{As an application of our new framework, we provide two algorithms for the incremental maintenance of the maximal $3$-edge-connected subgraphs. These algorithms allow for vertex and edge insertions, interspersed with queries asking whether two vertices belong to the same maximal $3$-edge-connected subgraph, and
provide a trade-off between time- and space-complexity. The first algorithm has $O(m\alpha(m,n) + n^2\log^2 n)$ total running time and uses asymptotically optimal $O(n)$ space, where $m$ is the number of edge insertions and queries, and $n$ is the total number of vertices inserted starting from an empty graph. The second algorithm improves the total running time to $O(m\alpha(m,n) + n^2\alpha(n,n))$ (i.e., almost optimal for dense graphs) at the expense of using $O(n^2)$ space.} We note that those are the first incremental algorithms for this problem, and thus provide significant improvements over recomputing the maximal $3$-edge-connected subgraphs after every insertion.

\item Using results from Benczúr and Karger~\cite{ben-kar}, we provide efficient constructions of (almost) sparse spanning subgraphs that have the same maximal $k$-edge-connected subgraphs as the original graph. These are useful in speeding up computations involving the maximal $k$-edge-connected subgraphs in dense undirected graphs. In particular, we use those certificates to speed up the computation of the maximal $k$-edge-connected subgraphs.

\item 
Two deterministic algorithms for computing the maximal $k$-edge-connected subgraphs in undirected graphs, with running times $O(m+k^{O(1)}n\sqrt{n}\,\mathrm{polylog}(n))$ and $O(m + k^{O(k)}n\sqrt{n} \log{n})$, respectively.
Hence, our first algorithm improves the dependence on $k$ for deterministic algorithms (which was previously $k^{O(k)}$~\cite{ChechikHILP17}), while our second algorithm is faster than the randomized algorithm of Forster et al.~\cite{forsterComputingTestingSmall2019} for constant $k$.

\item{A fully dynamic algorithm for maintaining information about the maximal $k$-edge-connected subgraphs for fixed $k$. Our update bounds are   ${O}(n\sqrt{n}\,\log n)$ worst-case time for $k>4$ and ${O}(n\sqrt{n}\,)$ worst-case time for $k\in\{3,4\}$. In both cases we achieve constant time for maximal $k$-edge-connected subgraph queries. Since no previous fully dynamic bound was known for this problem, once again this improves on the bounds obtainable by recomputing the maximal $k$-edge-connected subgraphs from scratch after each update.}
\end{enumerate}

We believe that our main technical contribution is given by $(2)$. To achieve this result, we provide a structural characterization of the maximal $3$-edge-connected subgraphs of an undirected graph by introducing a decomposition tree $\mathcal{T}$ into maximal $3$-edge-connected subgraphs, and show how to maintain it efficiently under edge and vertex insertions. $\mathcal{T}$ is a rooted tree whose root corresponds to the whole graph $G$ and is defined recursively by computing the (subgraphs induced by the) $k$-edge-connected components, for $k\in\{1,2,3\}$, of the graphs of the previous level. We proceed recursively in this decomposition until we reach a graph that is $3$-edge-connected, which is a maximal $3$-edge-connected subgraph of $G$ (and a leaf of $\mathcal{T}$). Therefore, the nodes of $\mathcal{T}$ correspond to subgraphs of $G$, and the parent relation is given by (vertex) set inclusion. (See Section~\ref{section:decomposition_tree}, and also Figure~\ref{figure:decompositiontreeexample}.)

Although maintaining the entire tree $\mathcal{T}$ may seem more challenging than maintaining only the maximal $3$-edge-connected subgraphs, we show that it is in fact easier to maintain the decomposition tree. This is because $\mathcal{T}$ contains enough information to facilitate its efficient update after new insertions to $G$, as its nodes corresponds to all successive partitions of $G$ into $k$-edge-connected components, for $k\in\{1,2,3\}$. In fact, if a new edge $e$ is inserted to $G$, then we only have to locate the deepest subgraph $N$ on $\mathcal{T}$ that contains the endpoints of $e$, and all changes to $\mathcal{T}$ due to this insertion apply to the subtree of $N$.
However, we do not explicitly maintain the correspondence between nodes of $\mathcal{T}$ and subgraphs of $G$, as this would require as much as $\Omega(mn)$ space (where $m$ is the number of edges and $n$ is the number of vertices of $G$), and a much worse time to maintain $\mathcal{T}$ during a sequence of insertions. Instead, we associate with the nodes of $\mathcal{T}$ some data structures that represent the interconnections between the $2$- and $3$-edge-connected components of the subgraphs that (abstractly) correspond to the nodes of $\mathcal{T}$. To be more precise, we associate to every node of $\mathcal{T}$ that corresponds to a connected component $C$ of its parent (a representation of) the tree of the $2$-edge-connected components of $C$; and to every node of $\mathcal{T}$ that corresponds to a $2$-edge-connected component of its parent, we associate (a representation of) the cactus of its $3$-edge-connected components. All this information reveals to be useful for maintaining our tree decomposition with the help of previous incremental approaches~\cite{galilMaintaining3EdgeConnectedComponents1993,lapoutreMaintainance23eccII,lapoutreMaintainance23ecc,westbrookMaintainingBridgeConnectedBiconnected1992}.
%
We describe in detail, in Sections \ref{section:structures} and \ref{section:improved_datastructures}, how we can augment the previous incremental approaches with enough information to suit our purposes.

Notice that in the bounds provided in $(2)$, there is a trade-off between space- and time-complexity. The second algorithm is more time-efficient (at least asymptotically), due to the following two reasons. First, we find a way to use the more sophisticated data structures of La Poutr{\'e} \cite{lapoutreMaintainance23eccII}, in order to efficiently maintain the $2$- and $3$-edge-connected components of the graphs in the various levels of $\mathcal{T}$. Secondly, we use an alternative and more intriguing method for answering nearest common ancestor and level ancestor queries, that are needed in the algorithm for maintaining $\mathcal{T}$. This method relies on the structure of those queries (i.e., they are not completely arbitrary but they depend on modifications already made on $\mathcal{T}$).
However, in order to facilitate the efficient answering of those queries, we pay an extra $O(n^2)$ in space. In any case, the $n^2$ factor in the time-bound of both algorithms is an inherent bottleneck of the basic procedure that we use, for any sequence of insertions, and a lower bound on the worst-case time for a single insertion (see Figure~\ref{figure:example-insertions} and Algorithm~\ref{algorithm:tree_modification}).

Both our algorithms can efficiently answer queries concerning the maximal $3$-edge-connected subgraphs in asymptotically optimal time, plus the time to perform one or two calls to a $\mathit{find}$ operation in an underlying disjoint set union (DSU) data structure \cite{tarjanEfficiencyGoodNot1975} that maintains the vertex sets of the maximal $3$-edge-connected subgraphs. (For the details on this DSU data structure, see Section~\ref{section:maintaining_T}.) For instance, given two query vertices $x$ and $y$, we can report whether $x$ and $y$ belong to the same maximal $3$-edge-connected subgraph using two calls to a $\mathit{find}$ operation, or given a query vertex $x$ we can report the maximal $3$-edge-connected subgraph that contains $x$ in time proportional to its size plus a call to a $\mathit{find}$ operation. (See Section~\ref{section:conclusion}.)

We note that these algorithms can be seen as applications of a more general framework for maintaining the maximal $k$-edge-connected subgraphs, by relying on algorithms that maintain the $k$-edge-connected components (result $(1)$). We think that it is possible that one can develop a similar framework for maintaining the maximal $k$-vertex-connected subgraphs, by relying on efficient algorithms for maintaining the $k$-vertex-connected components. (In particular, for the case $k=3$, one may rely on the algorithms of \cite{triconnBatTam,lapoutreTriconn} for maintaining the $3$-vertex-connected components.)

For $(3)$, we show that it is sufficient to compute (a superset of) all the edges whose endpoints lie in different maximal $k$-edge-connected subgraphs. Benczúr and Karger~\cite{ben-kar} provide efficient algorithms that achieve this.
From those algorithms we get two different constructions for spanning subgraphs that have the same maximal $k$-edge-connected subgraphs as the original graph, and there is a trade-off between the time complexity and the size of the output subgraph. Following the terminology of~\cite{opt-dec-conn}, we call such a subgraph a $k$-certificate. Then, we have a linear-time algorithm for computing a $k$-certificate with $O(kn\log{n})$ edges, and an $O(m\log^2{n})$-time algorithm for computing a $k$-certificate with $O(kn)$ edges (where $m$ and $n$ denote the number of edges and vertices of the graph, respectively). A key component in those algorithms is the certificates for $k$-edge connectivity of Nagamochi and Ibaraki~\cite{Nagamochi}.
For the details, see Section~\ref{section:sparse_subgraphs}.
We believe that it is an interesting question whether a $k$-certificate of $O(kn)$ size can be computed in linear time. This would be trivial if we had a linear-time algorithm for computing the maximal $k$-edge-connected subgraphs of a graph, but it is still an open problem whether this can be done for $k\geq 3$. Thus, we have to perform the construction of the certificates without explicitly computing the maximal $k$-edge-connected subgraphs, and this seems to be a challenging task.

For $(4)$ we provide two different deterministic algorithms by using two essentially different approaches.
The idea in the first approach is to repeatedly find and remove all $k'$-edge cuts, for $k'<k$. For this purpose we can rely on the fully dynamic min-cut algorithm of Thorup (Theorem $26$ in~\cite{thorupFullyDynamicMinCut2007}). Thus, we use this algorithm in order to successively find and remove all $k'$-edge cuts, for $k'<k$, until we are left with the maximal $k$-edge-connected subgraphs.
This is how we get a deterministic algorithm with 
$O(m+k^{O(1)}n\sqrt{n}\mathrm{polylog}(n))$ worst-case time for computing the maximal $k$-edge-connected subgraphs of a graph with $m$ edges and $n$ vertices. This algorithm is described in Section~\ref{section:computing_kecs}.
Alternatively, we may use our result $(3)$, in order to get in linear time a subgraph with $O(kn\log n)$ edges that has the same maximal $k$-edge-connected subgraphs as the original graph, and then we may apply the algorithm by Chechik et al.~\cite{ChechikHILP17}, that runs in 
$O(k^{O(k)}(m+n\log n)\sqrt{n}\,)$ time. Thus we get an $O(m+k^{O(k)}n\sqrt{n}\log n)$-time algorithm for computing the maximal $k$-edge-connected subgraphs.
We note that the previously best known deterministic algorithm for computing the maximal $k$-edge-connected subgraphs in dense graphs, for constant $k$, is given by Henzinger et al.~\cite{henzingerFinding2Edge2Vertex2015a}, and runs in $O(n^2\log n)$ time. For sparse graphs, the best time-bound is provided by Chechik et al.~\cite{ChechikHILP17}, which is $O((m+n\log n)\sqrt{n}\,)$ for $k>4$, and $O(m\sqrt{n}\,)$ for $k\in\{3,4\}$\footnote{The improved time-bound for $k\in\{3,4\}$ is due to the existence of linear-time algorithms for computing $2$- and $3$-edge cuts. (See, e.g. ~\cite{DBLP:conf/esa/GeorgiadisIK21, DBLP:conf/esa/NadaraRSS21,DBLP:journals/jda/Tsin09}.) For $k>4$ we rely on Gabow's algorithm~\cite{edge_connectivity:gabow}, which runs in almost linear time.}.
%
%
Hence, our first algorithm improves the dependence on $k$ for deterministic algorithms from exponential to polynomial, while our second algorithm is faster than the randomized algorithm of Forster et al.~\cite{forsterComputingTestingSmall2019} for constant $k$.

For $(5)$ we rely on the sparsification framework of Eppstein et al.~\cite{sparsification}. We show that we can apply the certificates that we provide in Section~\ref{section:sparse_subgraphs}, and we use the algorithm of Chechik et al.~\cite{ChechikHILP17} for the updates. For the details, see Section~\ref{section:fully-dynamic}.

\subsection{More related work}
Dynamic graph algorithms have been extensively studied for several decades, and many
important results have been achieved for fundamental problems such as dynamic connectivity \cite{franciosaIncrementalMaintenanceDepthFirstSearch1997,georgiadisDecrementalDataStructures2017,HK99,luigi2ConnectivityDirectedGraphs2015}, minimum spanning trees \cite{sparsification,fredericksonDataStructuresOnLine1985,nanongkaiDynamicMinimumSpanning2017}, minimum cut \cite{lapoutreDynamicGraphAlgorithms1991,lapoutreMaintainance23eccII,thorupFullyDynamicMinCut2007}, shortest paths \cite{abrahamFullyDynamicAllpairs2017,bernsteinDecrementalStronglyconnectedComponents2019,demetrescuNewApproachDynamic2004,thorupFullyDynamicAllPairsShortest2004}, and transitive closure \cite{henzingerFullyDynamicBiconnectivity1995}.
Regarding dynamic algorithms for maintaining $k$-connectivity information (for small values of $k$) has also received a lot of attention~\cite{Dinitz:5ECC,2LCactus,dinitzMaintainingClasses4EdgeConnectivity1998,sparsification,DynEC:GI,HLT01,st-conn_Sun,Kanevsky:4CC,lapoutreMaintainance23eccII,lapoutreMaintainance23ecc,westbrookMaintainingBridgeConnectedBiconnected1992}. This previous work, however, involves maintaining a graph under edge insertions and/or deletions, so that we can efficiently answer queries of whether two vertices are $k$-edge- or $k$-vertex-connected (i.e., belong to the same $k$-edge- or $k$-vertex-connected component). 
Very recently, Aamand et al.~\cite{opt-dec-conn} provided a Monte Carlo randomized algorithm for the decremental maintenance of the maximal $c$-edge-connected subgraphs, that has  $O(m + n^{3/2+o(1)})$ total update time and $O(1)$ query time, for $c=O(n^{o(1)})$.
To the best of our knowledge, no previous non-trivial bounds were known for the dynamic maintenance of maximal $3$-edge-connected subgraphs with a deterministic algorithm.

\subsection{Organization}
The rest of the paper is organized as follows.
First, we provide some preliminaries in Section~\ref{section:preliminaries}.
We describe our decomposition tree of the maximal $3$-edge-connected subgraphs in Section~\ref{section:decomposition_tree}. This can be seen as an application of a general framework for maintaining the maximal $k$-edge-connected subgraphs by using algorithms for maintaining the $k$-edge-connected components, which is described in Sections \ref{section:a_general_framework} and \ref{section:maintaining_T_g}. 
The details on maintaining the decomposition tree of the maximal $3$-edge-connected subgraphs under insertions of new edges and vertices is given in Section~\ref{section:maintaining_T}.
In Sections~\ref{section:structures} and~\ref{section:improved_datastructures} we present efficient implementations for the data structures that are associated with the nodes of the decomposition tree, in order to efficiently update it. 
In Section~\ref{section:sparse_subgraphs} we present a linear-time construction of  sparse certificates for  maximal $k$-edge-connected subgraphs.  From this we derive a 
deterministic $O(m + k^{O(k)}n\sqrt{n}\log n)$-time algorithm for computing the maximal $k$-edge-connected subgraphs.
In Section~\ref{section:computing_kecs} we provide an $O(m+k^{O(1)}n\sqrt{n}\mathrm{polylog}(n))$ deterministic algorithm for computing the maximal $k$-edge-connected subgraphs, using Thorup's fully dynamic mincut algorithm.
Section~\ref{section:fully-dynamic} presents our fully dynamic algorithm for maximal $k$-edge-connected subgraphs. 
We conclude in Section~\ref{section:conclusion} with suggestions for further applications of our decomposition tree.

\section{Preliminaries}\label{section:preliminaries}
In the sequel, we assume that the reader is familiar with standard graph terminologies, for example, as presented in \cite{CLRS01, Diestel, AlgAspects}. 
For the sake of completeness, we provide some definitions and state some results that will be used throughout concerning the structure of the $2$- and $3$-edge-connected components in undirected graphs.
Consider the quotient graph $Q^k$ of $G$ that is formed by shrinking every $k$-edge-connected component into a single vertex, maintaining all inter-connection edges and discarding self-loops. Then the quotient map $\nu: V(G)\rightarrow V(Q^k)$ induces a natural correspondence between the edges of $Q^k$ and some edges of $G$: that is, for every edge $(u,v)$ of $Q^k$, there is an edge $(x,y)\in E(G)$ such that $\nu(x)=u$ and $\nu(y)=v$. Now, for $k=2$, we have that $Q^2$ is a tree $T$. The nodes of $T$ correspond to the $2$-edge-connected components of $G$, and the edges of $T$ correspond to the bridges of $G$. Now let $C$ be a $2$-edge-connected component of $G$. Then we have that $G[C]$ is $2$-edge-connected, and every $k$-edge-connected component of $G$ that lies in $C$, for $k\geq2$, is also a $k$-edge-connected component of $G[C]$. For $k=3$, the quotient graph $Q^3$ of $G[C]$ is a \emph{cactus} $S$ (that is, a connected graph in which every edge belongs to a unique cycle) \cite{DBLP:conf/wg/Dinitz92,galilMaintaining3EdgeConnectedComponents1993,lapoutreMaintainance23ecc}. The nodes of $S$ correspond to the $3$-edge-connected components of $G[C]$ and the $2$-edge cuts of $Q^3$ correspond to the $2$-edge cuts of $G[C]$. (We note that these properties generalize to the \emph{cactus of the minimum cuts} of an undirected graph \cite{CactusMinCut}.)

Now let us consider how the insertion of a new edge $(x,y)$ to $G$ affects its $k$-edge-connected components, for $k\leq 3$. (In general, observe that the edge-connectivity for any pair of vertices increases at most by one.) We distinguish four different cases, depending on whether $\lambda(x,y)\in\{0,1,2\}$ or $\lambda(x,y)\geq 3$, prior to the insertion of $(x,y)$.
\begin{enumerate}[label={(\arabic*)}]
\item{If $\lambda(x,y)=0$, then $x$ and $y$ belong to two different connected components $C_1$ and $C_2$, respectively. Thus the only change that occurs is that $\lambda(u,v)=1$, for any two vertices $u\in C_1$ and $v\in C_2$.}

\item{If $\lambda(x,y)=1$, then $x$ and $y$ belong to the same connected component $C$, and lie in two different $2$-edge-connected components $X$ and $Y$, respectively. Let $T$ be the tree of the $2$-edge-connected components of $G[C]$, and let $P=X_1,\dots,X_k$ be the simple path on $T$ with endpoints $X$ and $Y$ (i.e., we have $X_1=X$ and $X_k=Y$). Then, after inserting $(x,y)$, we have $\lambda(u,v)\geq 2$, for any pair of vertices $u,v\in X_1\cup\dots\cup X_k$, and $\lambda(u,v)$ stays the same for any other pair of vertices. In particular, this implies that $X_1\cup\dots\cup X_k$ is a new $2$-edge-connected component. Now, for every edge $(X_i,X_{i+1})$, $i\in\{1,\dots,k-1\}$, of $T$, let $(x_i,y_i)$ be the edge of $G$ that corresponds to $(X_i,X_{i+1})$, and let $(x,y)=(y_0,x_k)$ (this is for notational convenience). Then, for every $i\in\{1,\dots,k\}$, we have $y_{i-1},x_i\in X_i$, and the remaining changes in the $k$-edge-connected components of $G$, after inserting $(x,y)$, are given by supposing that we introduce a virtual edge $(y_{i-1},x_i)$ to $G$; thus they are described sufficiently in the following two cases.}

\item{If $\lambda(x,y)=2$, then $x$ and $y$ belong to the same $2$-edge-connected component $C$, and lie in two different $3$-edge-connected components of $G[C]$. Let $S$ be the cactus of the $3$-edge-connected components of $G[C]$. Then, after inserting $(x,y)$, we have that at least $X$ and $Y$ are now united into a larger $3$-edge-connected component (together with some other nodes of $S$). To describe precisely the nodes of $S$ that get united into a new $3$-edge-connected component, we introduce the concept of the \emph{cycle-path} on $S$ connecting $X$ and $Y$. Let $Z_1,\dots,Z_k$ be a simple path on $S$ with $Z_1=X$ and $Z_k=Y$. Then the cycle-path $Q$ connecting $X$ and $Y$ on $S$ is the set of nodes consisting of $X$, $Y$, and all $Z_i$, $i\in\{2,\dots,k-1\}$, such that the edges $(Z_{i-1},Z_i), (Z_i,Z_{i+1})$ belong to different cycles of $S$. (Notice that this definition is independent of the choice of the simple path $Z_1,\dots,Z_k$.) Then, after inserting $(x,y)$, all the nodes of $Q$ are merged into a new maximal $3$-edge-connected component. The edge-connectivity between any pair of vertices $u$ and $v$, not both of which lie in nodes of $Q$, stays the same. Observe that the new cactus of the $3$-edge-connected components of $G[C]$ is given by $S$ where we have ``squeezed'' every cycle $c$ that contains two consecutive nodes $Z$ and $W$ of $Q$, by merging $Z$ and $W$ into a new node $U$. (If $Z$ and $W$ are not connected with an edge on $S$, then the squeezing of $c$ at $Z$ and $W$ produces two new cycles that meet at $U$.)}

\item{If $\lambda(x,y)\geq 3$, then all the $k$-edge-connected components of $G$, for $k\leq 3$, stay the same \cite{Dinitz:5ECC}.}
\end{enumerate}

\section{The decomposition tree of the maximal $k$-edge-connected subgraphs}
\label{section:incremental_maintenance}

In this section we provide a general framework for maintaining the maximal $k$-edge-connected subgraphs of an undirected graph upon insertions of edges or vertices. This framework relies on a structural characterization of the maximal $k$-edge-connected subgraphs that is derived from the repeated decomposition of the graph into its $k$-edge-connected components. Thus we get a decomposition tree, which can be maintained by using any algorithm for the incremental maintenance of the $k$-edge-connected components (that satisfies some properties) as a black box. We provide a concrete application of this framework: two algorithms for maintaining the maximal $3$-edge-connected subgraphs. These algorithms provide a trade-off between time- and space-bounds, because they rely on different algorithms for maintaining the $3$-edge-connected components, and they utilize the structure of the decomposition tree in different ways.

\subsection{A general framework for maintaining the $k$-edge-connected subgraphs}
\label{section:a_general_framework}
In what follows we let ``$k$-ecc'' mean ``$k$-edge-connected component''. We will not distinguish between the vertex set of a $k$-ecc and the subgraph induced by it. Let $G$ be an undirected multigraph with $m$ edges and $n$ vertices. We consider the decomposition tree $\mathcal{T}$ of the maximal $k$-edge-connected subgraphs of $G$, which is a rooted tree whose nodes correspond to subgraphs of $G$. It is defined recursively as follows. $(1)$ The root of $\mathcal{T}$ corresponds to $G$. $(2)$ Every node of $\mathcal{T}$ that corresponds to a $k$-edge-connected subgraph of $G$ is a leaf. $(3)$ The children of every node $N$ of $\mathcal{T}$ that is not a leaf correspond to the $k$-eccs of the subgraph of $G$ corresponding to $N$. This completes the recursive definition of $\mathcal{T}$. Observe that the size of $\mathcal{T}$ is $O(n)$, because the leaves of $\mathcal{T}$ correspond precisely to the maximal $k$-edge-connected subgraphs of $G$ (which are at most $n$), and every node of $\mathcal{T}$ is either a leaf or it has at least two children.

We observe that $\mathcal{T}$ has the following useful property.

\begin{property}
\label{property:T}
Let $G$ be an undirected graph, and consider the quotient graph $Q$ that is formed by shrinking a maximal $k$-edge-connected subgraph of $G$ into a single vertex. Then $Q$ has the same decomposition tree into maximal $k$-edge-connected subgraphs as $G$.
\end{property}

We use this decomposition tree because we want to maintain the maximal $k$-edge-connected subgraphs of $G$ while new edges or vertices are inserted to it. The definition using successive partitions into $k$-edge-connected components is convenient, because every maximal $k$-edge-connected subgraph lies entirely within a $k$-edge-connected component, and we can utilize algorithms for the incremental maintenance of the $k$-edge-connected components (that satisfy some properties) in order to maintain the decomposition tree under new insertions.

Now let us describe the changes that take place in $\mathcal{T}$ after a new edge is inserted to $G$.
Suppose a new edge $e=(x,y)$ is inserted to $G$. Let $X$ and $Y$ be the maximal $k$-edge-connected subgraphs of $G$ that contain $x$ and $y$, respectively. If $X=Y$, there is nothing to do. Otherwise, we find the nearest common ancestor $N$ of $X$ and $Y$ on $\mathcal{T}$, and perform the insertion of $e$ to $N$. Observe that $N$ is the deepest node of $\mathcal{T}$ that contains both $x$ and $y$. It is sufficient to perform the insertion of $e$ to $N$, because $N$ is separated from its siblings by a $k'$-edge cut, for some $k'<k$, and this does not change even if we insert $e$ to $N$. Thus, all changes in $\mathcal{T}$ after the insertion of $e$ to $G$ take place in the subtree of $N$. If $x$ and $y$ do not become $k$-edge-connected in $N$, then no change takes place in $\mathcal{T}$. Otherwise, some $k$-eccs of $N$ get merged into a larger $k$-ecc. If all vertices of $N$ become $k$-edge-connected, then $N$ becomes a leaf node. Otherwise, we have to find its children that correspond to the $k$-eccs of $N$ that get merged, and substitute them with a single child that corresponds to the bigger $k$-ecc that is formed. Then we have to push down to the new child all the edges that have their endpoints in different $k$-eccs that got merged. In effect, we re-insert those edges to the child of $N$ that represents the new $k$-ecc that is formed. This completes the recursive description of the changes that take place in $\mathcal{T}$ after inserting a new edge to $G$. 

In order to efficiently maintain the decomposition tree, we will have to attach more information to it. We discuss this in the following section, where we use as a black box any algorithm that maintains the $k$-edge-connected components (and satisfies some conditions). In any case, it is interesting to note - and it is essential in proving our time bounds - that the number of edge insertions that can affect the decomposition tree, in any sequence of edges insertions, is $O(n)$. This is formally stated and proved in Theorem~\ref{theorem_1}.

\begin{lemma}\normalfont{\cite{opt-dec-conn}}
\label{lemma:edges_of_trivial}
Let $G$ be a graph with $n$ vertices and $m$ edges, such that every maximal $k$-edge-connected subgraph of $G$ is trivial. Then $m\leq (k-1)(n-1)$. 
\end{lemma}

\begin{theorem}
\label{theorem_1}
Let $k\geq 3$, and let $G_0=(V,E_0)$ be an empty graph with $n$ vertices and no edges ($E_0=\emptyset$). Let $e_1,\dots,e_m$ be a sequence of edges joining vertices of $V$, and define $E_i = \{e_1,\dots,e_i\}$ and $G_i=(V,E_i)$, for $i=1,\dots,m$. Then,  $|\{i\in\{1,\dots,m\} \mid e_i \mbox{ joins two different maximal } k \mbox{-edge-connected subgraphs of } G_{i-1}\}| \leq k(n - 1)$.
\end{theorem}
\begin{proof}
We proceed by induction on the number $n$ of vertices. For $n=1$, the theorem trivially holds, as there are no possible edge insertions. Assume that the theorem holds for every graph with at most $n$ vertices, for some $n\geq 1$. Given an empty graph with $n+1$ vertices, we will show that at most $kn$ edges can affect the decomposition tree during any sequence $e_1,\dots,e_m$ of edge insertions. 

Let ${\ell}$, $0\leq\ell\leq m-1$, be the lowest index such that $G_{\ell +1}=(V,E_{\ell +1})$ contains a non-trivial maximal $k$-edge-connected subgraph (i.e., $G_{i}=(V,E_{i})$, $0\leq i\leq\ell$ contain only trivial maximal $k$-edge-connected subgraphs). If no such index $\ell$ exists, then the final graph $G_m$ has only trivial maximal $k$-edge-connected subgraphs and the theorem  holds by Lemma~\ref{lemma:edges_of_trivial}.
Otherwise, let $S$ be the non-trivial maximal $k$-edge-connected subgraph that appears in $G_{\ell+1}$ after the edge $e_{\ell +1}$ has been inserted to $G_{\ell}$, and let $S$  contain $d$ vertices ($d\geq 2$) and $t$ edges. Note that $e_{\ell +1}$ must necessarily be in $S$. Let $S'= S\setminus e_{\ell +1}$: then $S'$ has $d\geq 2$ vertices and $t-1$ edges, and contains only trivial maximal $k$-edge-connected subgraphs.
By Lemma~\ref{lemma:edges_of_trivial} applied to $S'$, we have $t-1\leq (k-1)(d-1)$, and therefore $t\leq (k-1)(d-1)+1$. 

Now let us consider the quotient graph of $G_{\ell+1}/S$ that is formed by contracting $S$ into one single vertex. 
This quotient graph has $\ell +1-t$ edges and $(n+1)-d+1=n-d+2\leq n$ vertices (since $d\geq 2$).  Since we are contracting a maximal $k$-edge.connected subgraph, the quotient graph has exactly the same decomposition tree as $G_{\ell+1}$. By induction, there are at most $k((n-d+2)-1)-(\ell +1-t)=kn-kd+k-\ell -1+t$ new edges that can be added to the quotient graph and affect its decomposition tree. This number is at most $kn-\ell-d+1$ (since $t$ is at most $kd-k-d+2$). Since we have already added $\ell +1$ edges to the initial graph, we have that, in total, at most $kn-d+2$ edges can affect the decomposition tree. The fact that $d\geq 2$ yields the desired result.
\end{proof}

\subsection{Maintaining the decomposition tree}
\label{section:maintaining_T_g}

In order to be able to update $\mathcal{T}$ after a new insertion to the graph, we need a $k$-edge-connected components algorithm that satisfies the following properties. The algorithm works on an incremental graph $G$ and maintains a unique label for every $k$-ecc of $G$. Furthermore, it supports the following two operations.

\begin{enumerate}
\item{When a new edge is inserted to $G$, return the labels of the $k$-eccs of $G$ that get merged due to this insertion, the label of the new $k$-ecc that is formed, and the set of the interconnection edges between those $k$-eccs that get merged.}
\item{Given a collection of vertex-disjoint graphs $G_1,\dots,G_t$, join the underlying data structures that the algorithm uses for those graphs into a new data structure for the graph $G=G_1\cup\dots\cup G_t$.}
\end{enumerate}

We call an algorithm (plus the underlying data structure) that supports these operations an ``incremental $k$-edge-connected components algorithm'' (or ``a data structure for maintaining the $k$-edge-connected components'').

Now suppose that we are equipped with an incremental $k$-edge-connected components algorithm. In order to efficiently maintain the decomposition tree $\mathcal{T}$, we associate with every non-leaf node $N$ a data structure for maintaining the $k$-edge-connected components. Furthermore, for every $k$-ecc of $N$, we have a two-way correspondence between its label maintained by the data structure and the child of $N$ that corresponds to it. Now suppose that an edge $e$ is inserted to $N$. Then we use operation $(1)$ in order to find the labels $C_1,\dots,C_t$ of the $k$-eccs of $N$ that get merged due to the insertion of $e$. If there is only one such $k$-ecc returned, there is nothing to do. If all the $k$-eccs of $N$ are returned, then we make $N$ a leaf node. Otherwise, we get from those labels the corresponding children of $N$, we merge those children into a new child, and we establish a two-way correspondence between this child and the new $k$-ecc of $N$ that is formed. Furthermore, we retrieve the associated data structures of the children of $N$ that get merged, we join them into a new data structure using operation $(2)$, and we associate the new data structure with the new child of $N$. Finally, we insert the interconnection edges between $C_1,\dots,C_t$ that we got from operation $(1)$ into the new child of $N$. This completes the recursive description of the update of $\mathcal{T}$ using an incremental $k$-edge-connected components algorithm as a black box.

We will describe two different methods with which we can perform the merging of nodes of $\mathcal{T}$, with a trade-off in time- and space-complexity. The first method is the simplest one. Let $N$ be a node of $\mathcal{T}$, and let $C_1,\dots,C_t$ be the children of $N$ that we have to merge. Then we will use the node with the greatest number of children among $C_1,\dots,C_t$ as the new node in which $C_1,\dots,C_t$ get merged. So let $C$ be a node among $C_1,\dots,C_t$ with maximum number of children. Then we discard all nodes in $\{C_1,\dots,C_t\}\setminus\{C\}$, and we redirect the parent pointer of their children to $C$. Since this involves at most $n$ nodes at every level, and a parent pointer can be redirected at most $\log n$ times, we get at most $O(n\log n)$ redirections of parent pointers at every level, and $O(n^2\log n)$ redirections in total. In the second method we maintain all nodes of $\mathcal{T}$ (possibly in a ``deactivated'' mode) throughout the sequence of all insertions, and thus we may need as much as  $\Omega(n^2)$ space (since there are at most $O(n)$ insertions that can affect the decomposition tree). The idea is to use a disjoint-set union data structure $\mathit{DSU}_i$ \cite{tarjanEfficiencyGoodNot1975} on the nodes of level $i$ of the tree, for all $i\geq 1$. In order to merge nodes at level $i$ we unite all of them using $\mathit{DSU}_i$ and we choose one of them as the representative. We consider the representative to be an \emph{active} node, in the sense that we can use its parent pointer to move to level $i-1$, whereas the parent pointer of the other nodes that got united will never be used again. Thus, in order to access the parent of a node $N$ at level $i$, we use $\mathit{find}_{i-1}(\mathit{parent}(N))$, where $\mathit{parent}$ denotes the parent pointer in $\mathcal{T}$, and $\mathit{find}_{i-1}$ is the $\mathit{find}$ operation of $\mathit{DSU}_{i-1}$. Notice that this approach has the advantage that it does not explicitly redirect the parent pointers. Using an optimal implementation for $\mathit{DSU}_i$ \cite{tarjanEfficiencyGoodNot1975}, we can perform any sequence of $m$ operations $\mathit{find}_i$ or $\mathit{unite}_i$ in $O(m\alpha(m,n))$ time in total. Since at most $n$ $\mathit{unite}$ operations can take place at each level, we have an $O(n\alpha(n,n))$ time-bound at every level for the mergings, and an $O(n^2\alpha(n,n))$ time-bound in total.

We note that the $n^2$ expression in the time-bound is an inherent bottleneck of this approach for maintaining the maximal $k$-edge-connected subgraphs, for any $k\geq 3$. In order to demonstrate this, let us introduce some concepts. We call the edges that connect two different maximal $k$-edge-connected subgraphs of a graph \emph{$k$-interconnection} edges. Then we partition the $k$-interconnection edges into levels as follows. If an edge participates in a $k'$-edge cut of the graph, for $k'<k$, we define the level of this edge to be $1$. Now suppose that we have defined all $d$-level edges, for some $d\geq 1$. Then an edge that participates in a $k'$-edge cut of (a connected component of) the graph that remains after we remove all $d'$-level edges, for $d'\leq d$, is assigned level $d+1$. Now we observe that our method for updating the decomposition tree explicitly maintains the levels of the $k$-interconnection edges (because every $d$-level $k$-interconnection edge is essentially maintained in the associated data structure of a node of depth $d$ of the decomposition tree). There are sequences of insertions of edges in which $\Omega(n)$ insertions may force $\Omega(n)$ edges to change their level, and thus we need at least $\Omega(n^2)$ time in total to maintain the decomposition tree. An example for $k=3$ is shown in Figure~\ref{figure:example-insertions}. (This generalizes easily to any $k>3$.) 

\begin{figure*}[t!]
\begin{center}
\centerline{\includegraphics[trim={0 0cm 0 14.5cm}, clip=true, width=\textwidth]{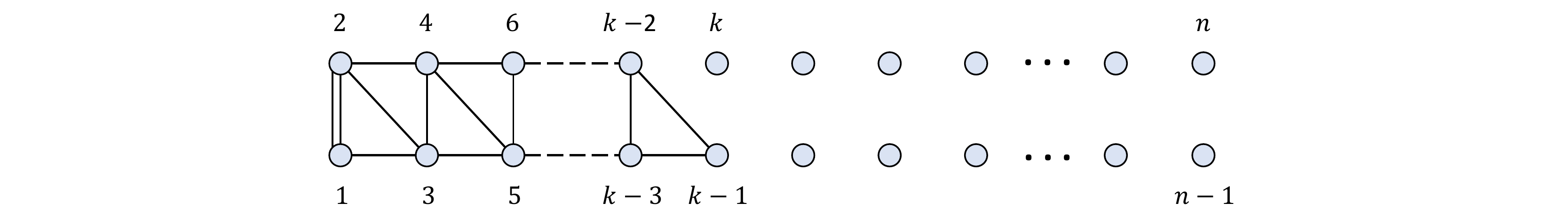}}
\caption{An example of a sequence of insertions that can force $\Omega(n)$ $3$-interconnection edges to increase their level by one $\Omega(n)$ times. The sequence of insertions, starting from the empty graph with vertices in $\{1,\dots,n\}$ is $(1,2),(1,2),(3,1),(3,2),(4,2),(4,3),\dots,(k,k-2),(k,k-1),\dots,(n,n-2),(n,n-1)$. After the insertion of $(k,k-1)$, we have that $(k,k-1)$ and $(k,k-2)$ are the only $1$-level edges, and the level of the other edges $(1,2),\dots,(k-1,k-2)$ increases by one. Note that $\{1,\dots,k-1\}$ is a $3$-edge-connected component at this point, whereas all the maximal $3$-edge-connected subgraphs are trivial. \label{figure:example-insertions}}
\end{center}
\end{figure*}

\subsection{The decomposition tree of the maximal $3$-edge-connected subgraphs}
\label{section:decomposition_tree}
Here we apply the framework outlined in the previous section for the special case $k=3$. We focus on the case $k=3$ for the following reasons. $(1)$ There exist very efficient/asymptotically optimal algorithms for the incremental maintenance of the $3$-edge-connected components that we can use for maintaining the decomposition tree for the case $k=3$. $(2)$ Although there exist very efficient algorithms for maintaining also the $4$- and $5$-edge-connected components \cite{dinitzMaintainingClasses4EdgeConnectivity1998, Dinitz:5ECC}, and they seem applicable in our framework, they involve a lot of technicalities and it would be very tedious to show how to extend them in order to apply them for our purposes. For $k>5$ we are not even aware of any efficient incremental algorithms that we could apply.\footnote{The fully dynamic algorithm of Jin and Sun~\cite{st-conn_Sun} seems very difficult to be usable in our framework, because it does not explicitly maintain the $k$-edge-connected components.} $(3)$ The case $k=3$ is the simplest one to consider. And yet, even here, we have to extend appropriately the existing data structures and algorithms, and this involves various technicalities.

For the case $k=3$ we incorporate the algorithms of \cite{galilMaintaining3EdgeConnectedComponents1993,lapoutreMaintainance23ecc,lapoutreMaintainance23eccII,westbrookMaintainingBridgeConnectedBiconnected1992}, for maintaining the $3$-eccs, in the decomposition tree. Thus, since these algorithms basically maintain the $2$-eccs of the $1$-eccs, and the $3$-eccs of the $2$-eccs, we augment the decomposition tree with more levels in order to capture the decomposition into all $k'$-eccs, for $k'\leq 3$. To be more precise, we consider the decomposition tree $\mathcal{T}$ of the maximal $3$-edge-connected subgraphs of $G$ that is formed by repeatedly removing all $1$- and $2$-edge cuts. The nodes of $\mathcal{T}$ correspond to subgraphs of $G$. First, the root of $\mathcal{T}$ corresponds to the entire graph. If $G$ is $3$-edge-connected, then this is the only node of $\mathcal{T}$, and the decomposition is over. Otherwise, the children of the root correspond to the connected components of $G$, which we call $1$-ecc nodes. Then the children of every $1$-ecc node correspond to its $2$-edge-connected components, and we call them $2$-ecc nodes. Finally, the children of every $2$-ecc node correspond to its $3$-edge-connected components and we call them $3$-ecc nodes. Now, if a $3$-ecc node corresponds to a $3$-edge-connected subgraph of $G$, then it is a leaf of $\mathcal{T}$, and a maximal $3$-edge-connected subgraph of $G$. Otherwise, it becomes the root of a subtree of $\mathcal{T}$ that is produced recursively with the same procedure. In this way, we compute the decomposition tree $\mathcal{T}$ of the maximal $3$-edge-connected subgraphs of $G$. (See Figure~\ref{figure:decompositiontreeexample} for an example.)

\begin{figure*}[t!]
\begin{center}
\centerline{\includegraphics[trim={0 0 0 0cm}, clip=true, width=\textwidth]{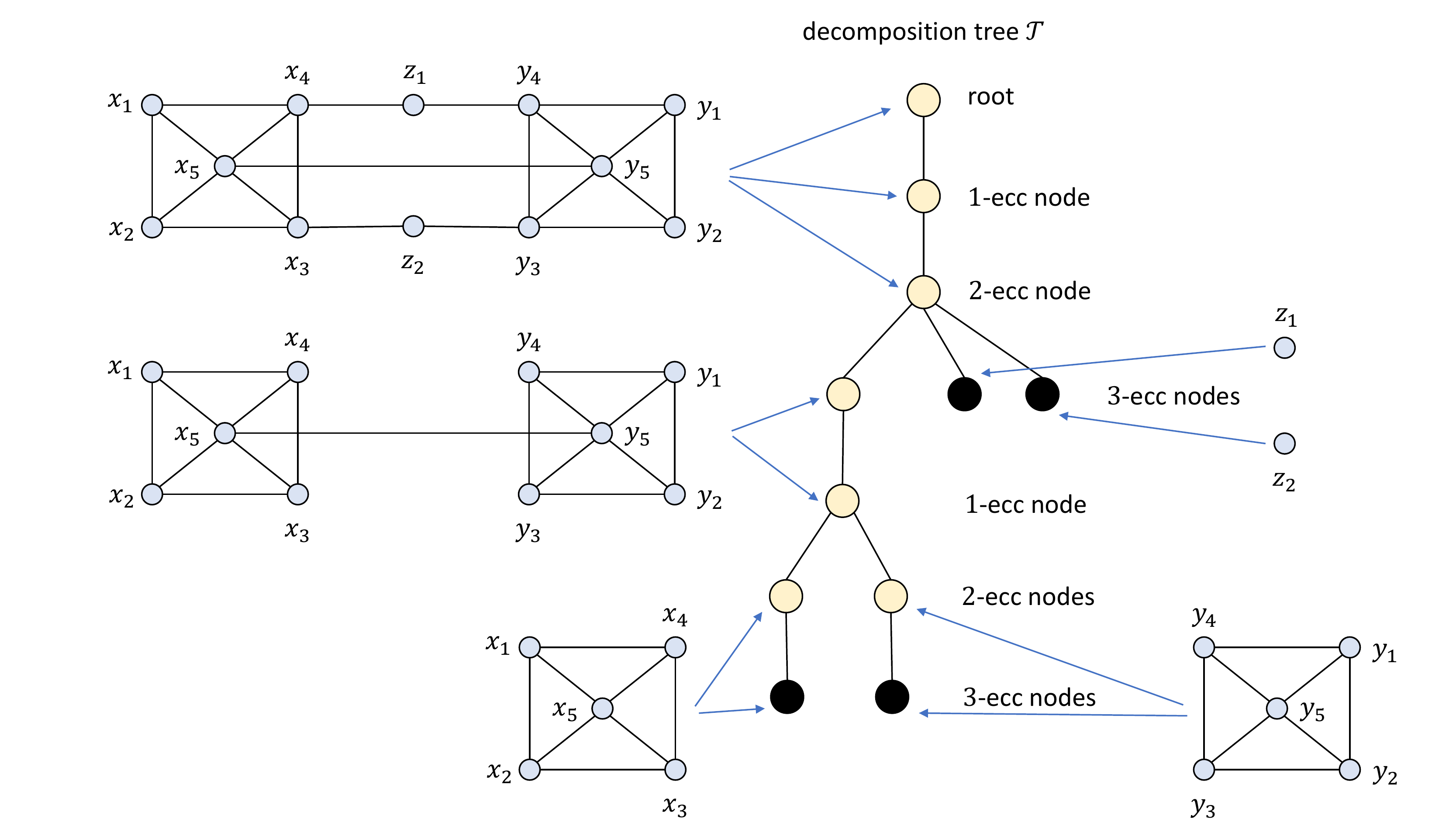}}
\caption{The decomposition tree $\mathcal{T}$ of the graph $G$ of Figure~\ref{figure:max3ecsexample}. Notice the correspondence between the nodes of $\mathcal{T}$ and some subgraphs of $G$. The leaves of $\mathcal{T}$ (coloured in black) correspond to the maximal $3$-edge-connected subgraphs of $G$.\label{figure:decompositiontreeexample}}
\end{center}
\end{figure*}

Observe that the size of $\mathcal{T}$ is $O(n)$,\footnote{However, in order to efficiently maintain $\mathcal{T}$ after the insertion of new edges to $G$, we may have to attach more information to $\mathcal{T}$, so that its total size will be $O(n^2)$. (See Section~\ref{subsection:alpha_alg}.)} since the children of every node $N$ of $\mathcal{T}$ correspond to subgraphs of the graph corresponding to $N$ that partition its vertex set (and although some nodes have only one child, by construction this child must have at least two children or two grandchildren, unless it is a leaf). For convenience, and since no ambiguity arises, we identify every node of $\mathcal{T}$ with the subgraph of $G$ that corresponds to it. (We do not maintain this correspondence explicitly in any data structure, as this would be impractically expensive in terms of space (and therefore time). It is only a conceptual convention that we make, and we will use it extensively in the sequel.) Furthermore, since $\mathcal{T}$ is rooted, we let the nodes of $\mathcal{T}$ be partitioned into levels according to their depth, where the level of the root is $0$. Thus we may speak of a subgraph of $G$ at some level of $\mathcal{T}$. Observe that $\mathcal{T}$ contains $3k$ levels, for some $k\geq 0$, and that every maximal $3$-edge-connected subgraph of $G$ is contained at some level $3k'$, for $k'\leq k$, as a leaf.
It is useful to note that $\mathcal{T}$ satisfies the following.

\begin{property}
\label{property:shrinking_g_vertices}
If we shrink a maximal $3$-edge-connected subgraph of $G$ into a single vertex, then the decomposition tree of the resulting graph is given again by $\mathcal{T}$.
\end{property}


Now we will give a high-level description of the changes that $\mathcal{T}$ undergoes when a new edge $(x,y)$ is inserted to $G$. 
Theorem~\ref{theorem_1} implies that at most $3n-3$ insertions of edges may affect the decomposition tree, in any sequence of edge insertions, where $n$ is the total number of vertices inserted.
By Property~\ref{property:shrinking_g_vertices}, we have that a newly inserted edge can affect the decomposition tree only if its endpoints lie in different maximal $3$-edge-connected subgraphs of $G$.
Thus we consider the nearest common ancestor $N$ of the maximal $3$-edge-connected subgraphs of $G$ that contain $x$ and $y$. This corresponds to the deepest node in the decomposition tree that contains both $x$ and $y$, and it should be clear that the insertion of $(x,y)$ affects only the decomposition of $N$. We distinguish three different cases, depending on whether $(1)$ $N$ is the root or a $3$-ecc node, or $(2)$ $N$ is a $1$-ecc node, or $(3)$ $N$ is a $2$-ecc node. The whole procedure is summarized in Algorithm~\ref{algorithm:tree_modification}.

\subsubsection{$N$ is the root or a $3$-ecc node}
In this case $(x,y)$ joins two different connected components $C_1$ and $C_2$ of $N$, and so it becomes a bridge of $N$ that connects them. Thus we only have to merge $C_1$ and $C_2$ into a new $1$-ecc node (see Fig.~\ref{fig:3ccexample}).

\begin{figure}[t!]\centering
\includegraphics[width=0.9\linewidth]{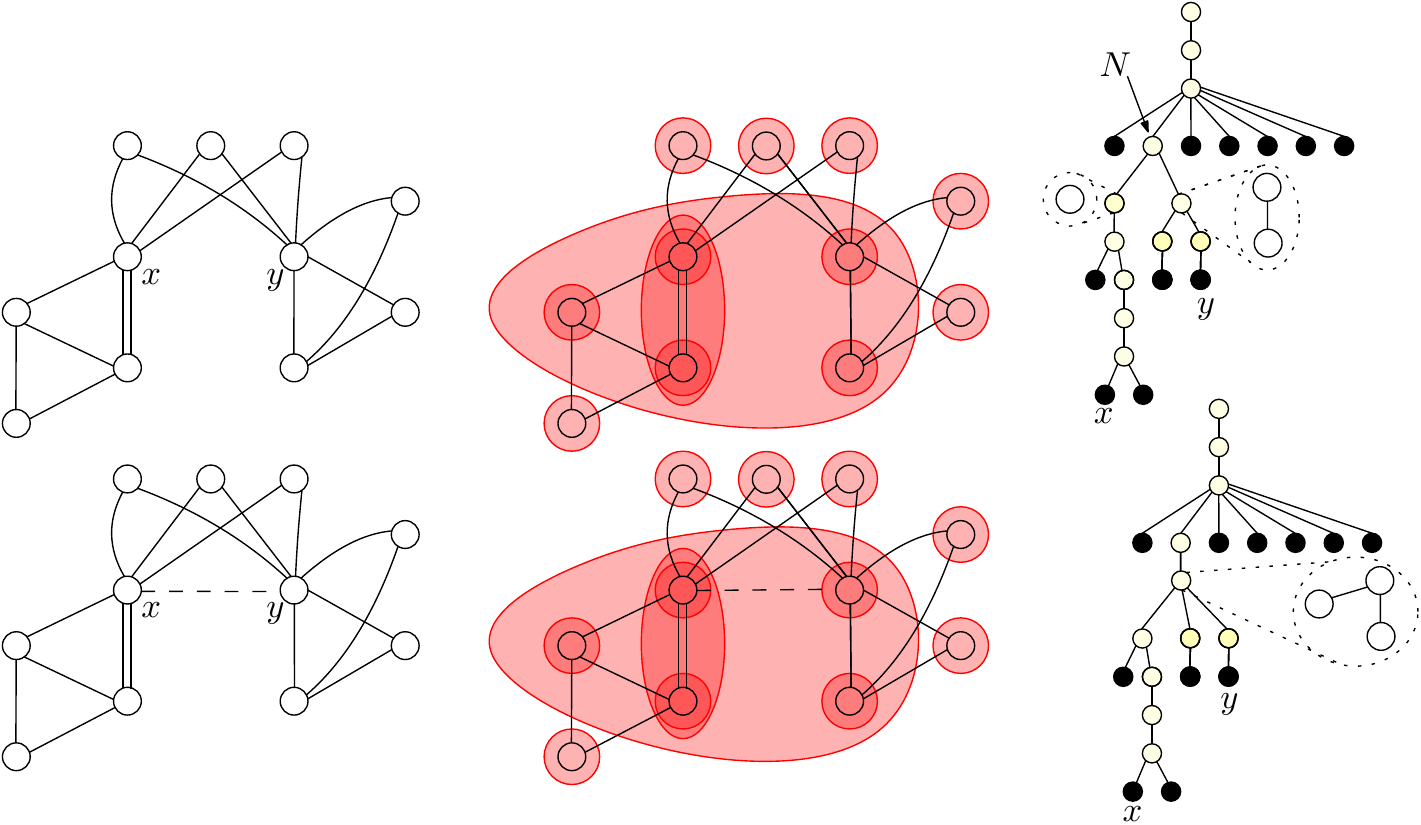}
\caption{\small{Here we show the changes to the decomposition tree on the insertion of the edge $(x,y)$. The nearest common ancestor $N$ of (the maximal $3$-edge-connected subgraphs that contain) $x$ and $y$ is a $3$-ecc node. Therefore, we need to know the two children of $N$ that contain $x$ and $y$; then we retrieve the associated trees, join them properly with (an equivalent of) the edge $(x,y)$, and then we merge the two nodes. Finally, we associate with it the new tree that has been formed. In this figure, the middle graphs show the decomposition of the graph into $3$-edge-connected components at every level. Also, in the decomposition trees, we can see the trees associated with the nodes we have to merge and the tree associated with the merged node. (The new edge $(x,y)$ is stored inside this tree.) The $3$-edge-connected components are highlighted with pink colour. Lower intensity signifies greater depth in the decomposition tree.}}\label{fig:3ccexample}
\end{figure}

\subsubsection{$N$ is a $1$-ecc node}
In this case $(x,y)$ joins two different $2$-edge-connected components $C_1$ and $C_2$ of $N$. Without loss of generality, assume that $x\in C_1$ and $y\in C_2$. Let $T$ be tree of the $2$-edge-connected components of $N$, and let $P=X_1,\dots,X_k$ be the path on $T$ with endpoints $C_1$ and $C_2$. (Thus we have $X_1=C_1$ and $X_k=C_2$.) Then the vertices that are contained in $X_1,\dots,X_k$ become $2$-edge-connected, and for every pair of vertices not both of which belong to $X_1\cup\dots\cup X_k$ the edge-connectivity remains the same. Furthermore, for every $i\in\{1,\dots,k-1\}$, let $(x_i,y_i)$ be the bridge of $N$ that corresponds to the edge $(X_i,X_{i+1})$ of $P$, and let also $(x,y)=(y_0,x_k)$ (this is for notational convenience). Then, every pair of edges in $\{(x_i,y_i)\mid i\in\{1,\dots,k-1\}\}\cup\{(x,y)\}$ is a $2$-edge cut of $N$. Thus we have to merge all $2$-ecc nodes $X_1,\dots,X_k$ into a new $2$-ecc node $C$, and no change takes place on $\mathcal{T}$ outside the subtree of $C$.

Now we have to consider the changes that possibly take place in the subtree of $C$. For every $i\in\{1,\dots,k\}$, let $S_i$ be the cactus of the $3$-edge-connected components of $X_i$, and let $Q_i$ be the cycle-path on $S_i$ with endpoints the $3$-ecc of $X_i$ that contains $y_{i-1}$ and the $3$-ecc of $X_i$ that contains $x_i$. Let also $D_{(i,1)},\dots,D_{(i,t(i))}$ be the subgraphs of $X_i$ that correspond to the nodes of $Q$. Then, the vertices that are contained in $D_{(i,1)},\dots,D_{(i,t(i))}$ become $3$-edge-connected; furthermore, for every pair of vertices of $X_i$ not both of which lie in $D_{(i,1)}\cup\dots\cup D_{(i,t(i))}$ the edge-connectivity remains the same. Thus we have to merge all $D_{(i,1)},\dots,D_{(i,t(i))}$ into a new $3$-ecc node $D_i$; furthermore, no change in the subtree of $X_i$ takes place outside the subtree of $D_i$. Now we have to consider the graphs $D_{(i,1)},\dots,D_{(i,t(i))}$ as the connected components of the new graph $D_i$. However, in order to maintain the decomposition subtree of $D_i$ on $\mathcal{T}$, we have to take care of two things. First, some graphs $D_{(i,j)}$, for $j\in\{1,\dots,t(i)\}$, might be leaves (prior to the insertion of $(x,y)$). This means that they are $3$-edge-connected subgraphs, and so we have to include them in the subtree of $D_i$ by expanding their corresponding nodes on $\mathcal{T}$ with the addition of three intermediary nodes. And secondly, there might exist edges in $X_i$ that connect some of $D_{(i,1)},\dots,D_{(i,t(i))}$. (Note that these edges correspond naturally to those of $S_i$ that connect the nodes of $Q_i$.) These edges are now included in $D_i$, and so we have to consider their effect on the decomposition subtree of $D_i$. Thus, in order to capture fully the effect on $\mathcal{T}$ of the insertion of $(x,y)$, we have to re-insert those edges to $G$. Observe that these edges constitute parts of $1$- or $2$-edge cuts of $D_i$, and so they will be re-inserted only once in order to perform the insertion of $(x,y)$. However, their re-insertion may force other edges of $G$, that lie in graphs on deeper levels, to be re-inserted to $G$. Nevertheless, the nearest common ancestor involved in each such re-insertion will either be a $3$-ecc or a $1$-ecc node. As a consequence, we note that no new maximal $3$-edge-connected subgraphs will be formed after the insertion of $(x,y)$ to $G$. Figure~\ref{figure:1cc} is an example of this case.

\begin{figure}[t!]\centering
\includegraphics[width = 0.8\linewidth,trim={1.5cm 17cm 2cm 0cm}]{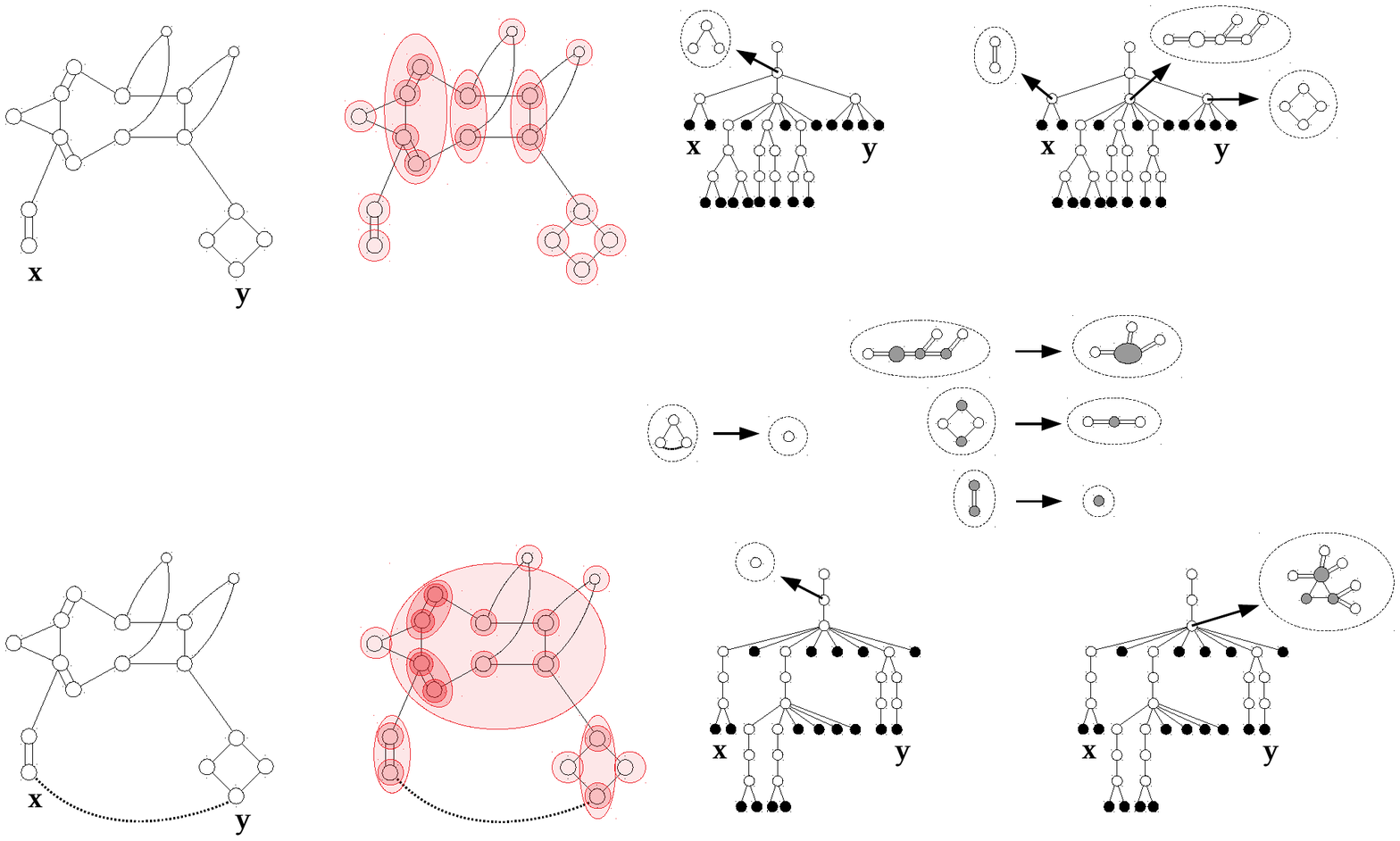}
\caption{\small{Here we insert the edge $(x,y)$ to the graph. This joins two different $2$-edge-connected components (at the first decomposition level), which lie in the tree $T$ associated with the nearest common ancestor of $\mathbf{x}$ and $\mathbf{y}$. The path of those $2$-eccs on $T$ consists of three nodes, which we have to merge. So first we have to retrieve their associated cactuses (which are shown in this figure), to merge some of their nodes using $\mathtt{merge3ecc}$, and then to join them on a new cactus using (essentially) the edges of the cycle. Here we can see the effect of $\mathtt{merge3ecc}$ on the three cactuses. With gray are shown the nodes that we have to merge. The resulting three cactuses get joined on their gray (compressed) nodes along the edges of the triangle, one of which is (essentially) the new one, and the other two lay on a higher level (inside the tree $T$ of the $2$-eccs). Thus, those two edges got pushed down one level. Also notice that the height of the decomposition tree increased.}} \label{figure:1cc}
\end{figure}

We refer to the three step process of $(a)$ merging all $D_{(i,1)},\dots,D_{(i,t(i))}$ into a single node $D_i$, $(b)$ expanding every node $D_{(i,j)}$ that is a leaf, for $j\in\{1,\dots,t(i)\}$, with the addition of three intermediary nodes, and $(c)$ re-inserting into $G$ the inter-edges between the subgraphs $D_{(i,1)},\dots,D_{(i,t(i))}$ of $X_i$, as \emph{merging} the $3$-ecc nodes $D_{(i,1)},\dots,D_{(i,t(i))}$ into $D_i$. This procedure is shown in Algorithm~\ref{algorithm:merge3ecc}. (The variable ``parent'' in Lines \ref{line:alg_p1}, \ref{line:alg_p2} and \ref{line:alg_p3} denotes the parent relation of $\mathcal{T}$.)

\subsubsection{$N$ is a $2$-ecc node}
In this case $(x,y)$ joins two different $3$-edge-connected components $C_1$ and $C_2$ of $N$. We note that this is the only case in which the formation of new maximal $3$-edge-connected subgraphs of $G$ may take place (by merging together smaller maximal $3$-edge-connected subgraphs of $G$). Let $S$ be cactus of the $3$-edge-connected components of $N$, and let $Q=X_1,\dots,X_k$ be the cycle-path on $S$ with endpoints $C_1$ and $C_2$. Then the vertices that are contained in $X_1,\dots,X_k$ become $3$-edge-connected, and for every pair of vertices not both of which belong to $X_1\cup\dots\cup X_k$ the edge-connectivity remains the same. Thus, if $Q$ contains all the nodes of $S$, then $N$ becomes $3$-edge-connected. Therefore, if $N$ prior to the insertion of $(x,y)$ was the only $2$-edge-connected component of its grandparent $R$ on $\mathcal{T}$, then this means that $R$ becomes $3$-edge-connected, and so its subtree on $\mathit{T}$ is condensed into $R$. Otherwise, all the proper descendants of $N$ are condensed into a new single $3$-ecc node of $\mathcal{T}$, which corresponds to the new maximal $3$-edge-connected subgraph that has been formed.

Now suppose that $Q$ does not contain all the nodes of $S$. Then we have to perform a merging of the $3$-ecc nodes $X_1,\dots,X_k$ into a new node $D$, and then repeat the insertion of the edge $(x,y)$. Let $X$ and $Y$ be the maximal $3$-edge-connected subgraphs of $G$ that contain $x$ and $y$, respectively. Then the nearest common ancestor of $X$ and $Y$ on $\mathcal{T}$ is a descendant of $D$, and it can either be a $3$-ecc node or a $1$-ecc node (in which case we are work as previously), or a $2$-ecc node again. In the last case, either the formation of a new maximal $3$-edge-connected subgraph of $G$ will take place (and we are done), or we will have to merge again some $3$-ecc nodes of $\mathcal{T}$ into a new node and then repeat the insertion of $(x,y)$. Observe that, eventually, this process must terminate, since every time that we have to merge some $3$-ecc nodes into a new node, we leave some of their siblings on the same level, and then the computation involves only the subtree of the new node. 

Figure~\ref{figure:2ccsimple} is a simple example of this case, where the whole graph becomes $3$-edge-connected.

\begin{figure}[t!]\centering
\includegraphics[width=0.9\linewidth]{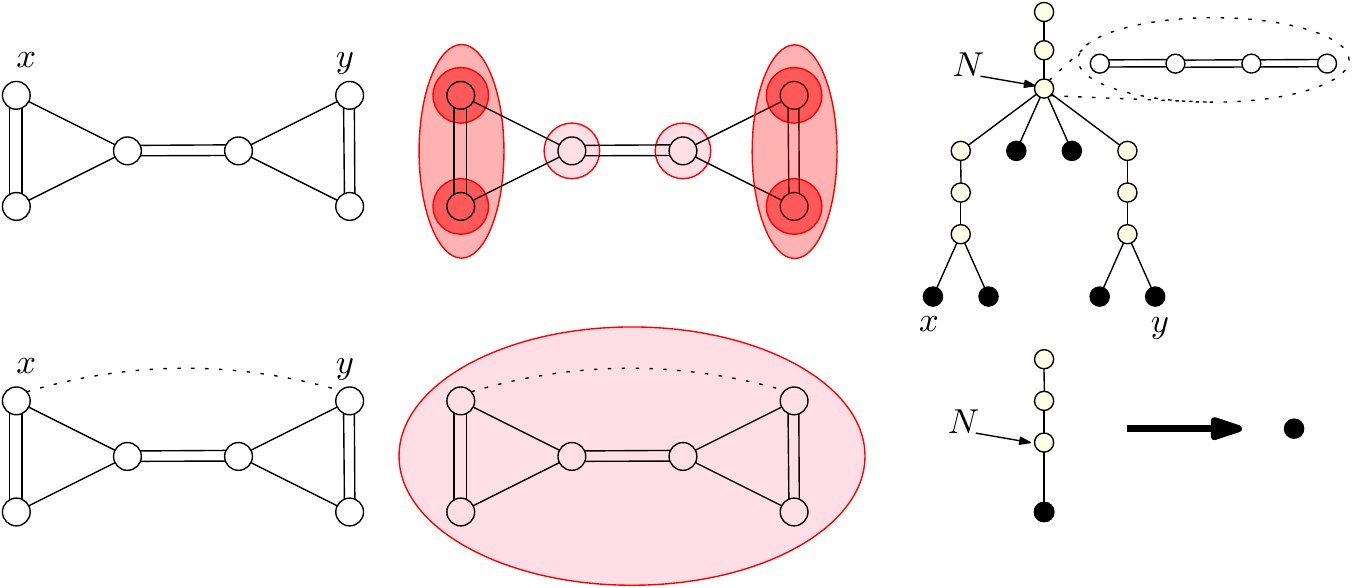}
\caption{\small{This is a simple example where the insertion of a new edge $(x,y)$ makes the whole graph $3$-edge-connected. This happens immediately, because the nearest common ancestor $N$ of $x$ and $y$ is a $2$-ecc node, where the associated cactus is the cycle-path connecting the nodes that contain $x$ and $y$. Thus we merge all leaves of $N$, we discard the whole subtree $N$, and assign as a child of $N$ the new maximal $3$-edge-connected subgraph that has been formed. Notice that the parent of $N$ has only one child, and the parent of the parent of $N$ has again only one child. Thus, we replace the grandparent of $N$ with the child of $N$, and we discard $N$, its parent, and its grandparent.}}\label{figure:2ccsimple}
\end{figure}


\begin{algorithm}[t!]
\caption{\textsf{Update the decomposition tree $\mathcal{T}$ after inserting the edge $(x,y)$ to $G$}}
\label{algorithm:tree_modification}
\LinesNumbered
\DontPrintSemicolon
\textbf{procedure} $\mathtt{insert}(x,y)$\;
\Begin{
$X\leftarrow$ the leaf of $\mathcal{T}$ that contains $x$\;
\label{alg:find_3ecsx}
$Y\leftarrow$ the leaf of $\mathcal{T}$ that contains $y$\;
\label{alg:find_3ecsy}
$N\leftarrow$ the nearest common ancestor of $X$ and $Y$ on $\mathcal{T}$\;
\label{alg:nca}
$C_1\leftarrow$ the child of $N$ that contains $x$\;
$C_2\leftarrow$ the child of $N$ that contains $y$\;
\If{$N$ is the root of $\mathcal{T}$ or a $3$-ecc node}{
  merge $C_1$ and $C_2$ into a new $1$-ecc node\;
  \label{alg:merge1}
}
\ElseIf{$N$ is a $1$-ecc node}{
  $T\leftarrow$ the tree of the $2$-edge-connected components of $N$\;
  $X_1,\dots,X_k\leftarrow$ the path on $T$ with endpoints $C_1$ and $C_2$\;
  \label{alg:tree_path}
  \ForEach{$i\in\{1,\dots,k-1\}$}{$(x_i,y_i)\leftarrow$ the edge of $N$ that corresponds to $(X_i,X_{i+1})$ of $T$\;
   \label{alg:tree_edges}}
  $(y_0,x_k)\leftarrow (x,y)$\;
  \ForEach{$i\in\{1,\dots,k\}$}{
    $S_i\leftarrow$ the cactus of the $3$-edge-connected components of $X_i$\;
    $D_{(i,1)},\dots,D_{(i,t(i))}\leftarrow$ the cycle-path on $S_i$ with endpoints the $3$-ecc of $X_i$ that contains $y_{i-1}$ and the $3$-ecc of $X_i$ that contains $x_i$\;
    \label{alg:cactus_path1}
    $\mathtt{merge3ecc}(D_{(i,1)},\dots,D_{(i,t(i))})$\;
    \label{alg:call_merge3}
  }
  merge all $X_1,\dots,X_k$ into a new $2$-ecc node of $\mathcal{T}$\;
  \label{alg:merge2}
}
\ElseIf{$N$ is a $2$-ecc node}{
  $S\leftarrow$ the cactus of the $3$-edge-connected components of $N$\;
  $D_1,\dots,D_k\leftarrow$ the cycle-path on $S$ with endpoints $C_1$ and $C_2$\;
  \label{alg:cactus_path2}
  \If{$V(S)=\{D_1,\dots,D_k\}$}{
    $R\leftarrow$ the grandparent of $N$ on $\mathcal{T}$\;
    \lIf{$N$ is the only $2$-ecc of $R$}{condense the subtree of $R$ into $R$}
    \lElse{condense all the proper descendants of $N$ into a new $3$-ecc node}
  }
  \Else{
    $\mathtt{merge3ecc}(D_1,\dots,D_k)$\;
    \label{alg:re-insert}
    $\mathtt{insert}(x,y)$\;  
    \label{alg:re-insert1}
  }
}}
\end{algorithm}

\begin{algorithm}[h!]
\caption{\textsf{merge the children $D_1,\dots,D_k$ of a $2$-ecc node $X$}}
\label{algorithm:merge3ecc}
\LinesNumbered
\DontPrintSemicolon
\textbf{procedure} $\mathtt{merge3ecc}(D_1,\dots,D_k)$\;
\Begin{
\For{$i\in\{1,\dots,k\}$}{
  \If{$D_i$ is a leaf of $\mathcal{T}$}{
     let $D_i'$, $D_i''$ and $D_i'''$ be a new $1$-ecc, $2$-ecc and $3$-ecc node, respectively, where $D_i'$, $D_i''$ and $D_i'''$ correspond to the same graph as $D_i$\;
     $\mathit{parent}(D_i''')\leftarrow D_i''$\;
     \label{line:alg_p1}
     $\mathit{parent}(D_i'')\leftarrow D_i'$\;
     \label{line:alg_p2}
     $\mathit{parent}(D_i')\leftarrow D_i$\;
     \label{line:alg_p3}
  }
}
$S\leftarrow$ the cactus of the $3$-edge-connected components of $X$\;
$\mathcal{E}\leftarrow$ the edge set of $S[D_1\cup\dots\cup D_k]$\;
\label{alg:cactus_edges}
merge all $D_1,\dots,D_k$ into a new $3$-ecc node\;
\label{alg:merge3}
\ForEach{edge $e$ in $\mathcal{E}$}{
  $\mathtt{insert}(e)$\;
  \label{alg:re-insert2}
}}
\end{algorithm}


\section{Maintaining the decomposition tree after insertions}
\label{section:maintaining_T}
As seen in Algorithms \ref{algorithm:tree_modification} and \ref{algorithm:merge3ecc}, in order to efficiently update $\mathcal{T}$ after inserting an edge $(x,y)$ to $G$, we have to provide efficient implementations for the following procedures: 

\begin{enumerate}[label={(\arabic*)}]
\item{Find the leaves $X$ and $Y$ of $\mathcal{T}$ that contain $x$ and $y$, respectively.}
\item{Find the nearest common ancestor $N$ of $X$ and $Y$ on $\mathcal{T}$.} 
\item{Find the grandchild of $N$ that contains a particular vertex $v\in N$. (This is needed in Line~\ref{alg:cactus_path1} of Algorithm~\ref{algorithm:tree_modification}, in order to find the $3$-eccs of $X_i$ that contain $y_{i-1}$ and $x_i$.)}
\item{Merge sets of nodes of $\mathcal{T}$ (lying on the same level) into a new single node.} 
\item{Find the nodes of $\mathcal{T}$ that have to get merged, and the inter-connection edges between the subgraphs that correspond to them.}
\end{enumerate}

We will provide two different solutions for this set of procedures. In Section~\ref{subsection:log_alg} we rely on the ideas and the data structures of \cite{galilMaintaining3EdgeConnectedComponents1993,lapoutreMaintainance23ecc,westbrookMaintainingBridgeConnectedBiconnected1992}, in order to provide an algorithm that uses $O(n)$ space and runs in $O(n^2\log^2 n + m\alpha(m,n))$ time, for any sequence of $m$ edge and $n$ vertex insertions. In Section~\ref{subsection:alpha_alg} we rely on \cite{lapoutreMaintainance23ecc} and \cite{lapoutreMaintainance23eccII}, in order to provide an algorithm that uses $O(n^2)$ space and runs in $O(n^2\alpha(n,n) + m\alpha(m,n))$ time, for any sequence of $m$ edge and $n$ vertex insertions.
We note that the first algorithm is more readily amenable to implementations. The second algorithm, although it is asymptotically more time-efficient, it uses the sophisticated data structures of \cite{lapoutreMaintainance23eccII} (summarized and expanded in Section~\ref{section:improved_datastructures}), and it may be less efficient than the first algorithm in practice.  

In both algorithms we use the same solution for $(1)$: an optimal disjoint set union data structure $DSU_{3ecs}$ \cite{tarjanEfficiencyGoodNot1975}, that operates on $V(G)$, and uses as representatives the leaves of $\mathit{T}$ (which correspond bijectively to the maximal $3$-edge-connected subgraphs of $G$). $DSU_{3ecs}$ supports the operations $\mathit{find_{3ecs}}(x)$ and $\mathit{unite_{3ecs}}(x,y)$. $\mathit{find_{3ecs}}(x)$ returns a pointer to the maximal $3$-edge-connected subgraph of $G$ that contains $x$, and $\mathit{unite_{3ecs}}(x,y)$ unites the sets of vertices of the maximal $3$-edge-connected subgraphs of $G$ that contain $x$ and $y$. $DSU_{3ecs}$ uses $O(n)$ space, and it can perform any sequence of $m$ operations $\mathit{find_{3ecs}}$ and $\mathit{unite_{3ecs}}$ in $O(m\alpha(m,n))$ time in total \cite{tarjanEfficiencyGoodNot1975}. (Thus we get this expression in the time-bounds of both algorithms.) 

Furthermore, for $(5)$ we rely on data structures associated with the nodes of $\mathcal{T}$, that represent trees or cactuses. In particular, recall that (the graph corresponding to) every $1$-ecc node $X$ is a connected component of (the graph corresponding to) its parent, and the children of $X$ correspond to its $2$-edge-connected components. Thus, we associate with every $1$-ecc node a data structure that represents the tree of its $2$-edge-connected components. This data structure is used in order to find the children of $X$ that we have to merge (in Line \ref{alg:tree_path} of Algorithm~\ref{algorithm:tree_modification}), and the inter-connection edges between (the graphs that correspond to) those children (in Line \ref{alg:tree_edges} of Algorithm~\ref{algorithm:tree_modification}). Similarly, every $2$-ecc node $X$ has an associated data structure that represents the cactus of its $3$-edge-connected components. This data structure is used in order to find the children of $X$ that we have to merge (in Lines \ref{alg:cactus_path1} and \ref{alg:cactus_path2} of Algorithm~\ref{algorithm:tree_modification}), and the inter-connection edges between (the graphs that correspond to) those children (in Line \ref{alg:cactus_edges} of Algorithm~\ref{algorithm:merge3ecc}). More details about the operations supported by these data structures will be given in Section~\ref{subsection:log_alg}.

Finally, we note that the $O(n^2)$ expression in the time-bounds of both algorithms is a bottleneck in their total running time. In fact, Lines \ref{alg:re-insert} and \ref{alg:re-insert1} of Algorithm~\ref{algorithm:tree_modification} and Line \ref{alg:re-insert2} of Algorithm~\ref{algorithm:merge3ecc} may create a sequence of recursive calls, so that, even for a single edge insertion, we may have to perform $O(n^2)$ re-insertions of already inserted edges. However, since by Theorem~\ref{theorem_1} there are only $O(n)$ edge insertions that can affect the decomposition tree, and every re-insertion (performed internally by the algorithm) can only affect deeper levels than before, and there can be only $O(n)$ levels in the decomposition tree, we have that any sequence of edge insertions to $G$ can initiate at most $O(n^2)$ calls to procedure $\mathtt{insert}$ of Algorithm~\ref{algorithm:tree_modification}.

\subsection{An $O(n^2\log^2 n + m\alpha(m,n))$-time algorithm for the incremental maintenance of $\mathcal{T}$}
\label{subsection:log_alg}
An efficient solution for $(2)$ is to use the top-trees data structure of Alstrup et al. \cite{DBLP:journals/talg/AlstrupHLT05}. With this data structure we can perform nearest common ancestor queries in dynamic trees with $n$ nodes in amortized $O(\log n)$ time per query. Furthermore, with top-trees we can also answer level ancestor queries within the same time-bounds. Thus we can solve $(3)$ efficiently with a query for the $d+2$ level ancestor of $V$, where $V$ is the leaf of $\mathcal{T}$ that contains $v$, and $d$ is the level of $N$. (We may assume that every node of $\mathcal{T}$ has an attribute for its level on $\mathcal{T}$.)

The merging of nodes in $(4)$ can be performed by redirecting the parents of the children of the nodes with the least number of children to the node with the most number of children (breaking ties arbitrarily), and then discarding these nodes. (We may assume that every node of $\mathcal{T}$ has an attribute for its number of children.) In other words, suppose that we have to merge the nodes $X_1,\dots,X_k$, where $X_1$ has the greatest number of children among $X_1,\dots,X_k$. Then we redirect the parents of the children of $X_2,\dots,X_k$ to $X_1$, and then we discard $X_2,\dots,X_k$ (and all the information associated with them). Since the nodes of every level of $\mathcal{T}$ are $O(n)$, this procedure ensures that at most $O(n\log n)$ redirections can take place in total in every level. However, for every such redirection, a deletion and an insertion of an edge of $\mathcal{T}$ must take place, and this takes $O(\log n)$ amortized time using the top-trees. Thus we get an $O(n\log^2 n)$ time-bound for every level of $\mathcal{T}$, and since there are at most $O(n)$ levels in $\mathcal{T}$, we get the $O(n^2\log^2 n)$ time-bound in total.

Now, for $(5)$, we assume, as above, that every $1$-ecc node has an associated data structure that represents the tree of its $2$-edge-connected components, and every $2$-ecc node has an associated data structure that represents the cactus of its $3$-edge-connected components. We use the associated data structures for the trees of the $2$-edge-connected components in order to be able to perform efficiently the following. In Line~\ref{alg:tree_path} of Algorithm~\ref{algorithm:tree_modification}, we access the associated data structure $T$ of $N$, in order to find its children $X_1,\dots,X_k$ that we have to merge in Line~\ref{alg:merge2}. Thus we assume that the data structure for trees supports the operation $\mathtt{compressPath(T,x,y)}$, which, given a pointer $T$ to a tree, and pointers $x$ and $y$ to nodes of this tree, it finds the simple path $P$ that connects these nodes on the tree, returns pointers to the nodes of $P$ and the edges of $P$, and then merges all nodes of $P$ into a new node and also returns a pointer to this node. We assume that the nodes and the edges of $T$ have some information associated with them. Specifically, since every edge $e$ of $T$ corresponds essentially to a bridge of $N$, and this bridge is an edge $(x,y)$ of $G$, we let $e$ point to the edge $(x,y)$ of $G$. This information is needed for Lines \ref{alg:tree_edges} and \ref{alg:cactus_path1}. Furthermore, since every node $v$ of $T$ corresponds essentially to a child $V$ of $N$, we let $v$ point to $V$. This is used precisely in order to find the children $X_1,\dots,X_k$ of $N$ that we have to merge. Conversely, every child of $N$ must point to its corresponding node of $T$, in order to be able to find the endpoints of the path on $T$ that we have to find with the call to $\mathtt{compressPath}$. 
Finally, since we may have to merge some children of the root or of a $3$-ecc node in Line~\ref{alg:merge1} of Algorithm~\ref{algorithm:tree_modification}, we also have to link the associated (representations of the) trees of those children, with the addition of an edge that corresponds to the one that has been (re-)inserted to $G$. Specifically, suppose that an edge $(x,y)$ is (re-)inserted to $G$, and that the nearest common ancestor $N$ of $X$ and $Y$ is the root or a $3$-ecc node, where $X$ and $Y$ are the leaves of $\mathcal{T}$ that contain $x$ and $y$, respectively. Let $C_1$ and $C_2$ the the children of $N$ that contain $x$ and $y$, respectively. Then, in Line~\ref{alg:merge1}, we have to merge $C_1$ and $C_2$ into a new $1$-ecc node $C$. Now, the data structure associated with $C$ must represent a tree that is formed by linking the tree represented by the associated data structure to $C_1$ with the tree represented by the associated data structure to $C_2$, through the addition of an edge $e=(u,v)$ with endpoints corresponding to the children of $C_1$ and $C_2$ that contain $x$ and $y$, respectively. (These children of $C_1$ and $C_2$ can be found with a query for the $d+2$ level ancestors of $X$ and $Y$, respectively, where $d$ is the level of $N$.) Furthermore, since $e$ corresponds essentially to $(x,y)$, the associated information to $e$ must be a pointer to $(x,y)$. Thus, the data structure for trees must support the operation $\mathtt{joinTrees(T_1,T_2,(u,v))}$, which, given a pointer $T_1$ to a (representation of a) tree $\tilde{T}_1$, a pointer $T_2$ to a (representation of a) tree $\tilde{T}_2$, a pointer $u$ to a node $\tilde{u}$ of $\tilde{T}_1$, a pointer $v$ to a node $\tilde{v}$ of $\tilde{T}_2$, and some extra information to be associated with the edge $(\tilde{u},\tilde{v})$, links the trees $\tilde{T}_1$ and $\tilde{T}_2$ through the addition of the edge $(\tilde{u},\tilde{v})$.

The associated data structures for the cactuses of the $3$-edge-connected components of $2$-ecc nodes are used to perform the analogous operations that are performed by the data structures for trees. Specifically, in Lines \ref{alg:cactus_path1} and \ref{alg:cactus_path2} of Algorithm~\ref{algorithm:tree_modification}, we access the associated data structure for the $3$-edge-connected components of a node $X$ of $\mathcal{T}$, in order to find its children that we have to merge later on (with a call to procedure $\mathtt{merge3ecc}$). Thus we assume that the data structure for cactuses supports the operation $\mathtt{compressCyclePath(S,x,y)}$, which, given a pointer $S$ to a cactus, and pointers $x$ and $y$ to nodes of this cactus, it finds the cycle-path $Q$ on $S$ with endpoints these nodes, it returns pointers to the nodes of $Q$ and the edges of $S$ between the nodes of $Q$, and then merges all nodes of $Q$ into a new node and also returns a pointer to this node. We assume that the nodes and the edges of $S$ have some information associated with them. Specifically, since every edge $e$ of $S$ corresponds essentially to an edge of $X$, and this edge is an edge $(x,y)$ of $G$, we let $e$ point to $(x,y)$. This information is needed for Line \ref{alg:cactus_edges} of Algorithm~\ref{algorithm:merge3ecc}. Furthermore, since every node $v$ of $S$ corresponds essentially to a child $V$ of $X$, we let $v$ point to $V$. This is used precisely in order to find the children of $X$ that we have to merge. Conversely, every child of $X$ must point to its corresponding node of $S$, in order to be able to find the endpoints of the cycle-path on $S$ that we have to find with the call to $\mathtt{compressCyclePath}$. 
Finally, since we may have to merge the children $X_1,\dots,X_k$ of the $1$-ecc node $N$ in Line~\ref{alg:merge2} of Algorithm~\ref{algorithm:tree_modification}, we also have to link the associated (representations of the) cactuses of those children, with the addition of a cycle that corresponds to the one that has been formed by the (re-)insertion of $(x,y)$ to $G$. Specifically, when we merge $X_1,\dots,X_k$ into a new $2$-ecc node $X$, the associated cactus of $X$ must be that which is formed by joining the cactuses associated with $X_1,\dots,X_k$ with a cycle corresponding to $D_1,\dots,D_k$, where $D_i$, for $i\in\{1,\dots,k\}$, is the new $3$-ecc node that has been formed by merging $D_{(i,1)},\dots,D_{(i,t(i))}$ through the call $\mathtt{merge3ecc}$ in Line~\ref{alg:call_merge3}. Furthermore, the edges of the new cactus that correspond to the edges $(D_1,D_2),\dots,(D_{k-1},D_k),(D_k,D_1)$ of $X$, must correspond to the edges $(x_1,y_1),\dots,(x_{k-1},y_{k-1}),(y,x)$. This information is needed in Line~\ref{alg:cactus_edges} of Algorithm~\ref{algorithm:merge3ecc}. Thus, the data structure for cactuses must support the operation $\mathtt{joinCactuses(S_1,\dots,S_k,(d_1,d_2),\dots,(d_k,d_1))}$, which, given pointers $S_1,\dots,S_k$ to (representations of) cactuses $\tilde{S}_1,\dots,\tilde{S}_k$, pointers $d_1,\dots,d_k$ to nodes $\tilde{d}_1,\dots,\tilde{d}_k$ of $\tilde{S}_1,\dots,\tilde{S}_k$, respectively, and some extra information to be associated with the edges $(\tilde{d}_1,\tilde{d}_2),\dots,(\tilde{d}_k,\tilde{d}_1)$, links the cactuses $\tilde{S}_1,\dots,\tilde{S}_k$ through the addition of the cycle $(\tilde{d}_1,\tilde{d}_2),\dots,(\tilde{d}_k,\tilde{d}_1)$.

In Section~\ref{section:structures} we provide efficient implementations for the associated data structures for trees and cactuses. These implementations use size $O(n)$ for a collection of trees or cactuses with $n$ nodes, and can perform any sequence of operations in $O(n\log n)$ time in total. We use a data structure for trees in every $3k+1$ level of the tree, and a data structure for cactuses in every $3k+2$ level of the tree. Since the number of nodes of the trees (of the $2$-eccs) or of the cactuses (of the $3$-eccs) that correspond to the nodes of every level of $\mathcal{T}$ can be at most $n$, and there are at most $O(n)$ levels on $\mathcal{T}$, we thus have that the operations in the associated data structures can take time $O(n^2\log n)$ in total for any sequence of edge insertions to $G$.

Finally, Algorithm~\ref{algorithm:insert_vertex} shows how we can handle the insertion of a new vertex $v$ to $G$. We simply introduce three new nodes $C$, $C'$, and $C''$ (an $1$-ecc node, a $2$-ecc node, and a $3$-ecc node, respectively), that correspond to $v$, and we set $\mathit{parent}(C)\leftarrow \mathit{root}$, $\mathit{parent}(C')\leftarrow C$, and $\mathit{parent}(C'')\leftarrow C'$. Since the leaf of $\mathcal{T}$ that contains $v$ is $C''$, we also set the representative $\mathit{find_{3ecs}}(v)\leftarrow C''$. We also associate a (representation of a) trivial tree to $C$, and a (representation of a) trivial cactus to $C'$. The DSU data structure and the data structures in Section~\ref{section:structures} allow for new node insertions, and the same time-bounds hold (where $n$ is interpreted as the total number of vertices that will have been inserted to $G$ at the moment we estimate the time that it took to perform all the operations of the algorithm so far).

\begin{algorithm}[H]
\caption{\textsf{Update the decomposition tree $\mathcal{T}$ after inserting a new vertex $v$ to $G$}}
\label{algorithm:insert_vertex}
\LinesNumbered
\DontPrintSemicolon
\textbf{procedure} $\mathtt{insert}(v)$\;
\Begin{
let $C$, $C'$ and $C''$ be a new $1$-ecc, $2$-ecc and $3$-ecc node, respectively\;
$\mathit{parent}(C)\leftarrow $ root of $\mathcal{T}$\;
$\mathit{parent}(C')\leftarrow C$\;
$\mathit{parent}(C'')\leftarrow C'$\;
let the representative of $v$ in $\mathit{DSU_{3ecs}}$ be $C''$\;
initialize a new trivial tree $T$, and associate it with $C$\;
initialize a new trivial cactus $S$, and associate it with $C'$\;
}
\end{algorithm}

\subsection{An $O(n^2\alpha(n,n) + m\alpha(m,n))$-time algorithm for the incremental maintenance of $\mathcal{T}$}
\label{subsection:alpha_alg}
Now we will describe a more time-efficient algorithm to handle insertions to the graph. (This algorithm, however, uses $O(n^2)$ space.) First, we can handle the operations in $(5)$ exactly as we did in Section~\ref{subsection:log_alg}, although here we use the more sophisticated data structures for trees and cactuses that are described in Section~\ref{section:improved_datastructures}. The implementations given in Section~\ref{section:improved_datastructures} use size $O(n)$ for a collection of trees or cactuses with $n$ nodes, and can perform any sequence of operations in $O(n\alpha(n,n))$ time in total. Since we use a data structure for trees in every $3k+1$ level, and a data structure for cactuses in every $3k+2$ level, and there are at most $O(n)$ levels in $\mathcal{T}$, we thus we get the $O(n^2\alpha(n,n))$ expression in the total time-bound .

To perform the merging of nodes in $(4)$ we use an optimal disjoint set union data structure $DSU_i$ \cite{tarjanEfficiencyGoodNot1975}, on every level $i$ of $\mathcal{T}$, that operates on the node set of that level and uses as representatives nodes on that level that we consider to be ``active''. The function of an active node is to represent all the nodes that have been merged with it; initially, all nodes of $\mathcal{T}$ are active, and, throughout, the parent pointer on $\mathcal{T}$ exists only for the active nodes. $DSU_i$ supports the operations $\mathit{find_i}(x)$ and $\mathit{unite_i}(x,y)$. $\mathit{find_i}(x)$ returns a pointer to the active node that has been merged with $x$, and $\mathit{unite_i}(x,y)$ merges the sets of nodes that are represented by the active nodes $x$ and $y$, and sets as representative one of $x$ or $y$ (while deactivating the other). Thus, in order to find the parent of an active node $X$ of $\mathcal{T}$ on level $i$, we use $\mathit{find_{i-1}}(\mathit{parent}(X))$. $DSU_i$ uses $O(n)$ space, and it can perform any sequence of $m$ operations $\mathit{find_i}$ and $\mathit{unite_i}$ in $O(m\alpha(m,n))$ time in total \cite{tarjanEfficiencyGoodNot1975}. 

To perform efficiently the operations $(2)$ and $(3)$, we basically maintain, for every inter-connection edge $e=(x,y)$ of $G$ (that is, for every edge $e$ whose endpoints, at the time of its insertion, lie in different maximal $3$-edge-connected subgraphs of $G$), two paths of $\mathcal{T}$ whose nodes contain an endpoint of $e$, that start from the leaves of $\mathcal{T}$ that contain $x$ and $y$, and go up until at least their nearest common ancestor. In order to achieve this, we augment the information associated with $\mathcal{T}$ as follows. For every node $Z$ of $\mathcal{T}$ we associate a linked list $L_Z$, where every element of $L_Z$ corresponds to an endpoint of an inter-connection edge that lies in an ancestor of $Z$. Thus, the existence of those lists means that we may need as much as $O(n^2)$ space (since the number of inter-connection edges can be at most $O(n)$, and the number of levels of $\mathcal{T}$ is $O(n)$). Every element of $L_Z$ is a pointer that points to an element of $L_C$, for some child $C$ of $Z$. Also, every element of $L_Z$ has a pointer to $Z$ (so that we can find in constant time the node in whose associated list this pointer lies). Thus we may say that a pointer of $L_Z$ lies in $Z$. Finally, for every inter-connection edge $e=(x,y)$, we maintain two pointers $e_x$ and $e_y$, that point to the elements of a list $L_Z$, for some node $Z$ that contains both $x$ and $y$.

Now we work as follows. Let $e=(x,y)$ be an edge that is inserted to $G$ for the first time. (I.e., this is the first time that we call $\mathtt{insert}(e)$ of Algorithm~\ref{algorithm:tree_modification}. Notice, of course, that an edge $e'=(x,y)$ may have previously been inserted to $G$, since we allow for multiple edges. But now we consider $e$ as a new insertion. Furthermore, we can easily determine whether $\mathtt{insert}(e)$ is the first time we call $\mathtt{insert}$ on $e$, by using a flag every time we call $\mathtt{insert}$ from Line~\ref{alg:re-insert1} of Algorithm~\ref{algorithm:tree_modification} or Line~\ref{alg:re-insert2} of Algorithm~\ref{alg:merge3}, in order to signify that this is a re-insertion.) First we find the leaves $X$ and $Y$ of $\mathcal{T}$ that contain $x$ and $y$, respectively, by calling $\mathit{find_{3ecs}}(x)$ and $\mathit{find_{3ecs}}(y)$, respectively. Suppose that $X\neq Y$ (for otherwise there is nothing to do). Then we crawl up the tree, starting from $X$ and $Y$, by following the parents alternatingly, marking all the nodes that we meet, and keeping them in a stack in order to unmark them later. When we meet a node $N$ that is already marked, we have that this is the nearest common ancestor of $X$ and $Y$ on $\mathcal{T}$, and we unmark all the nodes that we traversed. Now we start again from $X$, and we follow the parents until we reach $N$. For every node $Z$ that we traverse, including $N$, we append a new pointer $Z_{(e,x)}$ to $L_Z$. If $Z=X$, then we let $Z_{(e,x)}$ point to $\mathit{null}$. Otherwise, let $C$ be the child of $Z$ that we traversed before reaching $Z$; then we let $Z_{(e,x)}$ point to $C_{(e,x)}$. We repeat the same process starting from $Y$. Finally, we let $e_x$ point to $N_{(e,x)}$ and we let $e_y$ point to $N_{(e,y)}$. Observe that this process may need as many as $O(n)$ calls to the parent function, and $O(n)$ pointer manipulation operations. Thus, since the total number of first-time edge insertion in any sequence of edge insertions is $O(n)$, and every call to the parent function involves a call to $\mathit{find_i}$, for some $i$, the total cost of this process in the running time of the algorithm is $O(n^2\alpha(n,n))$.

Now suppose that we re-insert an edge $e=(x,y)$. Then we can find the nearest common ancestor of the leaves $X$ and $Y$ of $\mathcal{T}$ that contain $x$ and $y$, respectively, by following the path of pointers starting from $e_x$ and $e_y$, and moving on to deeper levels in parallel, until we reach two pointers that lie on two nodes $C$ and $D$ that have not got merged together, in which case the common parent of $C$ and $D$ is the nearest common ancestor of $X$ and $Y$ on $\mathcal{T}$. Then, since $C$ and $D$ will have to get merged, we let $e_x$ point to $C_{(e,x)}$ and $e_y$ point to $D_{(e,y)}$, in order to avoid traversing the same path of pointers again. Since the total number of edges that can affect the decomposition tree, in any sequence of edge insertions, is $O(n)$, and every time an edge is re-inserted it is moved into deeper levels, observe that the total cost of this process in the running time of the algorithm is $O(n^2\alpha(n,n))$.  
Now suppose that we want to find the grandchild of $N$ that contains $x$. (We can follow the analogous procedure in order to find the grandchild of $N$ that contains $y$.) Then we can use the pointer $N_{(e,x)}$ stored in $L_N$. This points to a pointer $C_{(e,x)}$ in $L_C$, for a node $C$ that is merged with a child $C'$ of $N$. And then $C_{(e,x)}$ points to $D_{(e,x)}$ in $L_D$, for a node $D$ that is merged with a child $D'$ of $C'$. Thus we can find $D$ in constant time, and then we can find $D'$ with a call $\mathit{find_i}(D)$, where $i$ is the level of $D$. Again, since there can be at most $O(n^2)$ re-insertions of edges, the total cost of this process in the running time of the algorithm is $O(n^2\alpha(n,n))$.  

All the data structures with which $\mathcal{T}$ is equipped (the DSU data structures, and the associated data structures for trees and cactuses), allow for new node insertions, and the same time-bounds hold. Thus, the insertion of new vertices to $G$ is performed as shown in Algorithm~\ref{algorithm:insert_vertex}.

\section{Data structures for trees and cactuses}
\label{section:structures}

To maintain the decomposition tree of the maximal $3$-edge-connected subgraphs, we will need the following two data structures, which we may use as black boxes in the algorithm. They operate on collections of trees and cactuses, respectively. We assume that there is some information associated with every edge of every tree or cactus
\footnote{For our algorithm, this information is a pointer to an edge of the original graph (whose maximal $3$-edge-connected subgraphs we want to maintain).}.

The first is a data structure that works on a collection of trees and supports the following operations. 

\begin{itemize}
\item{$\mathtt{compressPath(T,x,y)}$. Given a tree $T$ and nodes $x$ and $y$ on $T$, this operation finds the simple path $P$ that connects $x$ and $y$ on $T$ and compresses it into a new node $z$ (converting $T$ into a smaller tree). Furthermore, it returns pointers to all nodes and (information associated with the) edges on $P$, as well as to the new node $z$. Finally, if $(u,v)$ was an edge of $T$, prior to the call of this operation, with $u\notin P$ and $v\in P$, then $(u,z)$ (in the new tree) must contain the same information that was previously associated with $(u,v)$.}

\item{$\mathtt{joinTrees(T_1,T_2,(x,y))}$. Given two trees $T_1, T_2$ and nodes $x\in T_1$ and $y\in T_2$, this operation links the trees $T_1$ and $T_2$ by introducing a new edge $(x,y)$.}
\end{itemize}

The second data structure works on a collection of cactuses and supports the following operations. 

\begin{itemize}
\item{$\mathtt{compressCyclePath(S,x,y)}$. Given a cactus $S$ and nodes $x$ and $y$ on $S$, this operation finds the ``cycle-path'' $P$ that connects $x$ and $y$ on $S$. It compresses all nodes on $P$ into a new node $z$ by properly squeezing the involved cycles. Furthermore, it returns pointers to all nodes on $P$, to all (information associated with the) edges that connect pairs of those nodes (if there are any), as well as to the new node $z$. Finally, if $(u,v)$ was an edge of $S$, prior to the call of this operation, with $u\notin P$ and $v\in P$, then $(u,z)$ (in the new cactus) must contain the same information that was previously associated with $(u,v)$.}

\item{$\mathtt{joinCactuses(S_1,\dots,S_k,(x_1,x_2),\dots,(x_k,x_1))}$. Given a collection of cactuses $S_1,\dots,S_k$ and nodes $x_i\in S_i$, $i\in\{1,\dots,k\}$, this operation links the cactuses $S_1,\dots,S_k$ into a larger cactus by introducing the edges $(x_1,x_2),\dots,(x_k,x_1)$ (which induce a new cycle).}
\end{itemize}

In Sections~\ref{subsection:treedatastructure} and ~\ref{subsection:cactusdatastructure}, we provide implementations for data structures maintaining trees and cactuses where we can perform any sequence of operations in $O(n\log n)$ time, assuming that the total number of nodes is at most $n$. To achieve this time-bound, we will draw ideas from \cite{galilMaintaining3EdgeConnectedComponents1993}, \cite{lapoutreMaintainance23ecc}, and \cite{westbrookMaintainingBridgeConnectedBiconnected1992}. 
Using these datastructures, we spend an additional $O(n^2\log n)$ time for the incremental maintenance of the decomposition tree.
To see this, recall that we use representations for the trees of the $2$-edge-connected components at levels $3k+1$ of the decomposition tree (where $0\leq k<n$), and representations for the cactuses of the $3$-edge-connected components at levels $3k+2$ (where $0\leq k<n$). Then, every operation $\mathtt{compressPath}$, $\mathtt{joinTrees}$, $\mathtt{compressCyclePath}$ and $\mathtt{joinCactuses}$ is performed on structures on the same level. Since there are $O(n)$ levels on which these operations may be performed, and the total number of nodes at any level is $O(n)$, we get the $O(n^2\log n)$ time bound.

In Section~\ref{section:improved_datastructures}, we use the more sophisticated implementations for those data structures provided by \cite{lapoutreMaintainance23eccII}, in order to get a $O(n\alpha(n,n))$-time bound for performing any sequence of operations on trees or cactuses (on the same level), which gives $O(n^2\alpha(n,n))$ in the total time bound of the algorithm for maintaining the decomposition tree.

\subsection{An implementation for trees}\label{subsection:treedatastructure}

\textbf{(a) Representation}

An idea is to represent the trees as rooted trees. For every node $x$ we let $p(x)$ denote its parent (if $x$ is the root, we leave $p(x)$ undefined). Also, every node $x$ has a pointer $x.\mathit{edge}$ to the edge $(x,p(x))$ that connects $x$ with its parent. We use a DSU data structure to perform mergings of nodes. The DSU data structure supports the operations $\mathit{find}(x)$ and $\mathit{unite}(x,y)$. Thus, in order to access the parent of node $x$ we actually use $\mathit{find}(p(x))$ (although we will not write this explicitly in what follows).\\

\noindent 
\textbf{(b) Operations}
 
We execute $\mathit{compressPath}(T,x,y)$ as follows. We start from $x$ and $y$, and we alternatingly climb up the tree following the parents, marking all nodes that we traverse, until we meet an already marked node $z$. (This ensures that we traverse at most $2d-1$ nodes, where $d$ is the number of nodes of the simple path connecting $x$ and $y$.) Then there are two possibilities: either $(1)$ $z=x$ or $z=y$, or $(2)$ $z$ is neither $x$ nor $y$. In the first case, assuming w.l.o.g. that $z=y$, the path that we have to compress consists of the nodes $x,p(x),p(p(x)),\dots,y$. In the second case, the path consists of two parts: $x,p(x),\dots,z$ and $y,p(y),\dots,z$. In case $(1)$ we work as follows. We start from $x$; we retrieve the edge $(x,p(x))$ (using the pointer stored in $x$), we call $\mathit{unite}(x,p(x))$, setting as representative of the new set $p(x)$, and we repeat, until we reach $z=y$. In case $(2)$ we work as in case $(1)$ twice, starting from both $x$ and $y$, until we reach $z$ each time.

The $\mathit{joinTrees}(T_1,T_2,(x,y))$ operation works as follows. First, we determine which of $T_1,T_2$ is smaller (w.r.t. its number of nodes). This can be done easily if we maintain for every tree an attribute $\mathit{size}$ indicating the number of its nodes. Assume w.l.o.g. that the smallest tree is $T_1$ (or that its size is equal to that of $T_2$). Then we reroot $T_1$ at $x$, and we set $p(x)\leftarrow y$, $x.\mathit{edge}\leftarrow (x,y)$, and $T_2.\mathit{size}\leftarrow T_2.\mathit{size}+T_1.\mathit{size}$. The rerooting of $T_1$ at $x$ is performed as follows. First we find the path $P$ connecting $x$ with the root $r$ of $T_1$, by following the parents. Then we process the nodes of $P$ in reverse order, starting from $r$ until we reach $x$ (which we do not process), and we simply hand over the information of every child to its parent. To be precise, let $v$ be a node that we process, and let $u$ be the node on $P$ that has $p(u)=v$. Then we set $p(v)\leftarrow u$ and $v.\mathit{edge}\leftarrow u.\mathit{edge}$.\\

\noindent 
\textbf{(c) Analysis}

For convenience in the time analysis, we may use a simple implementation for the DSU data structure which performs all unions in $O(n\log n)$ time and every $\mathit{find}$ in $O(1)$ worst-case time. Now, the total time it takes to perform any sequence of $\mathit{compressPath}$ operations is $O(n\log n)$. This is because, every time $\mathit{compressPath}(T,x,y)$ is performed, it takes $O(|P|)$ time to find the simple path $P$ that connects $x$ and $y$, using $O(|P|)$ calls to the $\mathit{find}$ operation. Then we merge the nodes of $P$ (thereby deleting them essentially), so they will not be traversed again. Thus, the total time to find all paths, in all $\mathit{compressPath}$ calls, is $O(n)$. Therefore, the time bound for all $\mathit{compressPath}$ operations is dominated by that for performing the mergings (using the $\mathit{union}$ operation), and thus it is $O(n\log n)$.  

When we perform a $\mathit{joinTrees}$ operation on two trees, we may have to reroot the smallest tree $T$, which means that we may have to traverse all of its nodes. This takes $O(|T|)$ time and uses $O(|T|)$ calls to $\mathit{find}$. However, a node can only be accessed during a rerooting (that is, it may lie on the smallest tree) at most $O(\log n)$ times (because the trees are subsequently joined to produce a bigger tree whose size is at least twice that of the smallest). Thus, we need at most $O(n\log n)$ steps to perform all rerootings, using at most $O(n\log n)$ calls to the $\mathit{find}$ operation. Thus, the total time, for any sequence of $\mathit{compressPath}$ and $\mathit{joinTrees}$ operations, is at most $O(n\log n)$.

\subsection{An implementation for cactuses}\label{subsection:cactusdatastructure}

\textbf{(a) Representation}

In order to efficiently determine the cycle-path connecting two nodes of a cactus, we represent the cactus as a rooted tree with ``real nodes'' and ``cycle nodes''. Every real node corresponds to a cactus node, and every cycle node corresponds to a cycle. Thus, for every cactus $S$, there is a \emph{cactus tree} $T$ on the set of real and cycle nodes corresponding to $S$, which is formed by adding edges corresponding to the incidence relation of the real nodes to the cycle nodes. Specifically, there is an edge $(x, C)$ in $T$, where $x$ is a real node and $C$ is a cycle node, if and only if the cactus node corresponding to $x$ lies on the cactus cycle corresponding to $C$. We root this tree on an arbitrary real node. Thus, the parent of a real node (different from the root) is a cycle node, and the parent of a cycle node is a real node. The parent of a (real or cycle) node $v$ is denoted as $p(v)$. We use a DSU data structure to perform mergings of real nodes. This data structure supports $\mathit{find}(x)$ and $\mathit{unite}(x,y)$,  Thus, in order to access the parent of a cycle node $C$, we use $\mathit{find}(p(C))$.

In order to retrieve the information associated with the cactus edges and to efficiently perform the ``cycle-squeezing'' operation, we use doubly-linked circular lists to represent the cycles of the cactuses. Specifically, for every cycle $C$ of a cactus, there is a doubly-linked circular list $L$ consisting of as many nodes as there are in $C$. The nodes of $L$ correspond to the nodes of $C$ and are ordered accordingly. Every node $x$ of $L$ has pointers $x.\mathit{left}$ and $x.\mathit{right}$ to its neighboring nodes, and pointers $x.\mathit{leftEdge}$ and $x.\mathit{rightEdge}$ to the respective cactus edges. 

Finally, we have the following correspondence between the nodes of the cactus trees and the nodes of the circular lists. Let $x$ be a real node of a cactus tree with parent $C$. Then $C$ corresponds to a cactus cycle and therefore to a circular list $L$. Then we have a pointer for $x$ to its corresponding node on $L$, and vice versa. Furthermore, $p(C)$ also corresponds to a node on $L$, but we do not keep a pointer from $p(C)$ to this node (because it cannot be uniquely determined from $p(C)$, as this can be the parent of many cycle nodes). Instead, we have a pointer from $C$ to the node of $L$ corresponding to $p(C)$. For notational convenience, we may use the same name to denote a real node and its corresponding circular list node. Also, for a cycle node $C$, we may denote the circular list node pointed to by $C$ as $p(C)$.\\

\noindent 
\textbf{(b) Operations}\\

\noindent 
$\mathtt{compressCyclePath}$

To perform $\mathtt{compressCyclePath(S,x,y)}$ we work as follows. First, we determine the cycle-path $P$ connecting $x$ and $y$. To do this, we work on the cactus tree associated with $S$; we start from $x$ and $y$, and we alternatingly climb up the tree following the parents, marking all the nodes that we traverse, until we meet an already marked node $z$. Then there are two possibilities: either $(1)$ $z=x$ or $z=y$, or $(2)$ $z$ is neither $x$ nor $y$. In the first case, assuming w.l.o.g. that $z=x$, $P$ is given by the real nodes on the ascended path starting from $x$. In the second case, $P$ is given by two parts: by the real nodes on the path starting from $x$ and ending in $z$, and by the real nodes on the path starting from $y$ and ending in $z$. (Note that $z$ can be either a cycle node or a real node.) In case $(1)$, assuming w.l.o.g. that $z=x$, let $u_1,C_1,\dots,C_k,u_{k+1}$, where $u_1=x$ and $u_{k+1}=y$, be the cycle-path connecting $x$ and $y$, including the cycles nodes that are connected with every two consecutive real nodes. Then we only have to perform the operations $\mathtt{squeezeCycle}(u_i,u_{i+1},C_i)$, for every $i\in\{1,\dots,k\}$, successively. In case $(2)$, we work as in case $(1)$ for every one of the two parts that form $P$. Furthermore, in the case that $z$ is a cycle node, we also have to perform $\mathtt{squeezeCycle}(u,v,z)$, where $u$ and $v$ are the second-last nodes of the two parts that form $P$. In particular, let $u_1,C_1,\dots,u_k,z$ and $v_1,C'_1,\dots,v_l,z$, where $u_1=x$ and $v_1=y$, be the two paths that we ascended, starting from $x$ and $y$, respectively. Then, in addition to all operations $\mathtt{squeezeCycle}(u_i,u_{i+1},C_i)$, for every $i\in\{1,\dots,k-1\}$, and $\mathtt{squeezeCycle}(v_i,v_{i+1},C'_i)$, for every $i\in\{1,\dots,l-1\}$, we also have to perform $\mathtt{squeezeCycle}(u_k,v_l,z)$.

Now let us describe in detail the operation $\mathtt{squeezeCycle}(u,v,C)$, where $u$ and $v$ are real nodes connected with a cycle node $C$. (Intuitively, we want this operation to ``squeeze'' the cycle $C$ by merging its two nodes $u$ and $v$.) We distinguish two cases, depending on whether $(i)$ $u$ and $v$ are related as ancestor and descendant, or $(ii)$ $u$ and $v$ are siblings on the rooted cactus tree (the later occurs in case $(2)$ above, when $z$ is a cycle node).\\

\noindent
\textit{$(i)$ $u$ and $v$ are related as ancestor and descendant}

We may assume w.l.o.g. that $v$ is an ancestor of $u$. Thus we have $p(u)=C$ and $p(C)=v$. We distinguish three cases, depending on whether $u$ has no siblings, or $u$ is the leftmost or rightmost child of $C$ (to be defined shortly), or neither of the previous two cases holds.
First, if $u$ has no siblings, then we simply return the two edges that connect $u$ and $v$, we discard the cycle $C$, and we merge $u$ and $v$, setting as representative of the union $v$. To return the two edges we use the pointer from $u$ to the node in the circular list corresponding to $C$ (let us call it $u$ again). Then the two edges are precisely $u.\mathit{leftEdge}$ and $u.\mathit{rightEdge}$. 
Now let us consider the case that $u$ is the leftmost or rightmost child of $C$. This is determined by checking whether $u.\mathit{left}$ or $u.\mathit{right}$, respectively, is the node pointed to by $C$ (which essentially corresponds to $v$). So let us assume that $u.\mathit{left}=v$ (the case that $u.\mathit{right}=v$ is treated analogously). The first thing to do is to return the edge that connects $u$ and $v$. This is simply $u.\mathit{leftEdge}$. Then we discard the node $u$ from the circular list and we update it accordingly. That is, we set $u.\mathit{right}.\mathit{left}\leftarrow v$, $v.\mathit{right}\leftarrow u.\mathit{right}$, and $v.\mathit{rightEdge}\leftarrow u.\mathit{rightEdge}$. Finally, we merge $u$ and $v$ (the real nodes on the cactus tree), setting as representative of the union $v$.

Now we consider the third case. That is, we have that $u.\mathit{left}\neq v$ and $u.\mathit{right}\neq v$. (Notice that, in this case, there is no cactus edge connecting $u$ and $v$.) Now the first thing to do is to determine the smallest segment of the circular list $L$ corresponding to $C$ that contains $u$ and $v$. To achieve this efficiently, we start from $u$ and $v$, and we alternatingly follow the $\mathit{left}$ pointers, marking all the nodes that we traverse, until we meet an already marked node $z$. Then we have that either $z=v$ or $z=u$. Let us assume that $z=v$ (the other case is treated analogously). Now we update the cactus tree as follows. Let $P$ be the internal path on $L$ connecting $u$ and $v$ following the $\mathit{left}$ pointers starting from $u$ (that is, $P=u.\mathit{left},u.\mathit{left}.\mathit{left},\dots,v.\mathit{right}$). We introduce a new cycle node $C'$, and we let $C'$ be the parent of all the real nodes that correspond to nodes of $P$. The parent of $C'$ is set to be $v$, and we introduce a new circular list node $\tilde{v}$, corresponding to $v$ on the cycle $C'$. Thus, there is a pointer from $C'$ to $\tilde{v}$. Now we detach $P$ from $L$ and we attach it appropriately to the circular list corresponding to $C'$. This is done by setting $v.\mathit{right}.\mathit{left}\leftarrow \tilde{v}$, $\tilde{v}.\mathit{right}\leftarrow v.\mathit{right}$, $u.\mathit{left}.\mathit{right}\leftarrow \tilde{v}$, and $\tilde{v}.\mathit{left}\leftarrow u.\mathit{left}$. To maintain the cactus edges, we also set $\tilde{v}.\mathit{rightEdge}\leftarrow v.\mathit{rightEdge}$ and $\tilde{v}.\mathit{leftEdge}\leftarrow u.\mathit{leftEdge}$. In the list $L$ we simply discard $u$ and we update $L$ accordingly. That is, we set $u.\mathit{right}.\mathit{left}\leftarrow v$ and $v.\mathit{right}\leftarrow u.\mathit{right}$. Again, to maintain the cactus edge, we also set $v.\mathit{rightEdge}\leftarrow u.\mathit{rightEdge}$. Finally, we merge $u$ and $v$ (the real nodes on the cactus tree), setting as representative of the union $v$.
\\

\noindent
\textit{$(ii)$ $u$ and $v$ are related as siblings}

Let us assume first that (the circular list nodes corresponding to) $u$ and $v$ are neighbors. This means that either $u.\mathit{left}=v$ or $u.\mathit{right}=v$. Suppose that $u.\mathit{left}=v$ (the other case is treated similarly). First we return $u.\mathit{leftEdge}$: the edge that connects $u$ and $v$. Then we discard $u$ from the circular list and we update it accordingly. That is, we set $v.\mathit{right}\leftarrow u.\mathit{right}$, $u.\mathit{right}.\mathit{left}\leftarrow v$, and $v.\mathit{rightEdge}\leftarrow u.\mathit{rightEdge}$. Finally, we merge $u$ and $v$ (on the cactus tree), setting as representative of the union $v$.

Now suppose that $u$ and $v$ are not neighbors. This means that $u.\mathit{left}\neq v$ and $u.\mathit{right}\neq v$. (Notice that, in this case, there is no cactus edge connecting $u$ and $v$.) Let $C$ be the common parent of $u$ and $v$, and let $w=p(C)$. Let also $L$ be the circular list corresponding to $C$, and denote as $w$ the node of $L$ pointed to by $C$. Now the first thing to do is to determine the smallest part of $L$ that contains $u$ and $v$. To achieve this efficiently, we start from $u$ and $v$, and we alternatingly follow the $\mathit{left}$ pointers, marking all the nodes that we traverse, until we meet an already marked node $z$. Then we have that either $z=v$ or $z=u$. Let us assume that $z=v$ (the other case is treated analogously). Let $P$ be the internal path on $L$ connecting $u$ and $v$ following the $\mathit{left}$ pointers starting from $u$ (that is, $P=u.\mathit{left},u.\mathit{left}.\mathit{left},\dots,v.\mathit{right}$). Here we have that either $w\notin P$ or $w\in P$. If $w\notin P$, then we work precisely as in the second paragraph of case $(ii)$ above. So let us assume that $w\in P$. Then we introduce a new cycle node $C'$, we let $C'$ be the parent of all the real nodes that correspond to nodes of $P\setminus{w}$, and we let $C'$ point to $w$. Then we introduce a new circular list node $\tilde{v}$ (corresponding to the node that will be formed by merging $u$ and $v$), we let $C$ point to $\tilde{v}$, and we attach the segment $L\setminus(P\cup\{u,v\})$ of $L$ to $\tilde{v}$. Again, this is done precisely as in the second paragraph of case $(ii)$ above. Furthermore, we discard the circular list node $u$, and we properly update the list corresponding to $C'$. That is, we set $v.\mathit{right}\leftarrow u.\mathit{right}$, $u.\mathit{right}.\mathit{left}\leftarrow v$, and $v.\mathit{rightEdge}\leftarrow u.\mathit{rightEdge}$. Finally, we let the parent of $C$ be $v$, and we merge $u$ and $v$ (the real nodes on the cactus tree), setting as representative of the union $v$.
\\

\noindent 
$\mathtt{joinCactuses}$

To perform $\mathtt{joinCactuses(S_1,\dots,S_k,(x_1,x_2),\dots,(x_k,x_1))}$, we first have to find the largest (w.r.t. its number of nodes) cactus among $S_1,\dots,S_k$. This can be done easily if we keep for every cactus an attribute $\mathit{size}$, signifying its number of nodes. Now let us assume, w.l.o.g., that one of the cactuses with the greatest size is $S_k$. Then we have to do two things: to reroot the cactuses $S_1,\dots,S_{k-1}$ on their nodes that are to be connected on a new cycle, and then to form the cycle $(x_1,x_2),\dots,(x_k,x_1)$. 

Let us describe how to perform a rerooting of a cactus $S$ at a node $x\in S$. We assume that $x$ is not the root of the cactus tree of $S$ (for otherwise there is nothing to do). First we find the path on the cactus tree that connects $x$ with the root by following the parents. This has the form $u_1,C_1,\dots,C_{t-1},u_t$, where $u_1,\dots,u_t$ are real nodes with $u_1=x$, and $C_1,\dots,C_{t-1}$ are cycle nodes. Then we process this path in reverse order, and for every triple $(u,C,v)$, where $u,v$ are real nodes and $C$ is a cycle node, we do the following. (Due to reverse processing, notice that we have $p(v)=C$ and $p(C)=u$.) Let $L$ be the circular list corresponding to $C$. Then we simply change the pointer of $C$, so that it points to the node of $L$ corresponding to $v$, and we reverse the parent relation between the nodes of $(u,C,v)$. (Thus, we set $p(u)\leftarrow C$ and $p(C)\leftarrow v$.) Finally, we update the status of $x$ so that it is recognized as the root.

Thus we may assume that the cactus tree of every $S_i$, $i\in\{1,\dots,k-1\}$, is rooted at $x_i$. Now we introduce a new cycle node $C$, and we set the parent of every $x_i$, $i\in\{1,\dots,k-1\}$, to be $C$, and the parent of $C$ to be $x_k$. Then we introduce new circular list nodes $\tilde{x}_1,\dots,\tilde{x}_k$, pointers between $x_i$ and $\tilde{x}_i$ (in both directions), for every $i\in\{1,\dots,k-1\}$, and we let $C$ point to $\tilde{x}_k$. All nodes $\tilde{x}_i$ are linked in a circular structure according to their ordering on the new cycle. That is, we set $\tilde{x}_i.\mathit{left}\leftarrow \tilde{x}_{i-1}$, for $i\in\{2,\dots,k\}$, $\tilde{x}_1.\mathit{left}\leftarrow \tilde{x}_k$, $\tilde{x}_i.\mathit{right}\leftarrow \tilde{x}_{i+1}$, for $i\in\{1,\dots,k-1\}$, and $\tilde{x}_k.\mathit{right}\leftarrow \tilde{x}_1$. Furthermore, we fix the pointers to the newly introduced edges. That is, we set $\tilde{x}_i.\mathit{leftEdge}\leftarrow (x_{i-1},x_i)$, for $i\in\{2,\dots,k\}$, $\tilde{x}_1.\mathit{leftEdge}\leftarrow (x_k,x_1)$, $\tilde{x}_i.\mathit{rightEdge}\leftarrow (x_i,x_{i+1})$, for $i\in\{1,\dots,k-1\}$, and $\tilde{x}_k.\mathit{rightEdge}\leftarrow (x_k,x_1)$. Finally, we set $S_k.\mathit{size}\leftarrow S_1.\mathit{size}+\dots+S_{k-1}.\mathit{size}$, and the description of $\mathtt{joinCactuses(S_1,\dots,S_k,(x_1,x_2),\dots,(x_k,x_1))}$ is complete.
\\

\noindent 
\textbf{(c) Analysis}

First, if we use a simple implementation for the DSU data structure which performs all unions in $O(n\log n)$ time and every $\mathit{find}$ in $O(1)$ worst-case time, we can argue as in the analysis for the data structure for trees, that the total time it takes to find and merge all cycle-path during any sequence of $\mathtt{compressCyclePath}$ operations is $O(n\log n)$. Thus we only have to consider the total time it takes to perform all $\mathtt{squeezeCycle}$ operations. Recall that $\mathtt{squeezeCycle}(u,v,C)$ has to find the smallest segment $P$ of the circular list corresponding to $C$ that contains $u$ and $v$ (all other operations of $\mathtt{squeezeCycle}(u,v,C)$ take $O(1)$ time in total). This is performed in $O(|P|)$ time, and then $P$ forms a new circular list on its own (plus at most one more node). We say that the two new cycles into which $C$ was squeezed are \textit{formed by} $C$. (Thus, we may observe that after any repeated application of $\mathtt{squeezeCycle}$ operations on $C$ and on cycles formed by $C$, we have that all cycles formed by $C$ form a cactus.) We overload our terminology by saying that a node of a cycle formed by $C$ is also a node of $C$. Now we may argue as follows. Fix a cycle $C$ of size $k$ and a node $x\in C$. Then we observe that in any sequence of $\mathtt{squeezeCycle}$ operations on $C$ or on cycles formed by $C$, $x$ can be part of the smallest segment (on a circular list) explored by $\mathtt{squeezeCycle}$ at most $O(\log k)$ times. Thus, the total time to perform any sequence of $\mathtt{squeezeCycle}$ operations on $C$ or on cycles formed by $C$ is $O(k\log k)$. Now observe that the cycles are introduced into the collection of cactuses by the operation $\mathit{joinCactuses}$; and that once a cycle has been introduced, it can only get squeezed to form smaller cactuses. Thus, let $k_1,\dots,k_t$ be the sizes of all cycles that have been introduced in the data structure by the operation $\mathit{joinCactuses}$. Then we have that the total time to perform all $\mathtt{squeezeCycle}$ operations is $O(k_1\log k_1+\dots+k_t\log k_t) = O(k_1\log n+\dots+k_t\log n) = O((k_1+\dots+k_t)\log n) = O(n\log n)$, since $k_i=O(n)$, for every $i\in\{1,\dots,t\}$, and $k_1+\dots+k_t=O(n)$. 

When we perform a $\mathit{joinCactuses}$ operation on some cactuses, we may have to reroot the smallest ones, which means that we may have to access all of their nodes. However, we can argue as in the analysis of the data structure for trees, in order to show that the total time-bound for all rerootings is $O(n\log n)$. All the other operations performed in $\mathit{joinCactuses}$ take time analogous to the number of the cactuses involved. Since the total size of all cycles that we can introduce (while maintaining a collection of cactuses on $n$ nodes) is $O(n)$, this shows that the total time for any sequence of $\mathit{joinCactuses}$ operations is $O(n\log n)$.


\section{Improved data structures for trees and cactuses}
\label{section:improved_datastructures}

Here we provide data structures and algorithms for the operations $\mathtt{compressPath}$, $\mathtt{joinTrees}$, $\mathtt{compressCyclePath}$ and $\mathtt{joinCactuses}$, on collections of trees and cactuses with at most $n$ nodes, with time bounds better than $O(n\log n)$. To be precise, any sequence of $m$ operations on collections of trees or cactuses with at most $n$ nodes can be performed in $O((n+m)\alpha(m,n))$ time in total. The analysis for this time bound is essentially contained in \cite{lapoutreMaintainance23ecc} and \cite{lapoutreMaintainance23eccII}; our own additions in the algorithms of \cite{lapoutreMaintainance23ecc} and \cite{lapoutreMaintainance23eccII} involve only a worst-case $O(1)$ calls to some DSU operations and a worst-case $O(1)$ pointer manipulations, for every operation performed. Also, the arguments that establish correctness are sufficiently contained in the description of the algorithms. 

Now, to achieve this time bound, we will basically use the data structures provided by La Poutré et al. \cite{lapoutreMaintainance23ecc}, \cite{lapoutreMaintainance23eccII}, with some minor additions to suit our purposes. These data structures rely on the so-called ``fractionally rooted trees'' (FRT), introduced in \cite{lapoutreMaintainance23eccII}. In the next section we give a brief overview of the operations supported by FRTs. Then we describe the implementations for collections of trees and cactuses. We only give as many details of the data structures and algorithms of \cite{lapoutreMaintainance23ecc} and \cite{lapoutreMaintainance23eccII} as are needed in order to show where our own additions fit in and establish our results.

\subsection{Fractionally rooted trees}

The FRT data structure operates on a forest $F$ equipped with a partition of its edges such that the classes corresponding to this partition induce subtrees of $F$. Before we describe the operations supported by FRT, we introduce some terminology. We say that a node $x$ of $F$ belongs to an edge class if it is incident to at least one edge of that class. Also, we say that an edge class intersects a path $P$ if $P$ contains at least one edge of that class. Now let $x,y$ be two nodes of a tree of $F$, and let $P$ be the simple path that connects them. We call a node $z$ on $P$ a \textit{boundary node} of $P$ if it is either one of $x,y$ or it is incident to two edges of $P$ which belong to different classes. A \textit{boundary edge set} for a boundary node $z$ on $P$ is a set of ($0$, $1$ or $2$) edges incident to $z$, one from each edge class to which $z$ belongs and which intersects $P$. (We do not demand that the edges in a boundary edge set for $z$ lie on $P$; however, one of their endpoints must be $z$.) A \textit{boundary list} for $P$ is a list consisting of the boundary nodes of $P$, where each boundary node has a sublist that contains a boundary edge set for it on $P$.
Now let $L$ be a list of nodes where each node has a sublist of edges. We say that an edge class \textit{occurs} in $L$, or that $L$ \textit{contains} an edge class, if there is an edge of that class in some sublist of $L$. The edge classes \textit{of} $L$ are precisely those that occur in it. A \textit{joining list} $J$ is a list of nodes with sublists of edges such that the union of the classes occurring in $J$ induces some subtree in $F$. (We note that this is always the case for a boundary list of a path.) In addition, the nodes in $J$ must be the nodes belonging to at least two edge classes occurring in $J$, and the sublist for each node contains an edge for each class in $J$ to which this node belongs. (Thus we have that a boundary list of a path is a joining list if we remove the endpoints of the path from the list.)

Now the operations supported by FRT are the following.

\begin{itemize}
\item{$\mathtt{link}(x,y)$. Let $x,y$ be two nodes lying on different trees of $F$. Then join the two trees by introducing a new edge $(x,y)$ in the FRT data structure.}
\item{$\mathtt{boundary}(x,y)$. Let $x,y$ be two nodes lying on the same tree of $F$. Then return a boundary list for the simple path with endpoints $x$ and $y$.}
\item{$\mathtt{joinclasses}(J)$. Let $J$ be a joining list. Then join all the edge classes of which an
edge occurs in the list.}
\end{itemize}

We say that a call $\mathtt{boundary}(x,y)$ is essential if there are at least two different edge classes that intersect the simple path that connects $x$ and $y$. (Equivalently, a call $\mathtt{boundary}(x,y)$ is essential if the boundary list that it provides contains more than two nodes.)

By \cite{lapoutreMaintainance23eccII}, we have the following result.

\begin{theorem}(Theorem $9.2$ in \cite{lapoutreMaintainance23eccII})
\label{theorem:frt}
There is an implementation of fractionally rooted trees with the following guarantee. Suppose that we start with an empty forest, and we perform $n$ insertions of nodes and $m$ operations $\mathtt{link}$, $\mathtt{boundary}$, and $\mathtt{joinclasses}$, where every essential call of $\mathtt{boundary}$ is immediately followed by $\mathtt{joinclasses}$ on a joining list that contains all the edge classes that occur in the list provided by $\mathtt{boundary}$. Then all these operations can be performed in total $O((n+m)\alpha(m,n))$ time.
\end{theorem}

\subsection{An implementation for trees}

\textbf{(a) Representation}

For the operations $\mathtt{compressPath}$ and $\mathtt{joinTrees}$ on collections of trees, we essentially use the solution of La Poutré for the incremental maintenance of the $2$-edge-connected components in general graphs with asymptotically optimal time complexity \cite{lapoutreMaintainance23eccII}.

Thus we represent the collection of trees $\mathcal{C} $ as a forest $F$, which is implemented as an FRT data structure. Every tree $T\in \mathcal{C}$ corresponds to a tree $T_\mathit{FRT}$ of $F$. The edges of $F$ are partitioned into edge classes, where every edge class induces a subtree of $F$. Some edges of $F$ are called \textit{bridges} and belong to singleton classes which are called \textit{quasi} classes. All other edge classes are called \textit{real}. There is an one-to-one correspondence between the edges of a tree $T\in \mathcal{C}$ and the bridges of $T_\mathit{FRT}$. Thus every edge of $T_\mathit{FRT}$ that is a bridge contains a pointer to its corresponding edge on $T$. Every node $x$ of $T_\mathit{FRT}$ belongs to at most one real class. For every node $x$ that belongs to a real class, we maintain an edge $x_\mathit{assoc}$ of that class. By shrinking the subtrees that are induced by the real classes into single nodes, we get a natural isomorphism between $T$ and $T_\mathit{FRT}$. (The idea here is that the subtrees induced by the real classes correspond to maximal sets of nodes that got merged due to the operation $\mathtt{compressPath}$.) Thus, every node $x\in T$ corresponds to a subset $S$ of $T_\mathit{FRT}$, and we associate $x$ with a node $x_\mathit{FRT}$ of $S$; conversely, every node in $S$ is associated with $x$. To implement the later, we use a DSU data structure $\mathit{DSU}_F$ on the nodes of $F$, where the representatives of the sets are nodes of the trees of $\mathcal{C}$. $\mathit{DSU}_F$ supports the operations $\mathit{find}_F$ and $\mathit{unite}_F$. Thus, in order to find which node of $\mathcal{C}$ corresponds to a node $u$ of $F$, we use $\mathit{find}_F(u)$. $\mathit{DSU}_F$ (and every other DSU data structure in the sequel) is implemented using rooted trees, with path compression and union by size, thus achieving the asymptotically optimal time complexity \cite{tarjanEfficiencyGoodNot1975}. 

We can summarize the above properties as follows.

\begin{property}
\label{property:frt_trees}
Let $x,y$ be two nodes of a tree $T\in\mathcal{C}$, and let $P$ be the simple path on $T$ that connects $x$ and $y$. Let also $\tilde{P}$ be the simple path on $T_\mathit{FRT}$ that connects $x_\mathit{FRT}$ and $y_\mathit{FRT}$. Then there is an one-to-one correspondence between the edges of $P$ and the edges of $\tilde{P}$ that are bridges. This correspondence is compatible with that between the nodes of $T$ and the nodes of $T_\mathit{FRT}$. In other words, for every bridge $(u,v)$ on $\tilde{P}$, there is an edge $(\mathit{find}_F(u),\mathit{find}_F(v))$ on $P$.
\end{property}

\noindent 
\textbf{(b) Operations}

\noindent 
$\mathtt{compressPath}$

Let $x$ and $y$ be two distinct nodes of a tree $T\in\mathcal{C}$, and let $P$ be the simple path on $T$ that connects $x$ and $y$. To perform $\mathit{compressPath}(T,x,y)$ we work as follows. First we call $\mathit{boundary}(x_\mathit{FRT},y_\mathit{FRT})$ to obtain a boundary list $L$ for the simple path $\tilde{P}$ that connects $x_\mathit{FRT}$ and $y_\mathit{FRT}$ on $T_\mathit{FRT}$. According to Property \ref{property:frt_trees}, in order to retrieve the edges of $P$, it is sufficient to identify the edges of $\tilde{P}$ that are bridges. Since bridges belong to singleton (quasi) classes, we have that every bridge of $\tilde{P}$ must lie in the sublist of some node in $L$. Thus we can retrieve the bridges of $\tilde{P}$ by scanning the sublists of $L$. (Incidentally, we note here that the efficiency of the operation $\mathit{boundary}$ lies in the fact that it only computes boundary nodes and corresponding edge classes, without always traversing the entire path.) Furthermore, using again Property \ref{property:frt_trees}, we can retrieve the nodes of $P$ using the operation $\mathit{find}_F$ of $\mathit{DSU}_F$ on the endpoints of every bridge of $\tilde{P}$.

Now, for every node $u$ in $L$ that is incident to a real class, we append $u_\mathit{assoc}$ to the sublist of $u$. This is to ensure that the real class to which $u$ belongs will get joined to all classes that occur in $L$. (Because there is a possibility that the edge classes to which $u$ belongs and that intersect $\tilde{P}$ are quasi classes, and so $L$ does not contain any edge from the real edge class to which $u$ belongs.) If either $x_\mathit{FRT}$ or $y_\mathit{FRT}$ still has only one edge in its sublist, then it is removed from $L$. (This is to ensure that $L$ is a joining list.) If $L$ is non-empty, then we call $\mathit{joinclasses}(L)$. Otherwise, we have that $x_\mathit{FRT}$ and $y_\mathit{FRT}$ are connected with a bridge $(x_\mathit{FRT},y_\mathit{FRT})$ on $T_\mathit{FRT}$. Then we just change the status of $(x_\mathit{FRT},y_\mathit{FRT})$ so that it is no longer marked as a bridge, and we convert the quasi class that contains $(x_\mathit{FRT},y_\mathit{FRT})$ to a real class. We let the associated edge of every node on $\tilde{P}$ be any edge on $\tilde{P}$. (For this purpose, we may keep in a variable $e_\mathit{temp}$ the edge that was in the boundary edge set of $x_\mathit{FRT}$ after the call $\mathit{boundary}(x_\mathit{FRT},y_\mathit{FRT})$; so now we set $u_\mathit{assoc}\leftarrow e_\mathit{temp}$ for every node $u$ on $\tilde{P}$.) We introduce a new node $z$ on $T$ that replaces the entire path $P$, and we let $z_\mathit{FRT}$ be any node on $\tilde{P}$. Finally, we unite all nodes on $\tilde{P}$ using $\mathit{DSU}_F$, and we let $z$ be the representative of the resulting set. This completes the description of $\mathit{compressPath}$, which is summarized in Algorithm \ref{algorithm:compress_path}.\\

\begin{algorithm}[!h]
\caption{\textsf{$\mathit{compressPath}(T,x,y)$}}
\label{algorithm:compress_path}
\LinesNumbered
\DontPrintSemicolon
$P\leftarrow \emptyset$ \tcp{the set of nodes of the simple path on $T$ that connects $x$ and $y$}
$\mathcal{E}\leftarrow \emptyset$ \tcp{the set of edges of $T$ to be returned}
$L\leftarrow \mathtt{boundary}(x_\mathit{FRT},y_\mathit{FRT})$\;
let $e_\mathit{temp}$ be the edge in the sublist of $x_\mathit{FRT}$ in $L$\;
\ForEach{node $u$ in $L$}{
  \ForEach{edge $e$ in the sublist of $u$}{
    \If{$e$ is a bridge}{
      $\tilde{e}\leftarrow $ edge of $T$ pointed to by $e$\;
      $\mathcal{E}\leftarrow \mathcal{E}\cup\{\tilde{e}\}$\;
      \tcp{let $e=(z_1,z_2)$}
      $x_1\leftarrow \mathit{find}_F(z_1)$\;
      $x_2\leftarrow \mathit{find}_F(z_2)$\;
      $P\leftarrow P\cup\{x_1,x_2\}$\;
    }
  }
  \lIf{$u_\mathit{assoc}\neq\emptyset$}{append $u_\mathit{assoc}$ to the sublist of $u$}
  \lIf{the sublist of $u$ contains only one edge}{remove $u$ from $L$}
}
\lIf{$L\neq\emptyset$}{$\mathtt{joinclasses}(L)$}
\Else{
  unmark $(x_\mathit{FRT},y_\mathit{FRT})$ as a bridge\;
  mark $\{(x_\mathit{FRT},y_\mathit{FRT})\}$ as a real class\;
}
\lForEach{node $u$ in $L\cup\{x_\mathit{FRT},y_\mathit{FRT}\}$}{set $u_\mathit{assoc}\leftarrow e_\mathit{temp}$}
let $z$ be a new node on $T$ that replaces the path $P$\;
unite all nodes in $L\cup\{x_\mathit{FRT},y_\mathit{FRT}\}$ using $\mathit{DSU}_F$; let $z$ be the representative\;
set $z_\mathit{FRT}\leftarrow x_\mathit{FRT}$\;
\textbf{return} $\{P,\mathcal{E},z\}$\;
\end{algorithm}

\noindent 
$\mathtt{joinTrees}$

Let $T_1$ and $T_2$ be two distinct trees of $\mathcal{C}$, and let $x\in T_1$ and $y\in T_2$. Then $\mathit{joinTrees}(T_1,T_2,(x,y))$ is performed by simply calling $\mathit{link}(x_\mathit{FRT},y_\mathit{FRT})$, where $(x_\mathit{FRT},y_\mathit{FRT})$ is marked as a bridge and the corresponding singleton edge class is marked as a quasi class. Finally we let $(x_\mathit{FRT},y_\mathit{FRT})$ point to $(x,y)$. (See Algorithm~\ref{algorithm:join_trees}.)\\

\begin{algorithm}[!h]
\caption{\textsf{$\mathit{joinTrees}(T_1,T_2,(x,y))$}}
\label{algorithm:join_trees}
\LinesNumbered
\DontPrintSemicolon
$\mathtt{link}(x_\mathit{FRT},y_\mathit{FRT})$\;
mark $(x_\mathit{FRT},y_\mathit{FRT})$ as a bridge\;
mark $\{(x_\mathit{FRT},y_\mathit{FRT})\}$ as a quasi class\;
make $(x_\mathit{FRT},y_\mathit{FRT})$ point to $(x,y)$\;
\end{algorithm}

\subsection{An implementation for cactuses}

\textbf{(a) Representation}

Following \cite{lapoutreMaintainance23eccII}, we represent the cactuses with a data structure that generalizes the concept of the tree of cycles. Recall that the tree of cycles of a cactus is the graph that represents the incidence relation of the nodes of the cactus to its cycles \cite{galilMaintaining3EdgeConnectedComponents1993}, \cite{lapoutreMaintainance23ecc}. By the definition of the cactus, we have that this graph is a tree. Now we generalize the tree of cycles as follows.

\begin{itemize}
\item We partition the cactus $S$ into subcactuses. Based on this partition, we have the graph $T$ that represents the incidence relation of the nodes of $S$ to the subcactuses. (Again, we have that $T$ is a tree.)
\item We allow $T$ to be extended with extra nodes, that serve as copies of the nodes of $S$. However, we demand that the resulting graph is also a tree. All nodes of $T$ that serve as copies of the same node of $S$ are said to belong to the same \textit{cluster}. (Thus, the set of clusters is a partition of the set of nodes of $T$ that correspond to nodes of $S$.)
\end{itemize}

We call this type of representation of a cactus $S$ a \textit{tree of cactuses}. See Figure \ref{figure:cactus_tree} for an example of this kind of representation. (We note that, contrary to the tree of cycles, a tree-of-cactuses representation of a cactus is not unique.) We distinguish two types of nodes in a tree of cactuses: those that correspond to the nodes of $S$, and they are called \textit{real nodes}, and those that correspond to the subcactuses of $S$, and they are called \textit{cactus nodes}. Thus we have a correspondence between the nodes of $S$ and the real nodes of a tree-of-cactuses representation $T$ of $S$. Furthermore, we may speak of the \textit{nodes} of a cactus node $C$, and by that we mean the nodes of the subcactus of $S$ that corresponds to $C$.

\begin{figure}[t!]\centering
\includegraphics[trim={4cm 8cm 4cm 0cm}, width=0.6\linewidth]{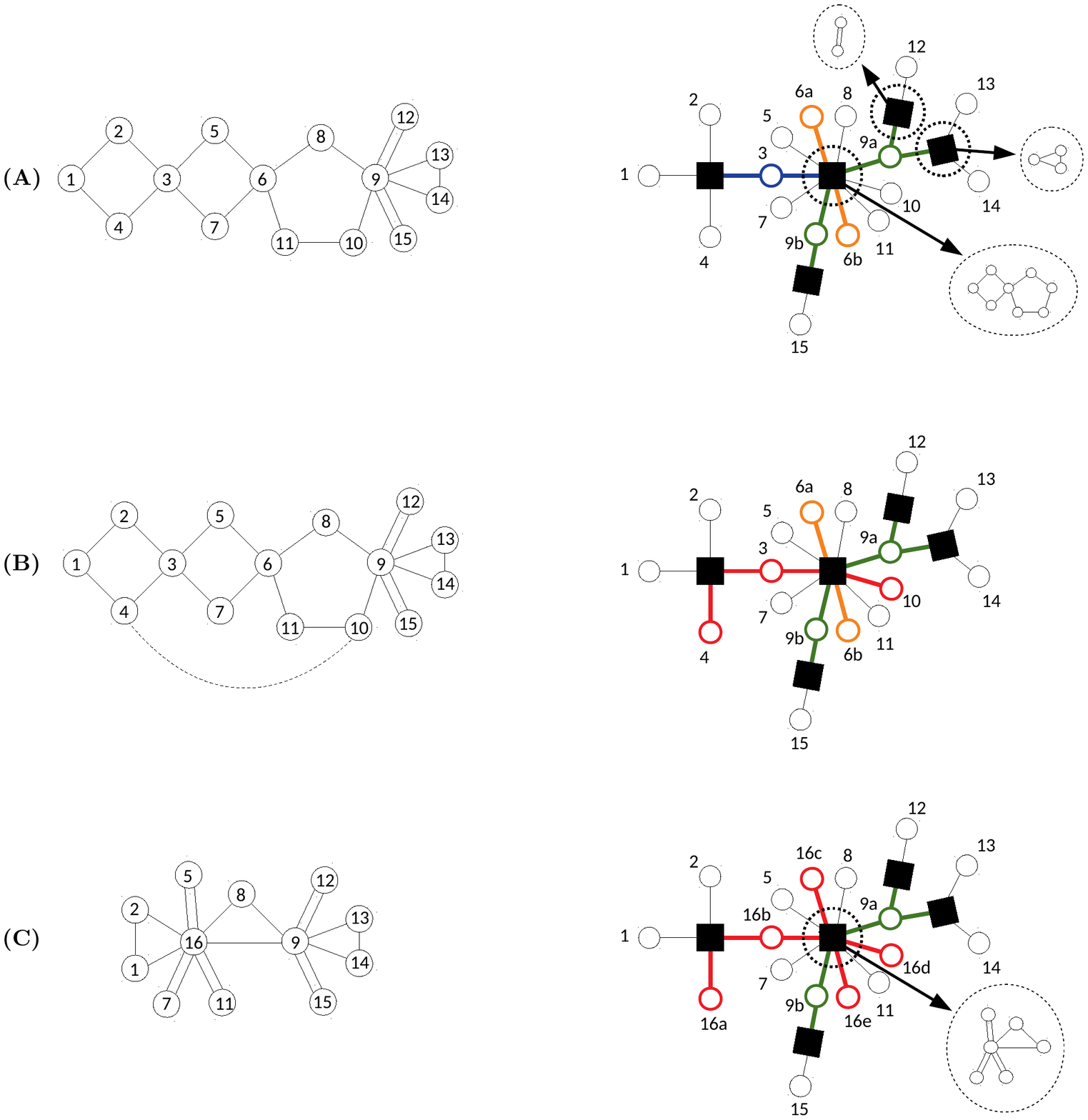}
\caption{\small{An example of a tree-of-cactuses representation of a cactus. The square nodes are the cactus nodes. In figures A and C we can see the corresponding subcactuses of some cactus nodes. The edges that belong to the same edge class are painted with the same colour. The edges in black colour constitute singleton edge classes. We can see that some real nodes correspond to the same nodes of the cactus (and thus they belong to the same cluster). For example, real nodes $6a$ and $6b$, in figure A, correspond to node $6$ of the cactus. In figure B we can see an intermediate step of the call $\mathtt{compressCyclePath}(4,10)$: this is where we have to determine the path on the tree that connects $4$ and $10$. In figure C we can see the result of this call. The nodes $4$, $3$, $6$ and $10$ were merged into a new node, which is called $16$. Observe that, in order to determine that $6$ was also part of the cycle path connecting $4$ and $10$, we were guided by the associated data structure of an intermediary cactus node. Notice also that the operation $\mathtt{compressCyclePath}$ does not affect the structure of the tree of cactuses, but only its edge classes, the names of its real nodes, and the associated data structures of the involved cactus nodes.}}
\label{figure:cactus_tree}
\end{figure}

In our implementation, every node of $S$ has a pointer to a unique corresponding node on $T$; whenever no confusion arises, we will denote these two nodes with the same symbol. Conversely, every node of $T$ has a pointer to its corresponding node on $S$. Specifically, to maintain the later correspondence we use a DSU data structure $\mathit{DSU}_\mathit{cls}$ on the real nodes of $T$ whose sets coincide with the clusters. This data structure supports the operations $\mathit{find}_\mathit{cls}$ and $\mathit{unite}_\mathit{cls}$. Every cluster has a representative real node $z$ in it, which can be found with a call $\mathit{find}_\mathit{cls}(x)$ on any node $x$ in this cluster. Then we only need to maintain the pointer of $z$ to its corresponding node on $S$. Finally, every edge $(x,C)$, where $x$ is a real node and $C$ is a cactus node, has pointers to its endpoints $x$ and $C$.
%

Now let $T$ be a tree of cactuses of a cactus $S$. We equip the set of edges of $T$ with the following equivalence relation. If $(x,C_1)$ and $(x,C_2)$ are two edges of $T$, where $x$ is a real node and $C_1,C_2$ are two cactus nodes, then $(x,C_1)$ and $(x,C_2)$ are equivalent. Furthermore, if $(x,C)$ and $(x',C')$ are two edges of $T$, where $x$ and $x'$ belong to the same cluster
and $C,C'$ are cactus nodes, then $(x,C)$ and $(x',C')$ are equivalent. It should be clear that this is indeed an equivalence relation, and that its equivalence classes induce subtrees of $T$. Thus, on a set of trees of cactuses (representing a set of cactuses), each one equipped with its respective equivalence relation, we can perform the operations of an FRT data structure. Furthermore, it should also be clear that every node $x$ (in a non-trivial tree of cactuses) belongs to a unique edge class, and we associate to $x$ an edge $x_\mathit{assoc}$ of that class.

Before we describe how we can use the FRT operations to perform $\mathtt{compressCyclePath}$ and $\mathtt{joinCactuses}$, let us give an overview of what is involved in using a tree-of-cactuses representation $T$ of a cactus $S$ in order to determine cycle paths on $S$. So let $x,y$ be two nodes of $S$. Let also $Q$ be the cycle path on $S$ connecting $x$ and $y$, and let $P = z_1,C_1,\dots,C_{k-1},z_k$ be the simple path on $T$ with endpoints $x$ and $y$, where $z_1=x$, $z_k=y$, and $C_1,\dots,C_{k-1}$ are the intermediary cactus nodes. Then we have that $Q$ consists of (the nodes of $S$ that correspond to) all the real nodes on $P$, plus (the nodes of $S$ that correspond to) some nodes $w_1,\dots,w_t$ that are incident to the cactuses $C_1,\dots,C_{k-1}$. The nodes $w_1,\dots,w_t$ are precisely those that appear on the cycle paths of the subcactuses of $S$ corresponding to $C_1,\dots,C_{k-1}$ that form part of $Q$. Thus, in order to determine $w_1,\dots,w_t$, we are guided by some data structures associated with $C_1,\dots,C_{k-1}$ that represent cactuses. (See also Figure \ref{figure:cactus_tree}.) For this purpose we use the data structures of \cite{lapoutreMaintainance23ecc}, that represent the cactuses as rooted trees of cycles equipped with \textit{circular split-find} structures. (These representations are used to solve the problem of the incremental maintenance of the $3$-edge-connected components of $2$-edge-connected graphs in asymptotically optimal time. We will not give a full exposition of this representation here, but later on we will describe how to augment it in order to facilitate the retrieval of the cactus edges.) Furthermore, in order to retrieve the edges of $S$ that connect any pair of consecutive nodes on $Q$, we use precisely the data structures associated with $C_1,\dots,C_{k-1}$, as the information concerning the cactus edges is stored in the intermediary cactuses.


Now we recall that a circular split-find data structure operates on a collection of circular lists, where every list is related with a representative element, and supports the operations $\mathit{find}$ and $\mathit{split}$. $\mathit{find}(x)$ on a list node $x$ returns the representative of the list containing $x$. $\mathit{split}(x,y)$ on two distinct nodes $x$ and $y$ of the same list creates two new circular lists, consisting of the nodes from $x$ (resp. $y$) and up to - but excluding - $y$ (resp. $x$), following the same direction in both cases (e.g., the left one), and relates each new list with a representative element. \cite{lapoutreDynamicGraphAlgorithms1991} provides an implementation for the circular split-find problem, where we can perform $m$ operations $\mathit{split}$ and $\mathit{find}$ on a collection of circular lists on $n$ element in $O(n+m\alpha(m,n))$ time in total.

In the tree-of-cycles representation of cactuses used by \cite{lapoutreMaintainance23ecc}, the circular lists correspond to the cycles of the cactus. Specifically, let $C$ be a cycle of a cactus. Then, for every node $x$ of $C$, there is an element $\mathit{repr}(x,C)$ that corresponds to the edge $(x,C)$ of the tree of cycles.
Now, the circular list corresponding to $C$ consists precisely of the elements $\mathit{repr}(x,C)$, for every node $x$ of $C$, ordered in the same way in which the nodes of $C$ occur on $C$. 

(We note that, in order to retrieve $x$ or $C$ from an element $\mathit{repr}(x,C)$, we use a $\mathit{find}$ operation on a DSU data structure that operates on the collection of all $\mathit{repr}$ elements. In fact, we use two different DSU data structures: one for retrieving the node $x$ from an element $\mathit{repr}(x,C)$, and one for retrieving $C$ from $\mathit{repr}(x,C)$. However, we will not mention explicitly those DSU data structures, because with an optimal implementation \cite{tarjanEfficiencyGoodNot1975} they do not affect the asymptotical time bounds.)

Now we have to augment the information on the circular lists, in order to be able to return the cactus edges that connect consecutive nodes on cycle paths on a cactus. To achieve this, we simply store four pointers $\mathit{left}$, $\mathit{right}$, $\mathit{leftEdge}$ and $\mathit{rightEdge}$ on every $\mathit{repr}$ element. Specifically, let $C$ be a cycle of the cactus, and suppose that we have a specific orientation of its nodes (that is, a way to determine, given two consecutive nodes $x$ and $y$ on $C$, which one is on the left and which one is on the right). Now let $x$ and $y$ be two consecutive nodes of $C$, where $x$ is on the left of $y$. Then we have $\mathit{repr}(y,C).\mathit{left}=\mathit{repr}(x,C)$, $\mathit{repr}(x,C).\mathit{right}=\mathit{repr}(y,C)$, $\mathit{repr}(y,C).\mathit{leftEdge}=(x,y)$ and $\mathit{repr}(x,C).\mathit{rightEdge}=(x,y)$.

Finally, for every edge $(x,C)$ on a tree of cactuses, where $x$ is a real node and $C$ is a cactus node, we have a pointer from $(x,C)$ to the node in the associated structure of $C$ that corresponds to $x$, and reversely. (To be more precise, in order to maintain the reverse correspondence we use a DSU data structure that operates internally on the nodes of the tree of cycles representing $C$. This is because the operations on this data structure merge nodes into larger sets from time to time; and thus we maintain representatives of those sets, and pointers from those representatives to their corresponding nodes on the tree of cactuses. However, we will not mention explicitly the calls to this DSU data structure in what follows, since they do not affect the asymptotical time complexity overall.)
\\

\noindent 
\textbf{(b) Operations}

\noindent 
$\mathtt{compressCyclePath}$

To perform $\mathtt{compressCyclePath(S,x,y)}$ (with $x\neq y$), we first have to determine the cycle path connecting $x$ and $y$ on $S$. Let us use the same symbols, $x$ and $y$, to denote the real nodes on the tree-of-cactuses representation $T$ of $S$ that correspond to $x$ and $y$, respectively. As we noted in subsection $(a)$, we have to determine two things: first, the simple path $P = z_1,C_1,\dots,C_{k-1},z_k$ on $T$ that connects $x$ and $y$, where $z_1=x$, $z_k=y$, and $C_1,\dots,C_{k-1}$ are the intermediary cactus nodes, and second, the real nodes $w_1,\dots,w_t$ that are incident to some of the cactuses $C_1,\dots,C_{k-1}$ and that have to get merged too with the other nodes. The property that characterizes the nodes $w_1,\dots,w_t$ can be explained as follows. For every edge $(z,C)$ on $P$, where $z$ is a real node and $C$ is a cactus node, we let $\tilde{z}$ denote the node in the associated tree of cycles of $C$ that is pointed to by $(z,C)$ (and corresponds essentially to $z$, and to any other node in the cluster of $z$). Now suppose that, for some $i\in\{1,\dots,k-1\}$, $(z_i,C_i)$ and $(C_i,z_{i+1})$ belong to different edge classes. Then $\tilde{z}_i$ and $\tilde{z}_{i+1}$ are different nodes of (the associated structure to) $C$, and so we have to find and merge the cycle path on $C$ with endpoints $\tilde{z}_i$ and $\tilde{z}_{i+1}$ (and also return the cactus edges that connect consecutive nodes on this path). The nodes on this cycle path are a subset of $w_1,\dots,w_t$; and with this procedure, applied to all $i\in\{1,\dots,k-1\}$, we get precisely all $w_1,\dots,w_t$.  

Now we notice that it is sufficient to determine only the subset of cactuses $C_1,\dots,C_{k-1}$ that consists of all $C_i$, $i\in\{1,\dots,k-1\}$, such that $(z_i,C_i)$ and $(C_i,z_{i+1})$ belong to different edge classes. (This is because, if $(z_i,C_i)$ and $(C_i,z_{i+1})$ belong to the same edge class, then $z_i$ and $z_{i+1}$ correspond to the same node of $S$, and therefore no real nodes incident to $C_i$ will be involved, and no operations on the associated data structure of $C_i$ will have to be performed.) We can determine those cactuses with a call $\mathit{boundary}(x,y)$ on $T$, by discarding the nodes $x$ and $y$ from the resulting boundary list $L$. (Observe that the intermediary real nodes of $P$ will not appear in $L$, since all the edges that they are incident to belong to the same edge class.) Let $D_1,\dots,D_l$ be the nodes that appear in $L$. Then, for every $i\in\{1,\dots,l\}$, there is an $a(i)$ such that $D_i=C_{a(i)}$. Now, for every $i\in\{1,\dots,l\}$, $D_i$ contains a sublist in $L$ consisting of two edges $(z'_{i_1},D_i)$ and $(z'_{i_2},D_i)$, such that $(z'_{i_1},D_i)$ and $(z_{a(i)},C_{a(i)})$ (resp. $(z'_{i_2},D_i)$ and $(C_{a(i)},z_{a(i)+1})$) belong to the same edge class. This implies that $z'_{i_1}$ corresponds to the same node as $z_{a(i)}$, and $z'_{i_2}$ corresponds to the same node as $z_{a(i)+1}$. Thus, the node (in the associated data structure) of $D_i$ pointed to by $(z'_{i_1},D_i)$ (resp. $(z'_{i_2},D_i)$) is precisely the node that is pointed to by $(z_{a(i)},C_{a(i)})$ (resp. $(C_{a(i)},z_{a(i)+1})$), and so we can get $\tilde{z}_{a(i)}$ (resp. $\tilde{z}_{a(i)+1}$) from this pointer.

Now, for every $i\in\{1,\dots,l\}$, we have to do the following things: $(1)$ determine the cycle path on $D_i$ whose nodes we have to merge, $(2)$ find the cactus edges that connect consecutive nodes on this path, $(3)$ find the corresponding nodes on the tree of cactuses $T$, and $(4)$ merge the nodes on this path and properly update the associated data structure of $D_i$. Let $z_1$ and $z_2$ be nodes on $D_i$ that are the endpoints of the cycle path we have to determine. (We recall that $z_1$ and $z_2$ are given by the pointers stored in $(z'_{i_1},D_i)$ and $(z'_{i_2},D_i)$.) We will not provide the details on how to determine this path (i.e., how to perform $(1)$), as these can be found in \cite{lapoutreMaintainance23ecc}. (In particular, we use the procedure $\mathit{TreePath_3}$ in \cite{lapoutreMaintainance23ecc}.) To perform $(2)$, we use the pointers that we introduced in the $\mathit{repr}$ elements of the circular split-find data structures. To be more precise, once we have determined the cycle path that connects $z_1$ and $z_2$, we then have to perform a $\mathit{split}$ operation on every cycle involved in this path. (For the full details of what is involved in this step, see the procedure $\mathit{AdjustCycles}$ in \cite{lapoutreMaintainance23ecc}.) So let $C$ be a cycle that we have to split on nodes $u_1$ and $u_2$. (We have that $u_1$ and $u_2$ are nodes on the cycle path connecting $z_1$ and $z_2$.) Then we check whether $\mathit{repr}(u_1,C).\mathit{left} = \mathit{repr}(u_2,C)$ or $\mathit{repr}(u_1,C).\mathit{right} = \mathit{repr}(u_2,C)$, or both. If either case holds, we have to return $\mathit{repr}(u_1,C).\mathit{leftEdge}$ or $\mathit{repr}(u_1,C).\mathit{rightEdge}$, or both, respectively. For $(3)$ we simply use the pointer of every node on the cycle path to the corresponding edge of $T$. (In particular, notice that for $z_1$ and $z_2$ we will get edges that belong to the same edge classes as $(z'_{i_1},D_i)$ and $(z'_{i_2},D_i)$, respectively. Every other edge $(w,D_i)$ that we get, belongs to a different edge class, and provides one of the extra nodes $w_1,\dots,w_t$ that we have to merge (and which cannot be derived simply from the call $\mathit{boundary}$).) Finally, we will only provide the details to $(4)$ that have to do with the additional information that we have attached to the associated data structure of $C$; for the rest, we refer again to \cite{lapoutreMaintainance23ecc}. We first have to ensure that the pointers of the $\mathit{repr}$ elements of the circular lists are updated correctly after the splittings. So let $C$ be a cycle that we split on nodes $u_1$ and $u_2$. Then $\mathit{repr}(u_1,C)$ and $\mathit{repr}(u_2,C)$ are assigned to different circular lists; let us call them $L_1$ and $L_2$, respectively. Suppose that $L_1$ contains more that one element (for otherwise there is nothing to do for $\mathit{repr}(u_1,C)$). Then, one of $\mathit{repr}(u_1,C).\mathit{left}$, $\mathit{repr}(u_1,C).\mathit{right}$ has been assigned to $L_1$, and the other one has been assigned to $L_2$. Assume w.l.o.g. that $\mathit{repr}(u_1,C).\mathit{left}$ has been assigned to $L_1$ (i.e., the same list that $\mathit{repr}(u_1,C)$ has been assigned to). Then we must set $\mathit{repr}(u_1,C).\mathit{right} \leftarrow \mathit{repr}(u_2,C).\mathit{right}$ and $\mathit{repr}(u_1,C).\mathit{rightEdge} \leftarrow \mathit{repr}(u_2,C).\mathit{rightEdge}$. The other case for $\mathit{repr}(u_1,C)$ and the analogous cases for $\mathit{repr}(u_2,C)$ are treated similarly. 
%
%
To conclude $(4)$, we note that the correspondence between the nodes of (the associated data structure to) $D_i$ and the edges of the tree of cactuses $T$ is maintained, because, although the nodes of the cycle path on $D_i$ that connects $z_1$ and $z_2$ got merged, this was done using a DSU data structure internal to the associated data structure to $D_i$. Thus we only have to make sure that, in order to access the node in an associated tree of cycles that corresponds to an edge of the tree of cactuses, we first perform a $\mathit{find}$ on the node pointed to by this edge, using the internal DSU data structure.
Algorithm \ref{algorithm:update_cactus} shows the operations that are performed in the associated data structure of $D_i$.

\begin{algorithm}[t!]
\caption{\textsf{$\mathtt{updateCactus}(D,z_1,z_2)$}}
\label{algorithm:update_cactus}
\LinesNumbered
\DontPrintSemicolon
\tcp{compress the cycle path on $D$ that connects $z_1$ and $z_2$ using the associated data structure; return the set of cactus edges that connect consecutive nodes on this path, and a list of edges on the tree of cactuses that contains $D$ that correspond to the nodes of this path}
$\mathcal{E}\leftarrow\emptyset$ \tcp{the set of cactus edges to be returned}
$\mathit{edgelist}\leftarrow\emptyset$ \tcp{the list of corresponding edges to be returned}
find the cycle path $P=\{u_1,\dots,u_k\}$ of $D$ that connects $z_1$ and $z_2$ \tcp{for this step we refer to \cite{lapoutreMaintainance23ecc}}
\ForEach{$i\in\{1,\dots,k-1\}$}{
  let $C$ be the cycle that contains $u_i$ and $u_{i+1}$\;
  \If{$\mathit{repr}(u_i,C).\mathit{left}=\mathit{repr}(u_{i+1},C)$}{
    $\mathcal{E}\leftarrow \mathcal{E}\cup\mathit{repr}(u_i,C).\mathit{leftEdge}$\;
  }
  \If{$\mathit{repr}(u_i,C).\mathit{right}=\mathit{repr}(u_{i+1},C)$}{
    $\mathcal{E}\leftarrow \mathcal{E}\cup\mathit{repr}(u_i,C).\mathit{rightEdge}$\;
  }
}
\ForEach{$u\in\{u_1,\dots,u_k\}$}{
  get the edge $(\tilde{u},D)$ that is pointed to by $u$\; 
  $\mathit{edgelist}\leftarrow\mathit{edgelist}\cup\{(\tilde{u},D)\}$\;
}
merge the nodes on $P$ and properly update the data structure \tcp{again, for this step we refer to \cite{lapoutreMaintainance23ecc}}
fix the pointers $\mathit{left}$, $\mathit{right}$, $\mathit{leftEdge}$ and $\mathit{rightEdge}$ of the $\mathit{repr}$ elements, as described in the text\;
\textbf{return} $\{\mathcal{E},\mathit{edgelist}\}$\;
\end{algorithm}

Thus, from the internal operations on the associated structure of every cactus node $D_i$, $i\in\{1,\dots,l\}$, we get a collection of edges $(w_{i,1},D_i),\dots,(w_{i,k_i},D_i)$ that we include in the edge sublist of $D_i$ in $L$. Furthermore, from the edges $\{(w_{i,j},D_i)\mid i\in\{1,\dots,l\}, j\in\{1,\dots,k_i\}\}$ we get the nodes of the cycle path on $S$ that connects $x$ and $y$, by using the pointer from every $(w_{i,j},D_i)$ to $w_{i,j}$, and then the pointer from $\mathit{find}_\mathit{cls}(w_{i,j})$ to the corresponding node of $S$. 

To conclude $\mathtt{compressCyclePath(S,x,y)}$ we call $\mathit{joinclasses}(L)$, we return a pointer to a new node $z$ that takes the place of all the nodes on the cycle path on $S$ connecting $x$ and $y$, we merge all $w_{i,j}$, $i\in\{1,\dots,l\}$, $j\in\{1,\dots,k_i\}$ into a larger cluster (using the $DSU_\mathit{cls}$ data structure), and we let $z$ point to the representative of this cluster, and reversely. The operation $\mathtt{compressCyclePath(S,x,y)}$ is summarized in Algorithm \ref{algorithm:compress_cyclepath}.\\

\begin{algorithm}[t!]
\caption{\textsf{$\mathtt{compressCyclePath(S,x,y)}$}}
\label{algorithm:compress_cyclepath}
\LinesNumbered
\DontPrintSemicolon
$P\leftarrow \emptyset$ \tcp{the set of nodes of the cycle path on $S$ that connects $x$ and $y$}
$\mathcal{E}\leftarrow \emptyset$ \tcp{the set of edges of $S$ to be returned}
$L\leftarrow \mathtt{boundary}(x,y)$\;
remove $x$ and $y$ from $L$\;
\ForEach{node $C$ in $L$}{
  let $(z_1,C)$ and $(z_2,C)$ be the two edges in the sublist of $C$ in $L$\;
  let $\tilde{z}_1$ be the node pointed to by $(z_1,C$) in the associated data structure of $C$\;
  let $\tilde{z}_2$ be the node pointed to by $(z_2,C$) in the associated data structure of $C$\;
  $\{\mathcal{E}_0,\mathit{edgelist}\}\leftarrow \mathtt{updateCactus}(C,\tilde{z}_1,\tilde{z}_2)$\;
  $\mathcal{E}\leftarrow \mathcal{E}\cup\mathcal{E}_0$\;
  \ForEach{edge $(w,C)$ in $\mathit{edgelist}$}{
    get $w$ using the pointer from $(w,C)$\;
    get the node $u$ of $S$ that corresponds to $w$, using the pointer of $\mathit{find}_{cls}(w)$\;
    $P\leftarrow P\cup\{u\}$\;
  }
  append $\mathit{edgelist}$ to the sublist of $C$ in $L$\;
}
$\mathtt{joinclasses}(L)$\;
let $z$ be a new node on $S$ substituting the cycle path $P$\;
merge all the real nodes that appear as endpoints of the edges in the sublists of $L$ using $\mathit{DSU}_{cls}$; let $\tilde{z}$ be the representative\;
make $z$ point to $\tilde{z}$ and $\tilde{z}$ to $z$\;
\textbf{return} $\{P,\mathcal{E},z\}$\;
\end{algorithm}

\vspace{1em}
\noindent 
$\mathtt{joinCactuses}$

To perform $\mathtt{joinCactuses(S_1,\dots,S_k,(x_1,x_2),\dots,(x_k,x_1))}$, we have to link the nodes $x_1,\dots,x_k$ in a new cycle. Thus we introduce a new cactus node $C$ (that will be made to correspond to the new cycle) and the edges $(x_1,C),\dots,(x_k,C)$, by performing $\mathit{link}(x_1,C),\dots,\mathit{link}(x_k,C)$. This links all the trees of cactuses corresponding to $S_1,\dots,S_k$ to $C$, and produces a larger tree of cactuses. For every $x\in\{x_1,\dots,x_k\}$ that has an associated edge $x_\mathit{assoc}$, we put $(x,C)$ in the edge class of $x_\mathit{assoc}$ by a call of $\mathit{joinclasses}$. If an $x\in\{x_1,\dots,x_k\}$ does not have an associated edge, then $(x,C)$ constitutes a new edge class of its own, and the associated edge of $x$ is set to be $x_\mathit{assoc}\leftarrow (x,C)$.

Then we have to construct the associated data structure to the cactus node $C$. This must be a rooted tree of cycles which contains only one cycle $\tilde{C}$ with nodes corresponding to $x_1,\dots,x_k$, in this order. Let $\tilde{x}_1,\dots,\tilde{x}_k$ be the nodes of $\tilde{C}$ that correspond to $x_1,\dots,x_k$, respectively. We root the tree arbitrarily to any one of $\tilde{x}_1,\dots,\tilde{x}_k$. Then we construct the circular list corresponding to $\tilde{C}$. The nodes of this list are the elements $\mathit{repr}(\tilde{x}_1,\tilde{C}),\dots,\mathit{repr}(\tilde{x}_k,\tilde{C})$, in this order. The pointers $\mathit{left}$ and $\mathit{right}$ on the $\mathit{repr}$ elements are easy to fix. Also, we set $\mathit{repr}(\tilde{x}_i,\tilde{C}).\mathit{leftEdge} \leftarrow (x_{i-1},x_i)$, for every $i\in\{2,\dots,k\}$, $\mathit{repr}(\tilde{x}_1,\tilde{C}).\mathit{leftEdge} \leftarrow (x_k,x_1)$, $\mathit{repr}(\tilde{x}_i,\tilde{C}).\mathit{rightEdge} \leftarrow (x_i,x_{i+1})$, for every $i\in\{1,\dots,k-1\}$, and $\mathit{repr}(\tilde{x}_k,\tilde{C}).\mathit{rightEdge} \leftarrow (x_k,x_1)$. (This is to be able to retrieve the real cactus edges.) The construction of the associated data structure to $C$ is shown in Algorithm \ref{algorithm:initialize_cycle}.

Finally, in order to establish the correspondence between the tree of cycles associated with $C$ and the tree of cactuses which contains $C$, for every $i\in\{1,\dots,k\}$ we have a pointer from $(x_i,C)$ to $\tilde{x}_i$, and reversely. (This is to be able to find the node within the cactus $C$ that corresponds to $x_i$, and conversely.) The procedure $\mathit{joinCactuses}$ is summarized in Algorithm \ref{algorithm:join_cactuses}.\\

\begin{algorithm}[!h]
\caption{\textsf{$\mathtt{joinCactuses(S_1,\dots,S_k,(x_1,x_2),\dots,(x_k,x_1))}$}}
\label{algorithm:join_cactuses}
\LinesNumbered
\DontPrintSemicolon
introduce a new cactus node $C$\;
\lForEach{$x\in\{x_1,\dots,x_k\}$}{$\mathtt{link}(x,C)$}
\lForEach{$x\in\{x_1,\dots,x_k\}$}{let $\{(x,C)\}$ constitute a new edge class}
\ForEach{$x\in\{x_1,\dots,x_k\}$}{
  \If{$x_\mathit{assoc}\neq\emptyset$}{
    let $L$ be a singleton list consisting of $x$\;
    let the sublist of $x$ be $\{x_\mathit{assoc},(x,C)\}$\;
    $\mathtt{joinclasses}(L)$\;
  }
  \lElse{$x_\mathit{assoc}\leftarrow (x,C)$}
}
perform $\mathtt{initialize\_cycle}(C,x_1,\dots,x_k,(x_1,x_2),\dots,(x_k,x_1))$, and collect the corresponding nodes $\{\tilde{x}_1,\dots,\tilde{x}_k\}$\;
\ForEach{$i\in\{1,\dots,k\}$}{
  make $(x_i,C)$ point to $\tilde{x}_i$ and $\tilde{x}_i$ point to $(x_i,C)$\;
}
\end{algorithm}

\begin{algorithm}[!h]
\caption{\textsf{$\mathtt{initialize\_cycle(C,x_1,\dots,x_k,(x_1,x_2),\dots,(x_k,x_1))}$}}
\label{algorithm:initialize_cycle}
\LinesNumbered
\DontPrintSemicolon
create nodes $\tilde{x}_1,\dots,\tilde{x}_k$ and $\tilde{C}$\;
create a tree of cycles consisting of the edges $(\tilde{x}_1,\tilde{C}),\dots,(\tilde{x}_k,\tilde{C})$\;
root the tree at $\tilde{x}_1$\;
create the elements $\mathit{repr}(\tilde{x}_1,\tilde{C}),\dots,\mathit{repr}(\tilde{x}_k,\tilde{C})$\;
initialize a circular split-find data structure on the elements $\mathit{repr}(\tilde{x}_1,\tilde{C}),\dots,\mathit{repr}(\tilde{x}_k,\tilde{C})$ (in this order), and associated it with $\tilde{C}$ \tcp{here we refer to \cite{lapoutreMaintainance23ecc}}
fix the pointers $\mathit{left}$, $\mathit{right}$, $\mathit{leftEdge}$ and $\mathit{rightEdge}$ of the $\mathit{repr}$ elements as described in the text\;
\textbf{return} $\{\tilde{x}_1,\dots,\tilde{x}_k\}$\;
\end{algorithm}

\section{Sparse certificates for the maximal $k$-edge-connected subgraphs}
\label{section:sparse_subgraphs}

In this section we discuss constructions of (almost) sparse subgraphs that have the same maximal $k$-edge-connected subgraphs as the original graph. Following the terminology of \cite{opt-dec-conn}, we define a \emph{$k$-certificate} of a graph to be a spanning subgraph that has the same maximal $k$-edge-connected subgraphs as the original graph.\footnote{This definition can be extended, in order to include objects such as a graph $G'$ from which ``we can derive easily'' the maximal $k$-edge-connected subgraphs of the original graph once we have computed those of $G'$. However, the definition we provided here is enough for our purposes.} 

Using results from Benczúr and Karger~\cite{ben-kar}, we show that $(1)$ in linear time we can construct a $k$-certificate of $O(kn\log{n})$ size, and $(2)$ in $O(m\log^2{n})$ time we can construct a $k$-certificate of $O(kn)$ size. This is a result analogous to \cite{Nagamochi}, which provides certificates for the $k$-edge-connected components. In fact, we use the following concept from \cite{Nagamochi} (defined formally in \cite{AlgAspects}): a \emph{forest decomposition} with $t$ forests of a graph $G$ is a collection of forests $F_1,\dots,F_t$, such that $F_1$ is a spanning forest of $G$, and $F_i$ is a spanning forest of $G\setminus(F_1\cup\dots\cup F_{i-1})$, for $i\in\{2,\dots,t\}$.

Let $G$ be a graph. An edge of $G$ whose endpoints lie in different \emph{maximal} $k$-edge-connected subgraphs of $G$ is called a \emph{$k$-interconnection edge} of $G$. Algorithm~\ref{algorithm:k-certificate} describes an algorithm for computing a $k$-certificate of $G$.

\begin{algorithm}[h!]
\caption{\textsf{Compute a certificate for the maximal $k$-edge-connected subgraphs of $G$}}
\label{algorithm:k-certificate}
\LinesNumbered
\DontPrintSemicolon
let $E'$ be a set of edges of $G$ that contains all its $k$-interconnection edges\;
\label{line:compute-k-inter}
compute a forest decomposition $\mathcal{F}$ of $G\setminus{E'}$ with $k$ forests\;
\textbf{return} $\mathcal{F}\cup E'$\;
\end{algorithm}

\begin{lemma}
\label{lemma:alg:k-certificate-correctness}
Algorithm~\ref{algorithm:k-certificate} outputs a certificate for the maximal $k$-edge-connected subgraphs of $G$.
\end{lemma}
\begin{proof}
If we remove all the $k$-interconnection edges from $G$, then the connected components of the resulting graph coincide with the maximal $k$-edge-connected subgraphs of $G$. Let $S$ be a maximal $k$-edge-connected subgraph of $G$ and let $E'$ be a set of edges of $G$ that contains all its $k$-interconnection edges. Then $[S,\bar{S}]=\emptyset$ in $G\setminus E'$. Now let $H$ be the set of all edges of $E'$ that are contained in $S$. Then the set $\mathcal{F}_S$ of all the edges of $\mathcal{F}$ that are contained in $S\setminus{H}$ is a forest decomposition of $S\setminus{H}$ with $k$ forests\footnote{more precisely: it either coincides with $S\setminus{H}$, or it has $k$ forests} (precisely because $S$, and therefore $S\setminus{H}$, is disconnected from the rest of the graph in $G\setminus{E'}$). Thus, we know from the sparsification paper of Eppstein et al.~\cite{sparsification} that $\mathcal{F}_S$ (considered as a graph) is a \emph{strong certificate}\footnote{For any graph property $\mathcal{P}$, and graph $G$, a strong certificate for $G$ is a graph $G'$ on the same vertex set
such that, for any graph $H$, $G\cup H$ has property $\mathcal{P}$ if and only if $G' \cup H$ has the property.} for the $k$-edge-connectivity of $S\setminus{H}$. Thus, $\mathcal{F}_S\cup H$ is a strong certificate for the $k$-edge-connectivity of $(S\setminus{H})\cup H = S$. But $S$ is $k$-edge-connected, and so $\mathcal{F}\cup E'\supseteq \mathcal{F}_S\cup H$ contains enough edges of $S$ to maintain it as a $k$-edge-connected subgraph of $\mathcal{F}\cup E'$. Since $\mathcal{F}\cup E'$ is a subgraph of $G$, we thus have that its maximal $k$-edge-connected subgraphs coincide with those of $G$.
\end{proof}

\begin{corollary}
\label{corollary:k-certificates}
Let $\mathcal{A}$ be an algorithm that, given a graph $G$ with $m$ edges and $n$ vertices, computes in $T(m,n)$ time a subset $E'$ of $E(G)$ with $S(m,n)$ size that contains all its $k$-interconnection edges. Then we can construct a $k$-certificate of $G$ with $O(kn+S(m,n))$ size in $O(m+n+T(m,n))$ time. 
\end{corollary}
\begin{proof}
First we apply algorithm $\mathcal{A}$ in order to compute a subset $E'$ of $E(G)$ with size $S(m,n)$ that contains all its $k$-interconnection edges. This takes $T(m,n)$ time. Then we apply Algorithm~\ref{algorithm:k-certificate}: we remove $E'$ from $G$, we compute a forest decomposition $\mathcal{F}$ of $G\setminus{E'}$ with $k$ forests, and we return $\mathcal{F}\cup E'$. Using \cite{Nagamochi}, the computation of $\mathcal{F}$ takes time $O(m+n)$. The output has size $O(kn+S(m,n))$. Correctness is guaranteed by Lemma~\ref{lemma:alg:k-certificate-correctness}.
\end{proof}

In the following discussion we refer to section 8 of \cite{ben-kar}. (Note that in \cite{ben-kar} the $k$-interconnection edges are called $k$-weak edges, and the maximal $k$-edge-connected subgraphs are called $k$-strong components.) In \cite{ben-kar}, they prove that a forest decomposition with at least $4k\log n$ forests contains all the $k$-interconnection edges of a graph.  Since this decomposition can be computed in linear time using the MA-ordering algorithm of Nagamochi and Ibaraki \cite{Nagamochi}, by Corollary~\ref{corollary:k-certificates} we get result $(1)$: we can construct a $k$-certificate of size $O(kn\log{n})$ in linear time.

Alternatively, \cite{ben-kar} defines the \emph{strength} of an edge $e$, denoted by $k_e$, as the largest $k'$ such that $e$ lies entirely within a maximal $k'$-edge-connected subgraph. Then the $k$-interconnection edges are precisely those whose strength is less than $k$. \cite{ben-kar} provides an $O(m\log^2{n})$-time algorithm that assigns a value $\tilde{k}_e$ to every edge $e$, such that $\tilde{k}_e\leq k_e$ and $\sum_{e\in E}1/{\tilde{k}_e}=O(n)$. Now we consider the set $E'$ of all edges $e$ that have $\tilde{k}_e<k$. Observe that $E'$ contains all the $k$-interconnection edges, since $\tilde{k}_e\leq k_e$ for every edge $e$. Then we have that $O(n)=\sum_{e\in E}1/{\tilde{k}_e}\geq\sum_{e\in E'}1/{\tilde{k}_e}>\sum_{e\in E'}1/{k}=|E'|/k$. Thus, $E'$ has size $O(kn)$. Therefore, by Corollary~\ref{corollary:k-certificates} we get result $(2)$: we can construct a $k$-certificate of size $O(kn)$ in $O(m\log^2{n})$ time.

Since we have a linear-time algorithm for computing a $k$-certificate of $O(kn\log n)$ size, we can apply the $O(k^{O(k)}(m+n\log n)\sqrt{n})$-time algorithm of Chechik et al.~\cite{ChechikHILP17} on this certificate in order to get the following.

\begin{corollary}
\label{corollary:first_algorithm}
There is an $O(m + k^{O(k)} n\sqrt{n}\log n)$-time algorithm for computing the maximal $k$-edge-connected subgraphs of an undirected graph.
\end{corollary}

Note that for constant $k \ge 3$, the time-bound provided by Corollary~\ref{corollary:first_algorithm} is $O(m + n\sqrt{n}\log n)$, which improves the randomized bound of $O(m\log^2{n} + n \sqrt{n} \log{n})$ given by Forster et al.~\cite{forsterComputingTestingSmall2019}.
Similarly, by applying the $O(km\log^2{n} + k^3 n \sqrt{n} \log{n})$-time algorithm of Forster et al.~\cite{forsterComputingTestingSmall2019} on our $k$-certificate we obtain the next result.
\begin{corollary}
\label{corollary:first_algorithm_2}
There is a randomized Las Vegas algorithm for computing the maximal $k$-edge-connected subgraphs of an undirected graph that has $O(m+k^3 n^{3/2} \log{n})$ expected running time.
\end{corollary}

We note that the time-bound provided by Corollary~\ref{corollary:first_algorithm} is an improvement over that provided in Section~\ref{section:computing_kecs}, for constant $k$, since there are less $\log$ factors involved. Furthermore, the algorithm we present here seems to be easier to be implemented, since it relies on the (relatively simple) construction of the forest decomposition, and the algorithm of Chechik et al.~\cite{ChechikHILP17}. However, the result in Section~\ref{section:computing_kecs} is still relevant, because it establishes that any improvement in the time-bounds for a fully dynamic mincut algorithm implies an improved time-bound for computing the maximal $k$-edge-connected subgraphs.

\section{Computing the maximal $k$-edge-connected subgraphs}
\label{section:computing_kecs}
Let $G$ be an undirected multigraph with $m$ edges and $n$ vertices, and let $k>2$ be a fixed integer. We can compute the maximal $k$-edge-connected subgraphs of $G$ by repeatedly removing any $k'$-edge cut from $G$, for $k'<k$, until there are no more $k'$-edge cuts in the graph for $k'<k$. Then the connected components of the resulting graph coincide with the maximal $k$-edge-connected subgraphs of $G$. The best known deterministic algorithm for computing a $k'$-edge cut of an undirected multigraph, for $k'<k$, or concluding that the graph is $k$-edge-connected, is given by Gabow~\cite{edge_connectivity:gabow}, and it runs in $O(m + k^2n\log(n/k))$ time. Thus the total running time of this algorithm is bounded by $O(mn+k^2n^2\log (n/k))$.

We can improve this bound by reducing the time that it takes to successively find a min-cut with $k'$-edges (for $k'<k$). It is a priori reasonable to assume that we can speed up this procedure, since an algorithm that computes a min-cut of a graph may have stored enough information to facilitate the search for further cuts. In fact, we can rely on Thorup's fully dynamic min-cut algorithm
\cite{thorupFullyDynamicMinCut2007}. This algorithm supports edge insertions and deletions in $\widetilde{O}(\sqrt n)$ time per update, and it maintains a min-cut of size up to $k-1$, for any fixed $k$ (polylogarithmic on the number of vertices). It also demands an additional $O(m+n)$ time to initialize the underlying data structure on a sparse certificate of $k$-edge-connectivity for $G$ \cite{Nagamochi}. (Thus, this algorithm maintains dynamically both a sparse certificate for $G$, using sparsification by Eppstein et al. \cite{sparsification}, and a min-cut of size up to $k-1$ of this certificate.) 

Now we can use the fully dynamic min-cut algorithm of \cite{thorupFullyDynamicMinCut2007} on each connected component $C$ of $G$ as follows. As long as there is a $k'$-edge cut $[S,\bar{S}]$ of $C$, for some $k'<k$, we remove the edges of $[S,\bar{S}]$ from $C$, and then we select two arbitrary vertices $x\in S$ and $y\in\bar{S}$, and we reconnect the two resulting components $S$ and $\bar{S}$ by adding $k$ multiple edges between $x$ and $y$. Thus, the reconnection of the two components $S$ and $\bar{S}$ with $k$ multiple edges ensures that all $k'$-edge cuts, for $k'<k$, are maintained in each component after the removal of $[S,\bar{S}]$, and that these are all the $k'$-edge cuts that may appear now in $C$. Eventually all components of $G$ will become $k$-edge-connected. In the meantime, we collect all the edge-cuts that we find, in every connected component of $G$, and in the end we remove all of them from $G$. It should be clear that the connected components of the resulting graph coincide with the maximal $k$-edge-connected subgraphs of $G$. Observe that every search for a $k'$-edge cut, for $k'<k$, on some connected component of $G$, is immediately followed by $k'$ deletions and $k$ insertions of edges. Furthermore, the total number of these cuts is $O(n)$. Thus, the total running time of this algorithm is 
$\widetilde{O}(m+n\sqrt n)$\footnote{To give a lower bound of the dependency of the complexity of this algorithm on $k$ and on the number of $\log$ factors involved, we note that Thorup's algorithm uses a greedy tree packing of the sparse certificate using $k^7\log^4 n$ trees, and it maintains it using Frederickson's dynamic minimum spanning tree algorithm \cite{fredericksonDataStructuresOnLine1985}. Thus, for every $k'$-edge cut that we find, the cost of the deletions and insertions that follow is  at least $\Omega(k^8\sqrt n\log^4 n)$, and the total complexity of the algorithm is at least $\Omega(m+k^8n\sqrt n\log^4 n)$.}.

\section{A fully dynamic algorithm for maximal $k$-edge-connectivity}
\label{section:fully-dynamic}

In this section we describe our fully dynamic algorithm for maintaining information about the maximal $k$-edge-connected subgraphs of a given graph, for any fixed $k\geq 3$. In more detail, we wish to maintain an undirected graph $G=(V,E)$ throughout an intermixed sequence of the following operations:
\vspace{-2mm}
\begin{mylist}{max-$k$-edge$(x,y)$: }

\litem{\emph{insert$(x,y)$:}} Add edge $(x,y)$ to $G$;

\litem{\emph{delete$(x,y)$:}}  Remove edge $(x,y)$ from $G$ (the operation assumes that $(x,y)$ is in $G$);

\litem{\emph{max-$k$-edge$(x,y)$:}} Return \emph{true} if vertices $x$ and $y$ are in the same maximal $k$-edge-connected subgraph of $G$, and \emph{false} otherwise.
\end{mylist}

Before stating our bounds, let us review some simple minded approaches for the problem. 
In Section~\ref{section:sparse_subgraphs} we showed that the maximal $k$-edge-connected subgraphs can be computed in time $O(m + n\sqrt{n}\log n)$ (for constant $k$). 
Note that if we recompute from scratch the maximal $k$-edge-connected subgraphs after each update, \emph{max-$k$-edge} queries can be  answered in constant time.  This yields a simple algorithm that implements  \emph{insert} and \emph{delete} operations in time $O(m + n\sqrt{n}\log n)$ and  \emph{max-$k$-edge} queries in constant time. On the other side, one could simply do no extra work during \emph{insert} and \emph{delete} operations, but then answering a \emph{max-$k$-edge} query would require recomputing the maximal $k$-edge-connected subgraphs from scratch, yielding constant time per update and $\widetilde{O}(m+n\sqrt{n}\,)$ time per query. We next show how to implement  \emph{insert} and \emph{delete} operations in better $\widetilde{O}(n\sqrt{n}\,)$ time, while still keeping the running time for \emph{max-$k$-edge} queries constant. 
To achieve our improved bounds, we exploit the sparsification technique of Eppstein et al.~\cite{sparsification}. 
We start with the following definition.

\begin{definition}\label{def:certificate}
Let $k\geq 3$ be a fixed integer. Given an undirected graph $G=(V,E)$ with $m$ edges and $n$ vertices, a \emph{sparse certificate of maximal $k$-edge-connectivity} for $G$ is a graph $G'$ defined on the same vertex set as $G$ such that the following holds:
\vspace{-2mm}
\begin{mylist}{(ii)}
\litem{(i)} $G'$ has $O(n)$ edges;

\litem{(ii)} For any graph $H$, any two vertices are in the same maximal $k$-edge-connected subgraph of $G'\cup H$ if and only if they are in the same maximal $k$-edge-connected subgraph in $G\cup H$.
\end{mylist}
\end{definition}

The following lemma provides a sufficient condition for inferring that a $k$-certificate is a sparse certificate of maximal $k$-edge-connectivity.

\begin{lemma}
\label{lemma:k-certificate}
Let $G$ be an undirected graph, and let $C$ be a $k$-certificate of $G$ with $O(n)$ edges that contains all the $k$-interconnection edges of $G$. Then $C$ is a sparse certificate of maximal $k$-edge-connectivity for $G$.
\end{lemma}
\begin{proof}
Condition $(i)$ of Definition~\ref{def:certificate} is satisfied. It remains to establish condition $(ii)$. Since $C$ is a $k$-certificate of $G$, it is a subgraph of $G$ that has the same maximal $k$-edge-connected subgraphs as $G$. Let $Q$ be the quotient graph that is formed by shrinking every maximal $k$-edge-connected subgraph of $C$ into a single vertex. Then, since $C$ contains all the $k$-interconnection edges of $G$, by Property~\ref{property:T} we have that $Q$ has the same decomposition tree into maximal $k$-edge-connected subgraphs as $G$. Thus, the changes that this tree undergoes after inserting all the vertices and edges of $H$ into $G$, are the same as if inserting them to $C$. Thus, $C\cup H$ has the same maximal $k$-edge-connected subgraphs as $G\cup H$. This shows that condition $(ii)$ is satisfied too. Therefore $C$ is a sparse certificate of maximal $k$-edge-connectivity for $G$.
\end{proof}

\begin{corollary}\label{certificate}
Let $G=(V,E)$ be an undirected graph with $m$ edges and $n$ vertices. A sparse certificate of maximal $k$-edge-connectivity for $G$ can be computed in time $O(m\log^2{n})$.
\end{corollary}
\begin{proof}
In Section~\ref{section:sparse_subgraphs} we show that we can construct a $k$-certificate $C$ of $G$ with $O(n)$ edges in time $O(m\log^2{n})$. Furthermore, this $k$-certificate has the property that it contains all the $k$-interconnection edges of $G$. Thus, Lemma~\ref{lemma:k-certificate} implies that $C$ is a sparse certificate of maximal $k$-edge-connectivity for $G$.    
\end{proof}

We are now ready to apply the sparsfication framework of Eppstein et al.~\cite{sparsification}:

\begin{theorem}[\cite{sparsification}]\label{eppstein-sparsification}
Let $k\geq 3$ be a fixed integer, let  $f(n, m)$ be the time required to compute a sparse certificate of maximal $k$-edge-connectivity, and let $g(n, m)$ the time required to compute the maximal $k$-edge-connected subgraphs. Then we can build a fully dynamic data structure that can handle \emph{insert} and \emph{delete} operations in time $O( f(n, O(n)) + g(n, O(n)))$ and \emph{max-$k$-edge} queries in constant time.
\end{theorem}

From Corollary~\ref{certificate} we have $f(m,n) = O(m\log^2{n})$. From  \cite{ChechikHILP17} we have $g(m,n) = O(m\sqrt{n}\,)$ for $k\in\{3,4\}$, and $g(m,n) = O((m+n\log n)\sqrt{n}\,)$ for fixed $k>4$. (The improved time-bounds for $k\in\{3,4\}$ are derived by using either of the algorithms of \cite{DBLP:conf/esa/GeorgiadisIK21, DBLP:conf/esa/NadaraRSS21} for computing $3$-edge cuts in linear time.) By fitting those bounds into Theorem~\ref{eppstein-sparsification} we obtain:

\begin{theorem}
Let $G=(V,E)$ be an undirected graph with $n$ vertices.
Then we can build a fully dynamic data structure that can handle \emph{insert} and \emph{delete} operations in time $O(n\sqrt{n}\,)$, for $k\in\{3,4\}$, and in time $O(n\sqrt{n}\,\log n)$, for fixed $k>4$, so that it can answer \emph{max-$k$-edge} queries in constant time.
\end{theorem}

\section{Conclusion}\label{section:conclusion}

We presented two algorithms for maintaining a decomposition tree structure of the maximal $3$-edge-connected subgraphs of a graph. The first algorithm uses $O(n)$ space and can handle any sequence of $m$ edge insertions and $n$ vertex insertions in total time $O(n^2\log^2 n + m\alpha(m,n))$. The second algorithm uses $O(n^2)$ space and can handle any sequence of $m$ edge insertions and $n$ vertex insertions in total time $O(n^2\alpha(n,n) + m\alpha(m,n))$. 

We note that one can use this data structure to efficiently answer interspersed queries concerning the maximal $3$-edge-connected subgraphs, such as:
\begin{itemize}
\item Given vertices $x$ and $y$, report whether $x$ and $y$ belong to the same maximal $3$-edge-connected subgraph (using two calls to a DSU-$\mathit{find}$ operation)
\item Find the maximal $3$-edge-connected subgraph that contains $x$, in time analogous to its size (plus a call to a DSU-$\mathit{find}$ operation)
\item Report the size (the number of vertices or edges) of the maximal $3$-edge-connected subgraph that contains $x$ in constant time (plus a call to a DSU-$\mathit{find}$ operation) 
\item Return all maximal $3$-edge-connected subgraphs in time analogous to their size
\item Return the number of all maximal $3$-edge-connected subgraphs in constant time
\end{itemize}, etc.

The methods used in this algorithm extend previous work in maintaining the $2$- and $3$-edge-connected components \cite{galilMaintaining3EdgeConnectedComponents1993,lapoutreMaintainance23ecc,westbrookMaintainingBridgeConnectedBiconnected1992}, and may prove useful in solving other similar problems.
For instance, it seems possible that we can add more levels to the decomposition tree, and rely on previous work for maintaining the $4$- and $5$-edge-connected components, in order to maintain the maximal $4$- and $5$-edge-connected subgraphs (by properly adjusting the data structures and algorithms in \cite{Dinitz:5ECC,dinitzMaintainingClasses4EdgeConnectivity1998}). Furthermore, it seems that any improvement in maintaining the $k$-edge-connected components, for any constant $k$, would imply (under some conditions) that this solution can be plugged in to our framework for maintaining the decomposition tree, in order to maintain the maximal $k$-edge-connected subgraphs (with time-bounds analogous to those provided in our paper).
It seems possible that this framework could also be useful for maintaining the maximal $k$-vertex-connected subgraphs, by relying on efficient algorithms for maintaining the $k$-vertex-connected components (such as those described in \cite{triconnBatTam,lapoutreTriconn}, for the cases $k=3$).

%
%
We also showed that we can rely on results of Benczúr and Karger~\cite{ben-kar} in order to construct (almost) sparse certificates for them maximal $k$-edge-connected subgraphs. Our result implies that the difficulty in designing efficient algorithms for computing the maximal $k$-edge-connected subgraphs lies essentially in sparse graphs, and we can use it in order to speed up the running time of already known algorithms. In particular, by using the algorithm by Chechik et at. \cite{ChechikHILP17}, we get an $O(m+k^{O(k)}n\sqrt{n}\log n)$-time algorithm for computing the maximal $k$-edge-connected subgraphs in undirected graphs.

We believe that it is an interesting question whether a $k$-certificate of $O(n)$ size can be computed in linear time. This would be trivial if we had a linear-time algorithm for computing the maximal $k$-edge-connected subgraphs of a graph, but it is still an open problem whether this can be done for $k\geq 3$. Thus, we have to perform the construction of the certificates without explicitly computing the maximal $k$-edge-connected subgraphs, and this seems to be a challenging task.

\clearpage


\end{document}